\documentclass[
aip,
jcp,
reprint,
amsmath,amssymb,
longbibliography,
superscriptaddress
]{revtex4-2}

\usepackage[T1]{fontenc}
\usepackage[utf8]{inputenc}
\usepackage{lmodern}
\usepackage{microtype}
\usepackage{amsmath,amssymb,amsfonts,amsthm,mathtools}
\usepackage{bm}
\usepackage{mathrsfs}

\usepackage{graphicx}
\usepackage[caption=false]{subfig}

\usepackage{xcolor}
\usepackage{soul}

\usepackage{siunitx}
\usepackage{hyperref}
\usepackage[nameinlink,noabbrev]{cleveref}

\theoremstyle{remark}
\newtheorem{remark}{Remark}

\theoremstyle{plain}
\newtheorem{theorem}{Theorem}[section]
\newtheorem{lemma}[theorem]{Lemma}
\newtheorem{corollary}[theorem]{Corollary}
\newtheorem{proposition}[theorem]{Proposition}
\theoremstyle{definition}
\newtheorem{definition}[theorem]{Definition}

\DeclareMathAlphabet{\mathtensor}{OT1}{cmss}{bx}{n}
\providecommand{\Bb}{\mathbb}

\hypersetup{colorlinks=true, linkcolor=blue, citecolor=blue, urlcolor=blue}

\begin{document}

\title{Weak Solutions to the Bloch Equations with Distant Dipolar Field}

\author{Louis-S. Bouchard}
\email[Address correspondence to: ]{lsbouchard@ucla.edu}
\affiliation{Department of Chemistry and Biochemistry, University of California, Los Angeles, CA 90095}

\date{\today}

\begin{abstract}
The distant dipolar field (DDF) is a long-range, nonlocal contribution to spin dynamics in liquids that arises from intermolecular dipolar couplings and can generate multiple-quantum coherences in liquids and produce novel MRI contrast. Its nonlocal, sign-changing kernel makes Bloch--DDF dynamics strongly geometry dependent, and common FFT-based dipolar convolutions are naturally aligned with periodic boxes or padded Cartesian grids rather than bounded samples with reflective diffusion boundaries. We study the Bloch equations with the DDF on bounded domains with homogeneous Neumann diffusion conditions. We derive a conforming finite-element weak formulation that allows spatially varying diffusion and relaxation parameters and uses a short-distance regularization of the secular DDF kernel with length $a>0$. For fixed $a$ we prove boundedness of the induced DDF operator, establish an $L^2$ energy balance in which precession is neutral while diffusion and transverse relaxation are dissipative, and obtain local well-posedness with continuous dependence on the data (with global existence under energy-neutral transport conditions). For the Galerkin semi-discretization we show a discrete energy identity that mirrors the continuum estimate. For computation, we evaluate the DDF in real space with a matrix-free near/far scheme and advance in time with a second-order IMEX splitting method that treats diffusion and relaxation implicitly and treats precession explicitly. The explicit stage applies a Rodrigues rotation at DDF quadrature points followed by an $L^2$ projection, which supports stable multi-cycle lab-frame calculations. We validate against three closed-form benchmarks that isolate distinct model components and we quantify curved-boundary effects by comparing mapped-geometry finite elements with a voxel-mask finite-difference baseline on a spherical Neumann eigenmode decay. These results provide an analyzable and reproducible route for Bloch--DDF dynamics on bounded domains with complex geometry.
\end{abstract}


\maketitle

\section{Introduction}\label{sec:intro}

Intermolecular dipolar couplings can generate long-range, nonlocal contributions to nuclear spin dynamics in liquids and soft matter. In many experimentally relevant regimes, these effects can be expressed as a distant dipolar field (DDF) that depends on the magnetization distribution over the sample. This mechanism underlies several families of intermolecular multiple-quantum coherence (iMQC) and related sequences, including early demonstrations of coherence pathways that couple spins over mesoscopic distances and can produce contrast mechanisms that differ from conventional local Bloch models~\cite{bib:warrenscience93,bib:warrenjcp93,bib:warrenjcp96,bib:warrenscience96,bib:warrenscience98,bib:bowtelljmr92,bib:bowtellprl96,bib:bowtellimaging,bib:deville}. In these settings the DDF acts as a sign-changing interaction whose net effect depends on geometry, boundary conditions, and the spatial encoding applied by gradients.

This physics has been leveraged for structural and materials characterization and for MRI contrast mechanisms. Examples include reconstruction of porous microstructure from bulk signals that depend on the dipolar field~\cite{bib:bouchardjmr}, vector-field and correlation-length encoding strategies for imaging and characterization~\cite{bouchard2005multiple,chin2003isolating}, and sensitivity to internal field gradients and anisotropy in heterogeneous media such as trabecular bone~\cite{bib:bouchardbone}, together with related diffusion-in-internal-field methods for trabecular bone microstructure characterization~\cite{bib:sigmund2008bone}. Related work has explored coherent diffraction-like phenomena in engineered structures~\cite{tang2004observing} and nonlinear dynamics in strongly polarized fluids where long-range dipolar interactions can drive instabilities~\cite{bib:bouchardxenon}. These applications motivate simulations of Bloch--DDF dynamics beyond idealized periodic boxes.

From a computational standpoint, Bloch--DDF models are nonlinear and nonlocal. Direct evaluation of the dipolar integral couples all source and target points and can dominate cost. Efficient approaches on structured grids often use Fourier methods or hybrid real/Fourier strategies to accelerate dipolar field evaluation~\cite{bib:tilman}. Complementary numerical and perturbative studies on structured models have visualized geometry-dependent DDF patterns and analyzed susceptibility-sensitive CRAZED-type signals~\cite{bib:kirsch2009visualization,bib:wong2010theoretical}. Analytical treatments have also derived steady-state longitudinal profiles for repeated DDF sequences and revisited the validity regime of approximate Bloch--Torrey-based DDF signal formulas~\cite{bib:corum2004spatially,bib:barros2009revisited}. However, FFT-based approaches are naturally aligned with periodic or padded rectangular domains. They can be awkward for curved boundaries, reflective diffusion conditions, and complex geometries, where boundary effects are part of the physical model rather than a numerical artifact.

In this work we develop and analyze a finite-element (FE) weak formulation for the Bloch--DDF system on bounded domains with reflective diffusion boundaries first proposed in Refs.~\cite{bouchard2007finite,bouchard2005characterization}. The weak form reduces regularity requirements and supports complex geometries. We introduce a short-distance regularization of the secular DDF kernel with length scale $a>0$, which yields bounded operator estimates on bounded domains. On this basis we prove boundedness of the regularized DDF operator, derive an $L^2$ energy balance in which precession is neutral and diffusion and transverse relaxation are dissipative, and obtain existence and uniqueness of weak solutions under assumptions stated in the text. For the semi-discrete FE system we show a discrete energy identity that mirrors the continuum estimate.

Algorithmically, we apply the DDF in a matrix-free real-space near/far evaluation and use a second-order implicit--explicit (IMEX) time integrator. Diffusion and relaxation are treated implicitly, while precession (and, when present, advection) is treated explicitly. In the implementation, the explicit precession stage uses a structure-preserving update (Rodrigues rotation at DDF quadrature points followed by an $L^2$ projection), which enables stable multi-cycle lab-frame simulations.

Validation is challenging because general closed-form solutions are not available for nonlinear, nonlocal Bloch--DDF dynamics on bounded domains. We therefore validate against three closed-form benchmarks that isolate distinct model components: a uniform-mode DDF reduction in which the DDF enters as a deterministic kernel average, a periodic plane-wave eigenmode based on the Fourier symbol of the regularized kernel, and a purely longitudinal Neumann diffusion+$T_1$ eigenmode on a bounded domain. These benchmarks provide parameter-free checks of phase evolution, decay rates, and boundary-condition handling. We further quantify a geometry-driven advantage of the FE formulation on curved Neumann boundaries by comparing mapped-geometry FE against a voxel-mask finite-difference baseline on a spherical Neumann eigenmode decay for a boundary-sensitive mode.

The remainder of the paper recalls the Bloch--DDF model and boundary conditions, states the weak form, sets the functional framework and operator properties, presents the FE semi-discretization and its stability, describes the time integrator and implementation details, and provides validation studies and representative dynamics. Long-time lab-frame oscillations and envelope-level DDF effects are shown in \Cref{sec:validation} (see \Cref{fig:osc_long,fig:ddf_envelope,fig:gz200}).

\section{Bloch equations with distant dipolar field}\label{sec:model}

Let $\Omega\subset\mathbb{R}^3$ be a bounded sample domain. The magnetization is a vector field $\vec M(\mathbf r,t)$ defined at every point in time on $\Omega$. When flow is absent or when $\vec v\cdot\hat{\mathbf n}=0$ on $\partial\Omega$, it is natural to view $\vec M$ as confined to $\Omega$. If $\vec v\cdot\hat{\mathbf n}\neq 0$, then an inflow boundary condition is required on $\Gamma_{\mathrm{in}}=\{\mathbf r\in\partial\Omega:\vec v\cdot\hat{\mathbf n}<0\}$ and magnetization is transported across $\partial\Omega$.

Write $\mathbf R=\mathbf r-\mathbf r'$, $R=\lVert \mathbf R\rVert$, and $\hat{\mathbf R}=\mathbf R/R$. In high magnetic fields, the secular DDF can be written as
\begin{equation}
\vec B_d(\mathbf r)
= \int_{\Omega} \frac{1-3\cos^2\theta(\mathbf r,\mathbf r')}{2\,R^{3}}
\Bigl[\,3 M_z(\mathbf r')\,\hat{\mathbf z}-\vec M(\mathbf r')\Bigr]\,d^3\mathbf r',
\label{eq:secular_ddf}
\end{equation}
where $\cos\theta(\mathbf r,\mathbf r')=\hat{\mathbf z}\cdot\hat{\mathbf R}$. Without the secular approximation, the full expression is
\begin{equation}
\vec B(\mathbf r)
= \int_{\Omega} \frac{1}{R^{3}}
\left[\vec M(\mathbf r') - 3\bigl(\vec M(\mathbf r')\cdot\hat{\mathbf R}\bigr)\hat{\mathbf R}\right]\,d^3\mathbf r'.
\label{eq:exact_ddf}
\end{equation}
We absorb constant prefactors (such as $\mu_0/4\pi$) into units. Radiation damping or other offsets can be added as extra precession terms.

\subsection{Regularized secular kernel and optional diffusion-length filtering}

For analysis and for bounded-domain numerics it is convenient to regularize the short-distance singularity. We introduce a length scale $a>0$ and use the regularized secular kernel
\begin{align}
K_a(\mathbf r-\mathbf r')
&= \frac{1-3\cos^2\theta(\mathbf r,\mathbf r')}{2\,\bigl(R^2+a^2\bigr)^{3/2}},
\label{eq:Ka_def}
\\
\vec B_d(\mathbf r)
&= \int_{\Omega} K_a(\mathbf r-\mathbf r')\,
\mathrm{diag}(-1,-1,2)\,\vec M(\mathbf r')\,d^3\mathbf r' .
\label{eq:kernel_reg}
\end{align}
The factor $1/2$ is chosen so that in the formal limit $a\to 0$ one recovers the standard secular kernel in \eqref{eq:secular_ddf}.
At $\mathbf r=\mathbf r'$, the direction $\hat{\mathbf R}$ is undefined.
In the continuum integral this is immaterial (a measure-zero set), but in discrete quadrature and point-sum evaluations we enforce a consistent convention by omitting self-interactions (equivalently, setting the self-kernel contribution to zero).
In this work we develop the analysis for fixed $a>0$, for which the induced operator $\vec M\mapsto \vec B_d[\vec M]$ is bounded on the Sobolev spaces used below.
From a modeling perspective, $a$ can be interpreted as a coarse-graining length that removes sub-resolution contributions to the dipolar field.
In computations, $a$ also acts as a numerical softening scale that prevents near-field singular behavior when $\vec B_d$ is evaluated from discrete point or quadrature samples.
For the bounded-domain simulations reported here, $a$ is chosen as a fixed fraction of the domain length scale and is held fixed under the stated validation tests (unless otherwise stated). A strict excluded-volume model can alternatively be represented by a pair-correlation factor $g(R)$ with $g(R)=0$ for $R<a$ and $g(R)\to 1$ as $R\to\infty$, in which case one replaces $K_a$ by $K_a g$.

In many pulse sequences, diffusion during a characteristic time window $\tau$ suppresses DDF contributions from length scales below the diffusion length $\ell_D=\sqrt{D\tau}$. One can model this effect by replacing $\vec M$ in \eqref{eq:kernel_reg} by a filtered field $\vec M_\tau=\mathcal G_\tau[\vec M]$, where $\mathcal G_\tau$ denotes convolution with the Neumann heat kernel on $\Omega$ (or an equivalent FE heat-step filter). This filtering is an optional extension; the numerical experiments reported in this paper use the unfiltered regularized operator \eqref{eq:kernel_reg}.

\subsection{Bloch--DDF dynamics}

Including diffusion, relaxation, and optional flow, the Bloch--DDF dynamics are
\begin{align}
\frac{\partial \vec M}{\partial t}
&= \gamma\,\vec M \times \bigl(\vec B_d[\vec M] + \delta\vec B(\mathbf r)\bigr)
- \bigl(\vec v(\mathbf r)\cdot\nabla\bigr)\vec M \nonumber\\
&\quad + \nabla\cdot\bigl(D(\mathbf r)\nabla \vec M\bigr)
- \frac{M_x\,\hat{\mathbf x}+M_y\,\hat{\mathbf y}}{T_2(\mathbf r)}
+ \frac{M_0(\mathbf r)-M_z}{T_1(\mathbf r)}\,\hat{\mathbf z},
\label{eq:bloch_ddf}
\end{align}
where $D(\mathbf r)\ge 0$ is the (scalar) diffusion coefficient (a diffusion tensor can be used without changing the main ideas), $T_{1,2}(\mathbf r)$ may vary in space, $\delta\vec B(\mathbf r)$ is a prescribed static offset field, and $\vec v(\mathbf r)$ is a prescribed velocity field. The total effective field entering precession is
$\vec B_{\mathrm{eff}}[\vec M](\mathbf r)=\vec B_d[\vec M](\mathbf r)+\delta\vec B(\mathbf r)$, and $\gamma$ is the gyromagnetic ratio.
In the analysis we retain $\gamma$ explicitly.
In the numerical experiments and implementation we instead work in angular-frequency units by absorbing $\gamma$ into the fields, i.e., $\vec\Omega_d=\gamma\,\vec B_d$ and $\delta\vec\Omega=\gamma\,\delta\vec B$, so the precession term is written as $\vec M\times(\vec\Omega_d+\delta\vec\Omega)$.
Accordingly, when we specify a uniform offset $\omega$ or gradient $g_z$, we state whether it is given in field units or in angular-frequency units; in the analytical benchmarks below, $\omega$ denotes a field offset so the corresponding angular frequency is $\gamma\,\omega$.

We write the transport term in advective form,
$\partial_t \vec M + (\vec v\cdot\nabla)\vec M=\cdots$,
equivalently, \eqref{eq:bloch_ddf} with the advection term moved to the
left-hand side. The system is supplemented by reflective diffusion boundary
conditions in the no-flux form
$\hat{\mathbf n}\cdot\bigl(D(\mathbf r)\nabla \vec M\bigr)=0$ on
$\partial\Omega$, understood componentwise. For scalar diffusion with
$D(\mathbf r)>0$ near the boundary, this reduces to
$\hat{\mathbf n}\cdot\nabla \vec M=0$. If advection is included and
$\vec v\cdot\hat{\mathbf n}\neq 0$ on $\partial\Omega$, an inflow boundary
condition for $\vec M$ is additionally required on
$\Gamma_{\mathrm{in}}$. In the numerical experiments considered here,
$\vec v\equiv 0$, so no advection term appears in the computed cases.

These equations are the basis for the weak formulation in \S\ref{sec:weakform}.

\subsection{Weak solutions}\label{sec:weakform}

Let $\Omega\subset\mathbb{R}^3$ be a bounded Lipschitz domain with boundary $\partial\Omega$ and outward unit normal $\hat{\mathbf n}$. The magnetization is a vector field $\vec M(\mathbf r,t)$ defined on $\Omega$. We include static offsets $\delta\vec B(\mathbf r)$ and the DDF $\vec B_d[\vec M]$. We allow spatial variation in $T_{1,2}(\mathbf r)$ and $D(\mathbf r)$. For reflective diffusion boundaries we impose the homogeneous no-flux condition
\begin{equation}
\hat{\mathbf n}\cdot\bigl(D(\mathbf r)\nabla \vec M\bigr) = 0
\quad \text{on } \partial\Omega,
\label{eq:neumann_bc}
\end{equation}
understood componentwise.
For scalar $D(\mathbf r)$ with $D(\mathbf r)>0$ near $\partial\Omega$, this is equivalent to $\hat{\mathbf n}\cdot\nabla \vec M=0$. If advection is included and $\vec v\cdot\hat{\mathbf n}\neq 0$ on $\partial\Omega$, an inflow boundary condition on $\Gamma_{\mathrm{in}}=\{\mathbf r\in\partial\Omega:\vec v\cdot\hat{\mathbf n}<0\}$ is also required; for simplicity, the analysis below emphasizes the case $\vec v\cdot\hat{\mathbf n}=0$ when flow is present.

In component form ($i=1,2,3$), with $\vec B\equiv\vec B_d[\vec M]$ and
$\epsilon_{ijk}$ denoting the Levi--Civita symbol, the Bloch--DDF system
reads
\begin{equation}
\begin{aligned}
\frac{\partial M_i}{\partial t}
&= \gamma\,\epsilon_{ijk}\,M_j\bigl(B_k+\delta B_k\bigr)
- \frac{\delta_{i1} M_1+\delta_{i2} M_2}{T_2(\mathbf r)} \\
&\quad + \delta_{i3}\,\frac{M_0(\mathbf r)-M_3}{T_1(\mathbf r)} \\
&\quad + \nabla\cdot\bigl(D(\mathbf r)\nabla M_i\bigr)
- \bigl(\vec v(\mathbf r)\cdot\nabla\bigr) M_i,
\end{aligned}
\label{eq:bloch_componform}
\end{equation}
for $\mathbf r\in\Omega$. The advection sign convention in
\eqref{eq:bloch_componform} is therefore consistent with the advective form
$\partial_t M_i + (\vec v\cdot\nabla)M_i=\cdots$. If
$\nabla\cdot\vec v=0$, this is equivalent to the conservative form
$\partial_t M_i + \nabla\cdot(\vec v\,M_i)=\cdots$.

Let $V:=H^1(\Omega)$ and take any test function $v_i\in V$. Multiply \eqref{eq:bloch_componform} by $v_i$ and integrate over $\Omega$:
\begin{equation}
\begin{aligned}
\int_\Omega \frac{\partial M_i}{\partial t}\,v_i\,\mathrm{d}^3\mathbf r
&= \gamma\sum_{j,k=1}^3 \epsilon_{ijk}
   \int_\Omega M_j(B_k+\delta B_k)\,v_i\,\mathrm{d}^3\mathbf r \\
&\quad - \int_\Omega \frac{\delta_{i1} M_1+\delta_{i2} M_2}{T_2(\mathbf r)}\,v_i\,\mathrm{d}^3\mathbf r \\
&\quad - \int_\Omega \frac{\delta_{i3} M_3}{T_1(\mathbf r)}\,v_i\,\mathrm{d}^3\mathbf r
+ \int_\Omega \nabla\cdot\bigl(D(\mathbf r)\nabla M_i\bigr)\,v_i\,\mathrm{d}^3\mathbf r \\
&\quad - \int_\Omega \bigl(\vec v(\mathbf r)\cdot\nabla M_i\bigr)\,v_i\,\mathrm{d}^3\mathbf r
+ \int_\Omega \frac{M_0(\mathbf r)}{T_1(\mathbf r)}\,\delta_{i3}\,v_i\,\mathrm{d}^3\mathbf r .
\end{aligned}
\label{eq:weak_step1}
\end{equation}
Integrate the diffusion term by parts:
\begin{multline}
\int_\Omega \nabla\cdot\bigl(D(\mathbf r)\nabla M_i\bigr)\,v_i\,\mathrm{d}^3\mathbf r
= - \int_\Omega D(\mathbf r)\,\nabla M_i\cdot\nabla v_i\,\mathrm{d}^3\mathbf r \\
+ \int_{\partial\Omega} v_i\,\hat{\mathbf n}\cdot\bigl(D(\mathbf r)\nabla M_i\bigr)\,\mathrm{d}S .
\label{eq:diffusion_ibp}
\end{multline}
Under the homogeneous Neumann condition \eqref{eq:neumann_bc}, the boundary term vanishes; for scalar $D(\mathbf r)$, $\hat{\mathbf n}\cdot\nabla M_i=0$ implies $\hat{\mathbf n}\cdot(D\nabla M_i)=0$.

For advection we denote by $a_{\mathrm{adv}}(M_i,v_i)$ the bilinear form
used in the weak formulation. In the energy-neutral case, assume
$\nabla\cdot\vec v=0$ in $\Omega$ and
$\vec v\cdot\hat{\mathbf n}=0$ on $\partial\Omega$, and define
\begin{equation}
a_{\mathrm{adv}}(M_i,v_i)
:= \frac12\int_\Omega (\vec v\cdot\nabla M_i)\,v_i\,\mathrm{d}^3\mathbf r
   - \frac12\int_\Omega (\vec v\cdot\nabla v_i)\,M_i\,\mathrm{d}^3\mathbf r .
\label{eq:adv_skew}
\end{equation}
With this choice, $a_{\mathrm{adv}}(M_i,M_i)=0$ for each component. If these
conditions do not hold, one may instead use the standard advective bilinear
form
\begin{equation}
a_{\mathrm{adv}}(M_i,v_i)
:=
\int_\Omega (\vec v\cdot\nabla M_i)\,v_i\,\mathrm{d}^3\mathbf r,
\label{eq:adv_standard}
\end{equation}
together with the appropriate inflow/outflow boundary terms.

\begin{definition}[Weak form of the Bloch--DDF system]
With reflective diffusion boundaries, the weak form is: for each component $i$ and for all $v_i\in H^1(\Omega)$,
\begin{equation}
\begin{aligned}
\int_\Omega \frac{\partial M_i}{\partial t}\,v_i\,\mathrm{d}^3\mathbf r
&= \gamma\sum_{j,k=1}^3 \epsilon_{ijk}
   \int_\Omega M_j\bigl(B_k+\delta B_k\bigr)\,v_i\,\mathrm{d}^3\mathbf r \\
&\quad - \int_\Omega
   \frac{\delta_{i1} M_1+\delta_{i2} M_2}{T_2(\mathbf r)}\,v_i\,
   \mathrm{d}^3\mathbf r \\
&\quad - \int_\Omega
   \frac{\delta_{i3} M_3}{T_1(\mathbf r)}\,v_i\,
   \mathrm{d}^3\mathbf r \\
&\quad - \int_\Omega
   D(\mathbf r)\,\nabla M_i\cdot\nabla v_i\,
   \mathrm{d}^3\mathbf r \\
&\quad - a_{\mathrm{adv}}(M_i,v_i) \\
&\quad + \int_\Omega
   \frac{M_0(\mathbf r)}{T_1(\mathbf r)}\,\delta_{i3}\,v_i\,
   \mathrm{d}^3\mathbf r .
\end{aligned}
\label{eq:weak_bloch}
\end{equation}
where $\vec B=\vec B_d[\vec M]$. When the skew form
\eqref{eq:adv_skew} is used, we assume
$\nabla\cdot\vec v=0$ in $\Omega$ and
$\vec v\cdot\hat{\mathbf n}=0$ on $\partial\Omega$. If the standard
advective form \eqref{eq:adv_standard} is used instead, then the
corresponding inflow/outflow boundary terms must be included.
\end{definition}

For the numerical approximation we use a conforming Galerkin method.

\begin{definition}[Galerkin FE formulation]
Let $V_h\subset H^1(\Omega)$ be a finite-dimensional space of scalar shape functions with $\dim V_h=N_h$. Seek
$u_i^h(\mathbf r,t)=\sum_{n=1}^{N_h} w_{in}(t)\,\varphi_n(\mathbf r)\in V_h$ ($i=1,2,3$) such that, for all $v_i=\varphi_m\in V_h$,
\begin{equation}
\begin{aligned}
\int_\Omega \frac{\partial u_i^h}{\partial t}\,v_i\,\mathrm{d}^3\mathbf r
&= \gamma\sum_{j,k=1}^3 \epsilon_{ijk}
   \int_\Omega
   u_j^h\bigl(B_k[u^h]+\delta B_k\bigr)\,v_i\,
   \mathrm{d}^3\mathbf r \\
&\quad - \int_\Omega
   \frac{\delta_{i1} u_1^h+\delta_{i2} u_2^h}{T_2(\mathbf r)}\,v_i\,
   \mathrm{d}^3\mathbf r \\
&\quad - \int_\Omega
   \frac{\delta_{i3} u_3^h}{T_1(\mathbf r)}\,v_i\,
   \mathrm{d}^3\mathbf r \\
&\quad - \int_\Omega
   D(\mathbf r)\,\nabla u_i^h\cdot\nabla v_i\,
   \mathrm{d}^3\mathbf r \\
&\quad - a_{\mathrm{adv}}(u_i^h,v_i) \\
&\quad + \int_\Omega
   \frac{M_0(\mathbf r)}{T_1(\mathbf r)}\,\delta_{i3}\,v_i\,
   \mathrm{d}^3\mathbf r .
\end{aligned}
\label{eq:galerkin_form}
\end{equation}
Here $B_k[u^h]$ is the $k$th component of the DDF evaluated from $u^h$ via the chosen DDF kernel (e.g., \eqref{eq:secular_ddf} or the regularized form \eqref{eq:kernel_reg}).
\end{definition}

\noindent\textbf{Remarks.}
The Galerkin method enforces orthogonality of the residual to the test space $V_h$ in the sense of \eqref{eq:galerkin_form}. Under standard regularity assumptions, $u^h$ converges to the weak solution as the mesh is refined~\cite{bib:quarteroni}. The diffusion term is well-defined for $H^1(\Omega)$ fields. For advection, one may use a conservative or skew-symmetric form (for example, \eqref{eq:adv_skew} when admissible) and add stabilization (such as SUPG) in advection-dominated regimes without changing the structure of \eqref{eq:galerkin_form}.

Choose a scalar FE basis $\{\varphi_n\}_{n=1}^{N_h}\subset H^1(\Omega)$ and expand each component as
\begin{equation}
\begin{aligned}
u_i^h(\mathbf r,t)
&= \sum_{n=1}^{N_h} w_{in}(t)\,\varphi_n(\mathbf r),
\quad i\in\{1,2,3\},
\end{aligned}
\label{eq:FE_expansion}
\end{equation}
where $w_{in}(t)$ are the unknown coefficients. Testing \eqref{eq:galerkin_form} with $v_i=\varphi_m$ defines the standard FE operators
\begin{align}
M_{mn}
&= \int_\Omega \varphi_n(\mathbf r)\,\varphi_m(\mathbf r)\,\mathrm{d}^3\mathbf r,
\label{eq:M_def}\\
K_{mn}
&= \int_\Omega D(\mathbf r)\,
        \nabla\varphi_n(\mathbf r)\cdot\nabla\varphi_m(\mathbf r)\,\mathrm{d}^3\mathbf r,
\label{eq:K_def}\\
(S_{1})_{mn}
&= \int_\Omega \frac{1}{T_1(\mathbf r)}\,
        \varphi_n(\mathbf r)\,\varphi_m(\mathbf r)\,\mathrm{d}^3\mathbf r,
\label{eq:S1_def}\\
(S_{2})_{mn}
&= \int_\Omega \frac{1}{T_2(\mathbf r)}\,
        \varphi_n(\mathbf r)\,\varphi_m(\mathbf r)\,\mathrm{d}^3\mathbf r,
\label{eq:S2_def}\\
(b_{T_1})_m
&= \int_\Omega \frac{M_0(\mathbf r)}{T_1(\mathbf r)}\,
        \varphi_m(\mathbf r)\,\mathrm{d}^3\mathbf r.
\label{eq:b_def}
\end{align}
For advection, we distinguish the non-skew and skew-symmetric discrete operators. The non-skew operator associated with the strong form $-(\vec v\cdot\nabla)M_i$ is
\begin{equation}
(N_v)_{mn}
= - \int_\Omega \bigl(\vec v(\mathbf r)\cdot\nabla\varphi_n(\mathbf r)\bigr)\,
        \varphi_m(\mathbf r)\,\mathrm{d}^3\mathbf r .
\label{eq:Nv_def}
\end{equation}
When $\nabla\cdot\vec v=0$ and $\vec v\cdot\hat{\mathbf n}=0$ on $\partial\Omega$, an energy-neutral alternative is the skew form
\begin{equation}
(N_v^{\mathrm{skew}})_{mn}
= -\frac12\int_\Omega \bigl(\vec v\cdot\nabla\varphi_n\bigr)\,\varphi_m\,\mathrm{d}^3\mathbf r
+\frac12\int_\Omega \bigl(\vec v\cdot\nabla\varphi_m\bigr)\,\varphi_n\,\mathrm{d}^3\mathbf r .
\label{eq:Nv_skew_def}
\end{equation}
In the numerical experiments reported in this paper we take $\vec v\equiv 0$, so $N_v=0$ and advection is omitted from the computed cases; we include $N_v$ here to state the general weak form and to support extensions.

Let $w_i=(w_{i1},\dots,w_{iN_h})^\top$. The semi-discrete system (with the mass matrix retained) reads
\begin{equation}
\begin{aligned}
M\,\dot w_i
&= \gamma\,\mathcal P_i(w)
- \delta_{i1}\,S_2\,w_1
- \delta_{i2}\,S_2\,w_2 \\
&\quad - \delta_{i3}\,S_1\,w_3
- K\,w_i
+ N_v\,w_i
+ \delta_{i3}\,b_{T_1},
\quad i=1,2,3 .
\end{aligned}
\label{eq:semi_discrete}
\end{equation}
where $\mathcal P_i(w)$ is the discrete precession contribution induced by the DDF and any static offset field $\delta\vec B$. In the energy and stability results we will assume either $N_v=N_v^{\mathrm{skew}}$ (when admissible) or else we will explicitly account for the symmetric part of $N_v$.

Given coefficient vectors $w$, we proceed in three steps. First, evaluate the FE fields $M_j^h(\mathbf r)=\sum_n w_{jn}\varphi_n(\mathbf r)$ at the chosen DDF quadrature points $\{\mathbf r_q\}_{q=1}^{N_q}$. Second, compute the DDF field $\vec B_d(\mathbf r_q)$ from $\vec M^h$ using the chosen DDF kernel (e.g., \eqref{eq:secular_ddf} or the regularized form \eqref{eq:kernel_reg}). In the reference implementation, $\vec B_d$ is evaluated in real space by either a direct all-pairs sum over source points for small problem sizes or validation runs, or a Barnes--Hut octree approximation for larger problems. In the Barnes--Hut mode, source points are grouped in an octree; a target point accepts a node's aggregated contribution when the opening criterion $s/d\le\theta$ is satisfied, where $s$ is the node size and $d$ is the distance from the target to the node center.
When a node is accepted, we approximate its contribution using the node's aggregated magnetization (monopole) together with the full regularized secular kernel $K_a(\mathbf R)$ evaluated with $\mathbf R$ taken as the vector from the target point to the node center (so the angular factor $1-3\cos^2\theta$ is retained).
The diagonal factor $\mathrm{diag}(-1,-1,2)$ is applied as in \eqref{eq:kernel_reg}.
When the opening criterion fails, we fall back to direct summation over the node's children, and ultimately over leaf contents.
The user-controlled parameters are the opening angle $\theta$ and the leaf size (maximum points per leaf).
In practice, we validate these choices by comparing the tree-based field $\vec B_d^{\mathrm{tree}}(\mathbf r_q)$ against the direct all-pairs field $\vec B_d^{\mathrm{direct}}(\mathbf r_q)$ on representative small meshes and choose $(\theta,\text{leaf size})$ so that the resulting relative field error remains below the time-discretization error in the reported runs. Third, assemble, for each $i$ and each test index $m$,
\begin{equation}
\bigl[\mathcal P_i(w)\bigr]_m
= \sum_{j,k=1}^3 \epsilon_{ijk}\int_\Omega
   M_j^h(\mathbf r)\,B_k(\mathbf r)\,\varphi_m(\mathbf r)\,\mathrm{d}^3\mathbf r .
\label{eq:P_matrix_free}
\end{equation}
In the implementation, the integral in \eqref{eq:P_matrix_free} is evaluated by the same quadrature rule used to define the DDF quadrature points, i.e.,
\[
\bigl[\mathcal P_i(w)\bigr]_m \approx \sum_{j,k=1}^3 \epsilon_{ijk}\sum_{q=1}^{N_q} w_q\,
M_j^h(\mathbf r_q)\,B_k(\mathbf r_q)\,\varphi_m(\mathbf r_q).
\]
If one prefers a fully assembled representation, define
\begin{align}
R_{k\,nm}
&= \int_\Omega \delta B_k(\mathbf r)\,
    \varphi_n(\mathbf r)\,\varphi_m(\mathbf r)\,
    \mathrm{d}^3\mathbf r,
\label{eq:R_def}\\
T_{k\,nmq}
&= a_k \int_\Omega\!\int_\Omega
    \varphi_n(\mathbf r)\,\varphi_m(\mathbf r)\,
    K_{\mathrm{DDF}}(\mathbf r,\mathbf r')\,
    \varphi_q(\mathbf r')\,
    \mathrm{d}^3\mathbf r'\,\mathrm{d}^3\mathbf r,
\label{eq:T_def}
\end{align}
with $a_1=a_2=-1$ and $a_3=2$. Here $K_{\mathrm{DDF}}(\mathbf r,\mathbf r')$ denotes the scalar DDF kernel used in the model; for the regularized secular choice it is $K_a(\mathbf r-\mathbf r')$ from \eqref{eq:Ka_def}, and for the unregularized secular kernel it is $(1-3\cos^2\theta(\mathbf r,\mathbf r'))/(2\lVert \mathbf r-\mathbf r'\rVert^3)$ interpreted in the principal-value sense.

Then
\begin{equation}
\begin{aligned}
\bigl[\mathcal P_i(w)\bigr]_m
&= \sum_{j,k=1}^3 \epsilon_{ijk}
   \Biggl[
      \sum_{n,q=1}^{N_h}
      w_{jn}\,w_{kq}\,T_{k\,nmq} \\
&\quad\quad\quad
      + \sum_{n=1}^{N_h}
      w_{jn}\,R_{k\,nm}
   \Biggr] .
\end{aligned}
\label{eq:P_tensor_form}
\end{equation}
In practice we do not assemble $T_{k\,nmq}$; the matrix-free approach above controls memory and cost and is the implementation used in \Cref{sec:discrete_stability} and \Cref{sec:time}.

\subsection{Time integration}\label{sec:time}

We advance \eqref{eq:semi_discrete} with a second-order IMEX Strang splitting method. Physically, diffusion and relaxation are the stiff, dissipative processes in the Bloch--Torrey dynamics, so they are advanced implicitly to avoid a severe stability restriction. By contrast, precession is non-dissipative and acts locally as a rotation of the magnetization, which makes an explicit rotation-based update natural for this part of the dynamics. Diffusion and relaxation are treated implicitly, while precession is treated explicitly. For the precession stage, the reference implementation uses a structure-preserving update that rotates the magnetization at DDF quadrature points with a Rodrigues map and then projects back to nodal coefficients by an $L^2$ projection (mass solve). This choice is important for stable multi-cycle lab-frame simulations such as those in \Cref{fig:osc_long}. The $L^2$ projection does not form $M^{-1}$ explicitly; it only requires a
sparse solve with the SPD mass matrix.

Let $w=(w_1,w_2,w_3)$, define the block operators
\begin{align*}
\mathcal M=& \mathrm{blkdiag}(M,M,M),\\
\mathcal K=& \mathrm{blkdiag}(K,K,K),\\
\mathcal S=& \mathrm{blkdiag}(S_2,S_2,S_1),
\end{align*}
and the source vector $f=(0,0,b_{T_1})$.
One time step from $t^n$ to $t^{n+1}=t^n+\Delta t$ proceeds as follows.
\begin{equation}
\begin{aligned}
\text{(i) }&
\bigl(\mathcal M+\tfrac{\Delta t}{4}(\mathcal K+\mathcal S)\bigr)\,w^{(a)}
= \bigl(\mathcal M-\tfrac{\Delta t}{4}(\mathcal K+\mathcal S)\bigr)\,w^{n}
  + \frac{\Delta t}{2}\,f .
\end{aligned}
\label{eq:time_half_impl}
\end{equation}
To describe the explicit precession stage, let $\{\mathbf r_q\}_{q=1}^{N_q}\subset\Omega$ be the DDF quadrature points with weights $\{w_q\}_{q=1}^{N_q}$. Given coefficients $w$, define the quadrature-point magnetization $\vec M_q=\vec M_h(\mathbf r_q)$ and the corresponding effective angular-frequency field
\[
\vec\Omega_q(w)=\gamma\Bigl(\delta\vec B(\mathbf r_q)+\vec B_d[\vec M_h](\mathbf r_q)\Bigr),
\]
evaluated using the same quadrature rules as in \eqref{eq:P_matrix_free}. Let $\mathcal R_{\tau}(\cdot;\vec\Omega)$ denote the Rodrigues rotation that maps $\vec m\in\mathbb{R}^3$ to the solution at time $\tau$ of $\dot{\vec m}=\vec m\times \vec\Omega$ with constant angular-frequency field $\vec\Omega\in\mathbb R^3$. Let $\Pi_h$ denote the componentwise $L^2$ projection from quadrature-point values back to FE coefficients: given values $\{u_q\}_{q=1}^{N_q}$, the projected coefficients $u_h=\sum_{n=1}^{N_h} c_n\varphi_n$ satisfy $M c = p$ with $p_m=\sum_{q=1}^{N_q} w_q\,u_q\,\varphi_m(\mathbf r_q)$.
\begin{equation}
\begin{aligned}
\text{(ii) }&
\vec M_q^{(a)}=\vec M_h^{(a)}(\mathbf r_q),\quad
\vec\Omega_q^{(a)}=\vec\Omega_q\!\bigl(w^{(a)}\bigr),\\
&\tilde{\vec M}_q=\mathcal R_{\Delta t/2}\!\bigl(\vec M_q^{(a)};\vec\Omega_q^{(a)}\bigr),\quad
\tilde w=\Pi_h(\tilde{\vec M}) .
\end{aligned}
\label{eq:time_mid_pred}
\end{equation}

\begin{equation}
\begin{aligned}
\text{(iii) }&
\vec\Omega_q^{(1/2)}=\vec\Omega_q(\tilde w),
\\
&\vec M_q^{(b)}=\mathcal R_{\Delta t}\!\bigl(\vec M_q^{(a)};\vec\Omega_q^{(1/2)}\bigr),\quad
w^{(b)}=\Pi_h(\vec M^{(b)}) .
\end{aligned}
\label{eq:time_full_expl}
\end{equation}
If advection is included, it is advanced explicitly with the same midpoint predictor, using $\tilde w$ as the midpoint state, and applied as an additional update to $w^{(b)}$. In the numerical experiments reported here we take $\vec v\equiv 0$, so the explicit stage consists only of \eqref{eq:time_mid_pred}--\eqref{eq:time_full_expl}. In the implementation, a spatially uniform term in $\delta\vec B$ (for example, $\omega\,\hat{\mathbf z}$) may be applied by an exact $\hat{\mathbf z}$-axis rotation before and after the DDF-driven rotation; this is equivalent to including it in the Rodrigues map above after converting to an angular frequency field $\vec\Omega=\gamma(\vec B_d+\delta\vec B)$. Accordingly, the angular-frequency quantities entering the Rodrigues map are $\gamma\,\omega$ and $\gamma\,g_z$ when $\omega$ and $g_z$ are specified in field units.
\begin{equation}
\begin{aligned}
\text{(iv) }&
\Bigl(\mathcal M+\tfrac{\Delta t}{4}(\mathcal K+\mathcal S)\Bigr)\,w^{n+1} \\
&\quad = \Bigl(\mathcal M-\tfrac{\Delta t}{4}(\mathcal K+\mathcal S)\Bigr)\,w^{(b)}
+ \tfrac{\Delta t}{2}\,f .
\end{aligned}
\label{eq:time_second_half_impl}
\end{equation}
Steps \eqref{eq:time_half_impl} and \eqref{eq:time_second_half_impl} are linear solves with a symmetric positive definite left-hand side when $D(\mathbf r)\ge 0$ and $T_{1,2}(\mathbf r)>0$. In the reference implementation we solve these systems either by sparse direct factorization (reused across time steps when coefficients are constant) or by conjugate gradients with a preconditioner (for example, AMG or incomplete factorization). The explicit stage \eqref{eq:time_mid_pred}--\eqref{eq:time_full_expl} evaluates the precession load by first computing the DDF field at quadrature points and then applying the pointwise Rodrigues map followed by an $L^2$ projection; this uses the same quadrature rules as the matrix-free construction \eqref{eq:P_matrix_free}. In physical terms, this explicit substep isolates the rotational part of the dynamics: during this stage the magnetization is rotated by the local effective field rather than damped. When advection is included, its action is also applied explicitly using the midpoint state $\tilde w$.

The scheme is second order in time for sufficiently smooth solutions. Its
local truncation error is $O(\Delta t^{3})$, and its global error is
$O(\Delta t^{2})$. The implicit Crank--Nicolson substeps are unconditionally
stable for the diffusion and relaxation terms. Accordingly, the time-step
restriction is not set by diffusion or relaxation, because those stiff
dissipative processes are handled implicitly. Instead, it is set by the
explicit part of the algorithm, namely the rate at which the effective
precession field, and any explicitly treated transport, can change the
magnetization over one time step. Thus, any time-step restriction is driven
by the explicit stage and depends on bounds for the effective precession rate
and, if advection is included, on bounds for the transport rate. In the
numerical experiments reported here we take $\vec v\equiv 0$, so the explicit
restriction is determined only by precession. In particular, a sufficient
condition is
\[
\Delta t \le c\,\Bigl(
|\gamma|\,\|\delta\vec B\|_{L^\infty(\Omega)}
+
|\gamma|\,\|\mathcal T_a\|_{L^2\to L^2}\,\sup_n \|w^n\|_M
\Bigr)^{-1},
\]
where $c>0$ is independent of the mesh size and
$\|w\|_M^2=\sum_{i=1}^3 w_i^\top M w_i$. If advection is included, an
additional contribution proportional to
$\|\vec v\|_{W^{1,\infty}(\Omega)}$ enters the sufficient condition.
This makes the role of the IMEX splitting transparent: the implicit part
removes the fast diffusive and relaxational stability barrier, while the
explicit restriction tracks only the physically nondissipative precessional
rotation and any explicitly retained transport. The stability results in
\Cref{sec:time_stability} make the dependence on precession and transport
bounds precise. They also allow either the skew advection operator
\eqref{eq:Nv_skew_def}, which is energy-neutral when admissible, or a
controlled contribution from the symmetric part of $N_v$.

\medskip
\noindent\emph{Remarks.} (a) If advection is included and
$\nabla\cdot\vec v=0$ and $\vec v\cdot\hat{\mathbf n}=0$ on $\partial\Omega$,
then the skew form $N_v^{\mathrm{skew}}$ in \eqref{eq:Nv_skew_def} is
energy-neutral and reduces artificial energy growth. (b) Mass lumping is not
used by default because it perturbs the $M$-inner-product identities used in
the stability analysis; it may be enabled for fully explicit variants when a
different stability argument is adopted. (c) If a far-field DDF approximation
is used, its tolerance should be coupled to $\Delta t$ so that the induced
defect in discrete skew-symmetry does not dominate the time discretization
error. (d) The Rodrigues map preserves $|\vec M|$ pointwise at the DDF
quadrature points during the explicit precession substep for fixed
$\vec\Omega_q$. The subsequent $L^2$ projection $\Pi_h$ preserves the
discrete $L^2$ structure used in the energy estimates, but it does not in
general preserve the pointwise magnetization magnitude at FE nodes.

\begin{figure*}[t]
\centering
\subfloat[\label{fig:analytic_uniform_re}]{\includegraphics[width=0.32\textwidth]{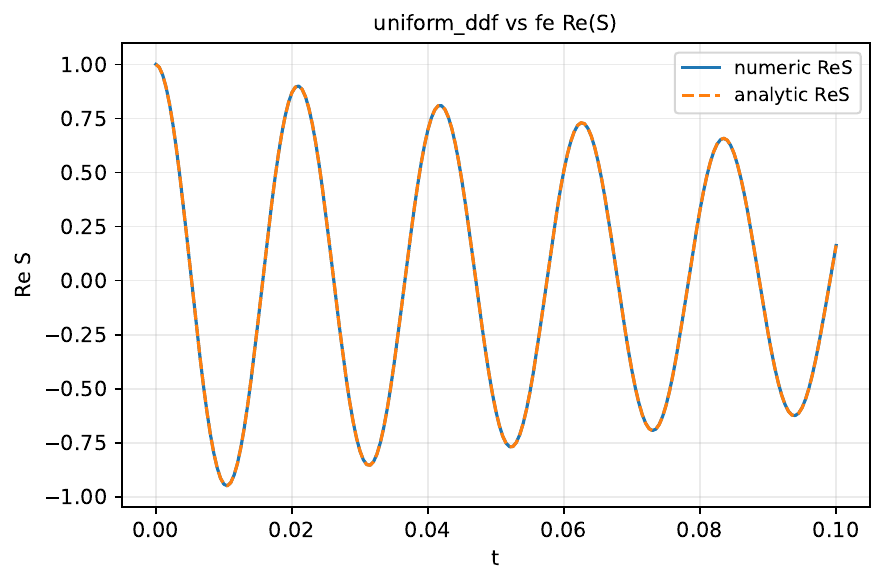}}
\hfill
\subfloat[\label{fig:analytic_uniform_im}]{\includegraphics[width=0.32\textwidth]{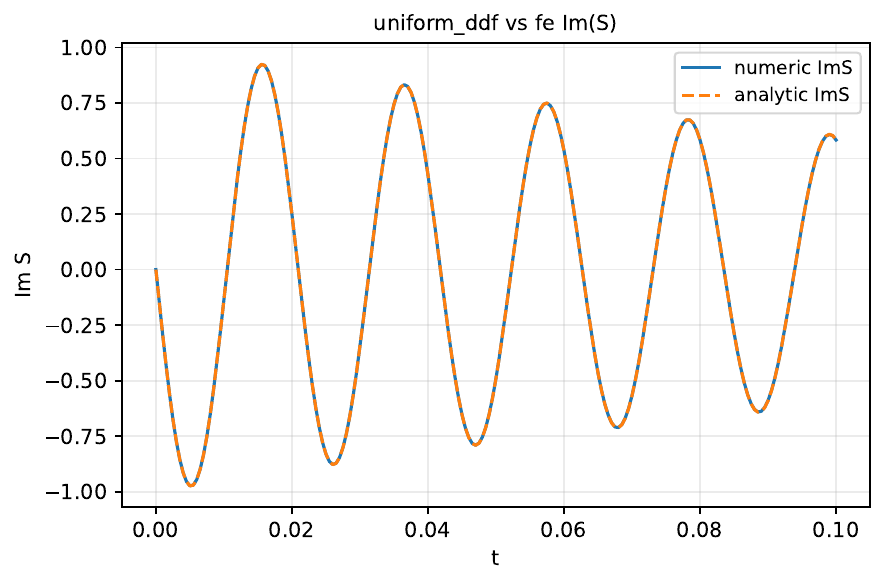}}
\hfill
\subfloat[\label{fig:analytic_uniform_abs}]{\includegraphics[width=0.32\textwidth]{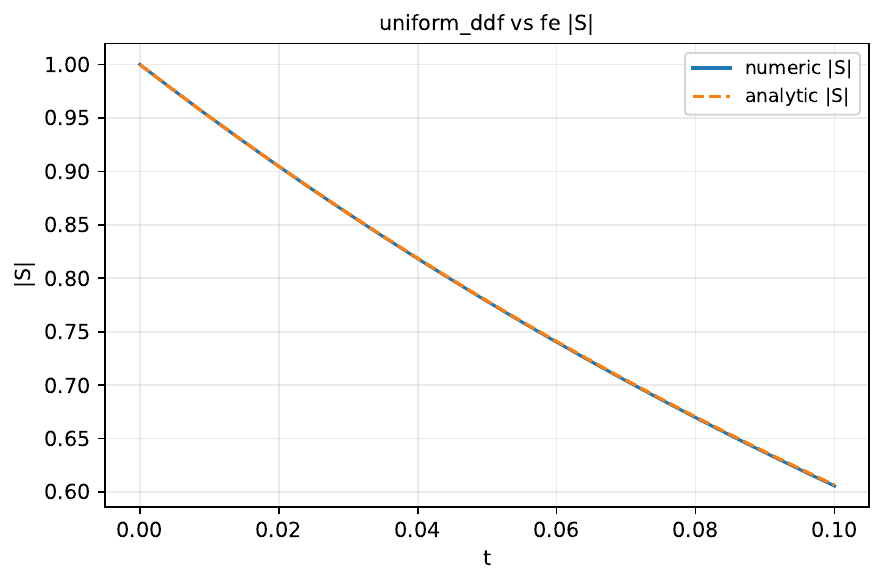}}
\caption{
Uniform-mode analytical benchmark.
Numerical $S(t)=\langle M_x+iM_y\rangle$ (from the FE run), where $\langle\cdot\rangle$ denotes the spatial average over $\Omega$ of the corresponding FE field, compared against the closed-form solution \eqref{eq:uniform_S_sol} with $\kappa_{\mathrm{eff}}$ computed directly from the regularized kernel and the stated quadrature rule.
Panels show $\mathrm{Re}\,S(t)$, $\mathrm{Im}\,S(t)$, and $|S(t)|$.
}
\label{fig:analytic_uniform}
\end{figure*}

\section{Functional setting and assumptions}\label{sec:assumptions}

Let $\Omega\subset\Bb R^3$ be a bounded Lipschitz domain with outward unit normal $\hat{\mathbf n}$. Vector fields are treated componentwise in the standard Sobolev spaces $L^2(\Omega)$ and $H^1(\Omega)$.

Assume $D\in L^\infty(\Omega)$ with $D(\mathbf r)\ge D_{\min}>0$ almost everywhere. (A symmetric positive definite diffusion tensor $\mathbf D(\mathbf r)\in L^\infty(\Omega)^{3\times 3}$ with $\mathbf D(\mathbf r)\xi\cdot\xi\ge D_{\min}\lVert\xi\rVert^2$ can be used without substantive changes.)
Assume $T_{1,2}\in L^\infty(\Omega)$ with $0<T_{\min}\le T_{1,2}(\mathbf r)\le T_{\max}<\infty$.
Let $\delta\vec B\in L^\infty(\Omega)^3$ and $M_0\in L^\infty(\Omega)$.
If advection is included, assume $\vec v\in W^{1,\infty}(\Omega)^3$. For the energy-neutral transport setting used in several stability statements below, we additionally assume
\begin{equation}
\nabla\cdot\vec v=0 \quad \text{in } \Omega,
\quad
\vec v\cdot\hat{\mathbf n}=0 \quad \text{on } \partial\Omega.
\label{eq:transport_assumptions}
\end{equation}
If $\vec v\cdot\hat{\mathbf n}\neq 0$ on $\partial\Omega$, then an inflow boundary condition for $\vec M$ on $\Gamma_{\mathrm{in}}=\{\mathbf r\in\partial\Omega:\vec v\cdot\hat{\mathbf n}<0\}$ is required and additional boundary terms appear in the energy estimates.

We use the regularized secular kernel with a short-distance length scale $a>0$,
\begin{equation}
K_a(\mathbf r)=\frac{1-3\cos^2\theta}{2\bigl(\lVert \mathbf r\rVert^2+a^2\bigr)^{3/2}},
\quad
\mathcal A=\mathrm{diag}(-1,-1,2),
\label{eq:Ka_operator_def}
\end{equation}
and define the linear map $\mathcal T_a:(L^2(\Omega))^3\to (L^2(\Omega))^3$ by
\begin{equation}
\mathcal T_a[\vec M](\mathbf r)=\int_{\Omega} K_a(\mathbf r-\mathbf r')\,\mathcal A\,\vec M(\mathbf r')\,d^3\mathbf r'.
\label{eq:Ta_def}
\end{equation}
The distant dipolar field is $\vec B_d=\mathcal T_a[\vec M]$.

The angular mean of the factor $(1-3\cos^2\theta)$ over a full sphere is zero,
\begin{equation}
\int_{S^2} \bigl(1-3\cos^2\theta\bigr)\,d\Omega=0.
\label{eq:angular_mean_zero}
\end{equation}
This fact helps interpret the DDF as a long-range, sign-changing interaction. On bounded domains, however, neighborhoods near $\partial\Omega$ are not spherically symmetric, so boundary effects can bias local contributions. In this work we keep the bounded-domain integral \eqref{eq:Ta_def} and treat boundary effects as part of the physical model. We develop the analysis for fixed $a>0$; the singular principal-value kernel at $a=0$ requires additional arguments and is not pursued here.

An optional pair-correlation factor $g(\lVert \mathbf r-\mathbf r'\rVert)\in L^\infty$ with $g=0$ near $0$ and $g\to 1$ at large separation may be included. All statements below hold with $K_a$ replaced by $K_a g$.

We impose reflective diffusion boundaries,
\begin{equation}
\hat{\mathbf n}\cdot\bigl(D(\mathbf r)\nabla \vec M\bigr)=0 \quad \text{on } \partial\Omega,
\label{eq:neumann_bc_assump}
\end{equation}
understood componentwise. For advection we use either the skew-symmetric discretization \eqref{eq:Nv_skew_def} under \eqref{eq:transport_assumptions}, or else we retain the standard variational form and explicitly bound the contribution from the symmetric part of the discrete advection operator.

\section{Properties of the DDF operator and an energy identity}\label{sec:ddf_props}

\stepcounter{theorem}
\noindent The regularized DDF operator remains bounded on the Lebesgue and Sobolev spaces used below, so the coarse-grained field does not generate additional singular growth. A detailed statement and proof are given in the Appendix. See \Cref{prop:Ta_bdd}.
\stepcounter{theorem}
\noindent The pointwise orthogonality identity used in the energy estimate is stated and proved in 
\Cref{lem:neutral} in the Appendix.
\stepcounter{theorem}
\noindent The local Lipschitz bound for the precession nonlinearity is stated and proved in 
\Cref{lem:lipschitz} in the Appendix.
\stepcounter{theorem}
\noindent The continuum energy calculation shows that precession is energy-neutral, while diffusion and relaxation are dissipative; transport contributes only through compression and boundary-flux terms. A detailed statement and proof are given in the Appendix. See \Cref{prop:energy}.

\section{Well-posedness of the weak problem}\label{sec:wellposed}

Let $V:=(H^1(\Omega))^3$, $H:=(L^2(\Omega))^3$, and $V':=(H^{-1}(\Omega))^3$. We use the Gelfand triple $V\hookrightarrow H\cong H'\hookrightarrow V'$.
Define
\[
X:=L^2(0,T;V)\cap H^1(0,T;V').
\]
A weak solution on $[0,T]$ is a function $\vec M\in X$ that satisfies \eqref{eq:weak_bloch} for all test functions $v_i\in H^1(\Omega)$ (with the chosen DDF kernel and with any required advection boundary conditions) and the initial condition $\vec M(\cdot,0)=\vec M_{\mathrm{init}}\in H$.

\stepcounter{theorem}
\noindent The time-continuity result used in the Gelfand-triple argument is stated and proved in 
\Cref{lem:LM} in the Appendix.
\stepcounter{theorem}
\noindent The local existence and uniqueness result is stated and proved in 
\Cref{thm:local}  in the Appendix. 
\stepcounter{theorem}
\noindent Weak solutions depend continuously on the initial data, so perturbations in the starting magnetization remain controlled on the local existence interval. A detailed statement and proof are given in the Appendix.  See \Cref{prop:contdep}.
\stepcounter{theorem}
\noindent When transport is absent or is imposed in an energy-neutral form, the a priori bounds extend the local weak solution to any finite time horizon. A detailed statement and proof are given in the Appendix. See \Cref{cor:global}.

\noindent\textbf{Remarks.}
(i) A uniformly elliptic diffusion tensor $\mathbf D(\mathbf r)$ can replace $D(\mathbf r)$ with minor notational changes. (ii) The analysis is developed for fixed $a>0$, for which $\mathcal T_a$ is bounded on the spaces used. A principal-value treatment of the singular kernel at $a=0$ requires additional kernel analysis and is not pursued here.

\section{Discrete stability of the FE semi-discretization}\label{sec:discrete_stability}

Let $M,K,S_1,S_2$ be as in \eqref{eq:M_def}--\eqref{eq:S2_def} and let $b_{T_1}$ be as in \eqref{eq:b_def}. For advection, let $N_v$ denote either the standard operator \eqref{eq:Nv_def} or the skew-symmetric operator $N_v^{\mathrm{skew}}$ in \eqref{eq:Nv_skew_def} when \eqref{eq:transport_assumptions} holds. Let $\mathcal P_i(w)$ be assembled with a quadrature that is exact on the products appearing below, or sufficiently accurate so that any quadrature defect is of higher order in the mesh size $h$. If the DDF is evaluated approximately (e.g., by a compressed far-field), define the discrete skew-symmetry defect
\[
\delta_{\mathrm{skew}}(w):=\sum_{i=1}^3 w_i^\top \mathcal P_i(w).
\]
We assume the approximation tolerance is chosen so that $|\delta_{\mathrm{skew}}(w)|$ remains small compared to the time-discretization error over the time window of interest.

\stepcounter{theorem}
\noindent The discrete skew-symmetry identity for the precession operator is stated and proved in 
\Cref{lem:disc_skew} in the Appendix.
\stepcounter{theorem}
\noindent The matrix positivity properties used in the discrete energy analysis are stated and proved in 
\Cref{lem:spd}  in the Appendix. 
\stepcounter{theorem}
\noindent The semi-discrete energy inequality and its proof are given in  
\Cref{thm:semi_disc_energy}  in the Appendix.

\begin{remark}[Non-energy-neutral advection]
If \eqref{eq:transport_assumptions} does not hold or a non-skew discretization of advection is used, an extra term
\[
\frac12\,w^\top\bigl(N_v+N_v^\top\bigr)w
\]
appears on the right-hand side of \eqref{eq:semi_disc_energy_id}. This term can be bounded by $C\|w\|_M^2$ with $C$ depending on $\|\nabla\cdot\vec v\|_{L^\infty(\Omega)}$ and any boundary flux contributions, yielding a Gr\"onwall-type inequality.
\end{remark}

\section{Time discretization: stability and consistency}\label{sec:time_stability}

We analyze the second-order IMEX splitting scheme described in \S\ref{sec:time}; see \eqref{eq:time_half_impl}--\eqref{eq:time_second_half_impl}. Diffusion and relaxation are treated implicitly by Crank--Nicolson. The explicit stage uses a midpoint evaluation of the effective field. For analysis it is convenient to express this explicit stage in an algebraic midpoint form. The reference implementation realizes the same midpoint-field evaluation by a structure-preserving Rodrigues rotation at DDF quadrature points followed by an $L^2$ projection (mass solve). The stability bound stated in \Cref{thm:imex_stability} applies provided the explicit stage satisfies the discrete skew-symmetry property in \Cref{lem:disc_skew} and the local Lipschitz bounds stated in the Appendix proof; any quadrature or projection defects enter as higher-order perturbations and are controlled in the same way.

\stepcounter{theorem}
The nonlinear IMEX stability result and its proof are given in 
\Cref{thm:imex_stability} in the Appendix.
\stepcounter{theorem}
\noindent The IMEX method remains second order in time for smooth solutions, and the FE spatial error retains the standard conforming approximation rates. A detailed statement and proof are given in the Appendix. See \Cref{prop:consistency}.

\noindent\textbf{Remarks.}
(i) Diffusion and relaxation are unconditionally stable in the Crank--Nicolson substeps; the time-step restriction is driven by explicit precession and advection. (ii) If $\vec B_d$ is evaluated approximately (near/far splitting with a compressed far field), the energy inequality acquires a defect term proportional to the approximation tolerance; this tolerance should be chosen so that the defect is smaller than the desired time-discretization error. (iii) If \eqref{eq:transport_assumptions} does not hold, then additional transport contributions enter both the continuous and discrete energy balances via $\nabla\cdot\vec v$ and boundary fluxes, and the stability estimate requires corresponding bounds.

\section{Validation and representative dynamics}\label{sec:validation}

We report several validation and demonstration cases for the bounded-domain Bloch--DDF solver.
All bounded-domain runs use the real-space DDF evaluator and the structure-preserving FE precession update (Rodrigues rotation at DDF quadrature points followed by an $L^2$ projection).
The periodic plane-wave benchmark uses FFT-based periodic convolution on a uniform grid, as described in \Cref{sec:analytic_planewave}.

\begin{figure*}[t]
\centering
\subfloat[\label{fig:analytic_planewave_re}]{\includegraphics[width=0.48\textwidth]{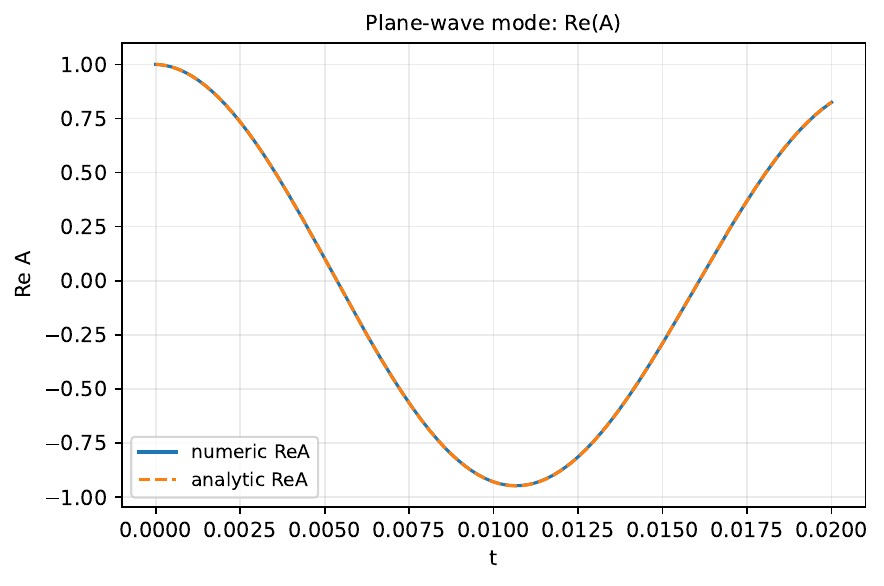}}
\hfill
\subfloat[\label{fig:analytic_planewave_abs}]{\includegraphics[width=0.48\textwidth]{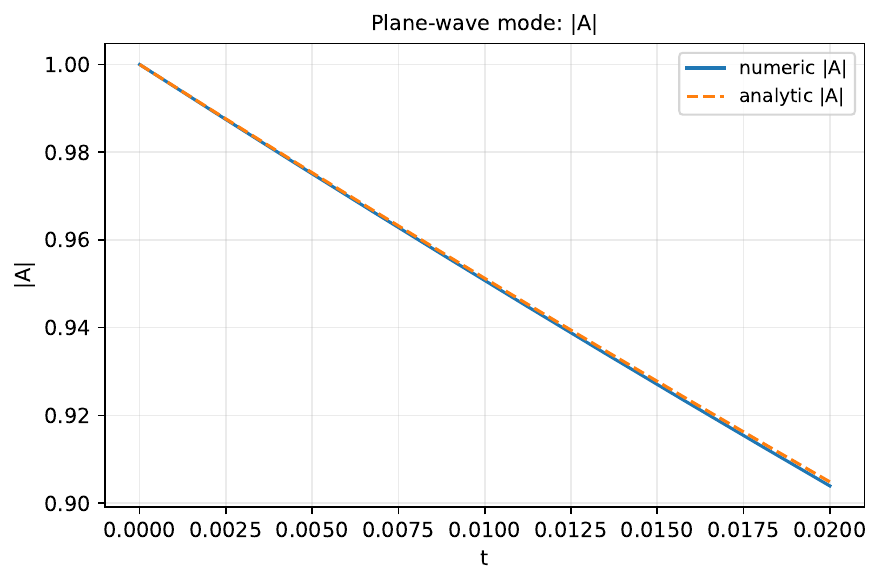}}
\caption{
Periodic plane-wave analytical benchmark for a single Fourier mode.
Numerical evolution of the mode amplitude $A(t)$ compared against the closed-form solution \eqref{eq:planewave_A_sol} using the kernel symbol value $\lambda=-1.3889$ (regularization length $a=0.04$, mode $(m_x,m_y,m_z)=(1,0,0)$, and periodic box $\Omega=[0,1)^3$ discretized on a $32\times 32\times 32$ uniform grid).
Panels show $\mathrm{Re}\,A(t)$ and $|A(t)|$.
The plotted case corresponds to the smallest time step in \Cref{tab:planewave_dt}.
}
\label{fig:analytic_planewave}
\end{figure*}

\subsection{Comparison against closed-form analytical benchmarks}\label{sec:validation_analytic}

This section compares the numerical solver against three analytical benchmarks with closed-form time dependence. Each benchmark isolates a specific subset of the Bloch--DDF dynamics.
No parameters are fitted. All constants that appear in the closed-form expressions (e.g., kernel averages or Fourier symbols) are determined directly from the chosen regularized DDF kernel and the stated geometry and discretization.  A summary of the three analytical benchmarks and the observed relative errors is given in Table~\ref{tab:analytic_summary}.

Given a complex time series $y(t_n)$ sampled at the stored times $\{t_n\}_{n=0}^N$, we report the discrete relative $L^2$ error
\begin{equation}
\epsilon_{\mathrm{rel}}
=
\frac{\bigl(\sum_{n=0}^N |y_{\mathrm{num}}(t_n)-y_{\mathrm{an}}(t_n)|^2\bigr)^{1/2}}
     {\bigl(\sum_{n=0}^N |y_{\mathrm{an}}(t_n)|^2\bigr)^{1/2}}.
\label{eq:analytic_relL2}
\end{equation}
For the longitudinal-mode test, $y(t)=c(t)$ is a real scalar coefficient.

\begin{table}[t]
\caption{Summary of analytical benchmarks and observed agreement.
The reported $\epsilon_{\mathrm{rel}}$ values are computed from the recorded numerical and analytical time series using \eqref{eq:analytic_relL2}.}
\begin{ruledtabular}
\begin{tabular}{lcc}
Benchmark & Observable & $\epsilon_{\mathrm{rel}}$ \\
\parbox[t]{0.48\columnwidth}{\raggedright
Uniform-mode DDF reduction (\Cref{sec:analytic_uniform})}
& $S(t)=\langle M_x+iM_y\rangle$
& $6.89\times 10^{-4}$ \\
\parbox[t]{0.48\columnwidth}{\raggedright
Periodic plane-wave mode (\Cref{sec:analytic_planewave})}
& $A(t)$
& $5.73\times 10^{-4}$ \\
\parbox[t]{0.48\columnwidth}{\raggedright
Longitudinal diffusion+$T_1$ mode (\Cref{sec:analytic_longitudinal})}
& $c(t)$
& $1.40\times 10^{-5}$ \\
\end{tabular}
\end{ruledtabular}
\label{tab:analytic_summary}
\end{table}

\subsubsection{Uniform transverse mode with regularized DDF}\label{sec:analytic_uniform}

Assume (i) constant coefficients, (ii) no gradients ($g_z=0$), (iii) no flow, and (iv) an initially uniform magnetization.
In this setting, diffusion does not act on the uniform mode.
We model the DDF contribution to the uniform-mode dynamics by a spatially uniform effective field of the form
\begin{equation}
\vec B_d(t) = \kappa_{\mathrm{eff}}\,\mathcal A\,\vec M(t),
\quad
\mathcal A=\mathrm{diag}(-1,-1,2),
\label{eq:uniform_Bd}
\end{equation}
where $\kappa_{\mathrm{eff}}=k\,\kappa(\Omega,a)$ is a geometry- and regularization-dependent constant, and $k$ is the dimensionless DDF coupling parameter used in the numerical model. A convenient definition is
\begin{equation}
\kappa(\Omega,a)
=
\frac{1}{|\Omega|}\int_{\Omega}\left(\int_{\Omega} K_a(\mathbf r-\mathbf r')\,d^3\mathbf r'\right)d^3\mathbf r,
\label{eq:kappa_def}
\end{equation}
with $K_a$ from \eqref{eq:Ka_def}. In the implementation used here, $\kappa$ is evaluated directly using the same DDF quadrature rule specified for the run (no fitting).

Let $\delta\vec B=\omega\,\hat{\mathbf z}$, where $\omega$ is a constant field offset (so the corresponding angular frequency is $\gamma\,\omega$), and define the transverse complex signal $S(t)=M_x(t)+iM_y(t)$.
Using \eqref{eq:uniform_Bd}, the transverse dynamics reduce to
\begin{equation}
\frac{dS}{dt}
=
-\frac{1}{T_2}\,S
-i\Bigl(\gamma\,\omega+3\gamma\,\kappa_{\mathrm{eff}}\,M_z(t)\Bigr)S,
\label{eq:uniform_S_ode}
\end{equation}
while the longitudinal component satisfies the standard $T_1$ recovery
\begin{equation}
\frac{dM_z}{dt}=\frac{M_0-M_z}{T_1},
\quad
M_z(0)=M_{z0}.
\label{eq:uniform_Mz_ode}
\end{equation}
The solution of \eqref{eq:uniform_Mz_ode} is
\begin{equation}
M_z(t)=M_0+\bigl(M_{z0}-M_0\bigr)e^{-t/T_1}.
\label{eq:uniform_Mz_sol}
\end{equation}
Substituting \eqref{eq:uniform_Mz_sol} into \eqref{eq:uniform_S_ode} yields the closed-form transverse signal
\begin{multline}
S(t)
=
S(0)\,
e^{-t/T_2}\,
\exp\bigl(-i\,\gamma\,\omega t\bigr) \\
\times
\exp\Bigl(
-i\,3\gamma\,\kappa_{\mathrm{eff}}
\bigl(M_0 t + (M_{z0}-M_0)T_1(1-e^{-t/T_1})\bigr)
\Bigr).
\label{eq:uniform_S_sol}
\end{multline}
For the discretization and parameters used in \Cref{fig:analytic_uniform}, the DDF constant $\kappa$ was computed deterministically from \eqref{eq:kappa_def} using the same quadrature and kernel discretization as in the numerical DDF evaluation (no fitting). For this case, $\kappa=7.8125$ and $\kappa_{\mathrm{eff}}=39.0625$ (with $k=5$, regularization length $a=0.04$, and a $10\times 10\times 10$ hexahedral discretization of the unit box). \Cref{fig:analytic_uniform} compares the numerical $S(t)$ against \eqref{eq:uniform_S_sol} with the same $(\omega,T_1,T_2,M_0,M_{z0})$ and with $\kappa_{\mathrm{eff}}$ determined by \eqref{eq:kappa_def}. The observed relative error is $\epsilon_{\mathrm{rel}}=6.89\times 10^{-4}$ over the reported time samples.

\begin{figure*}[t]
\centering
\subfloat[(a) $\mathrm{Re}\,S(t)$\label{fig:osc_long_re}]%
{\includegraphics[width=0.32\textwidth]{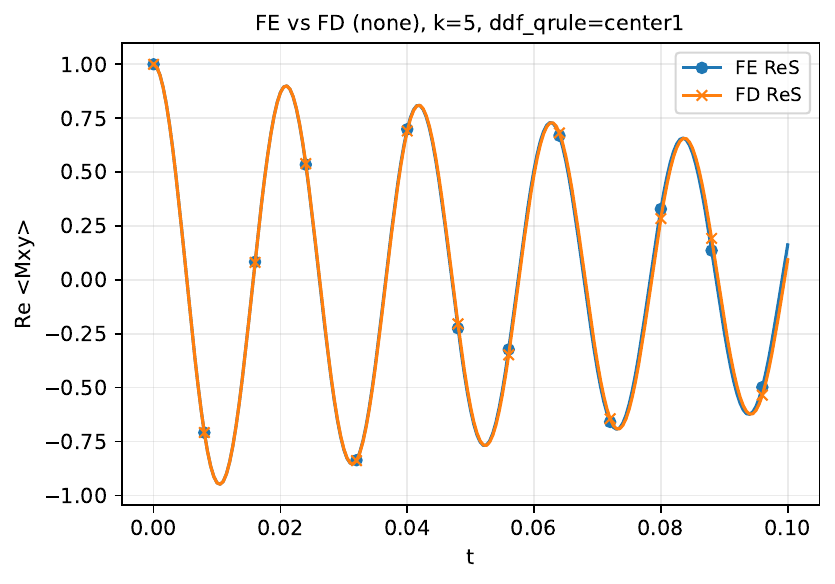}}
\hfill
\subfloat[(b) $\mathrm{Im}\,S(t)$\label{fig:osc_long_im}]%
{\includegraphics[width=0.32\textwidth]{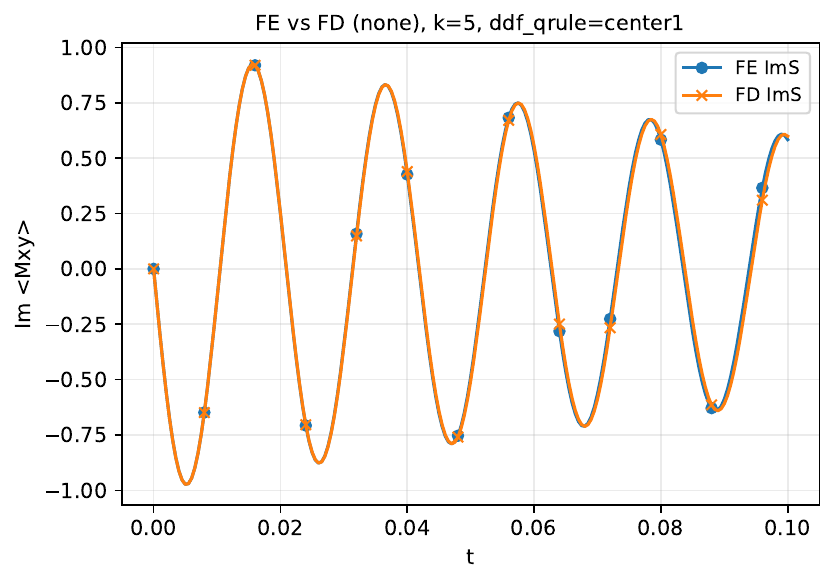}}
\hfill
\subfloat[(c) $|S(t)|$\label{fig:osc_long_abs}]%
{\includegraphics[width=0.32\textwidth]{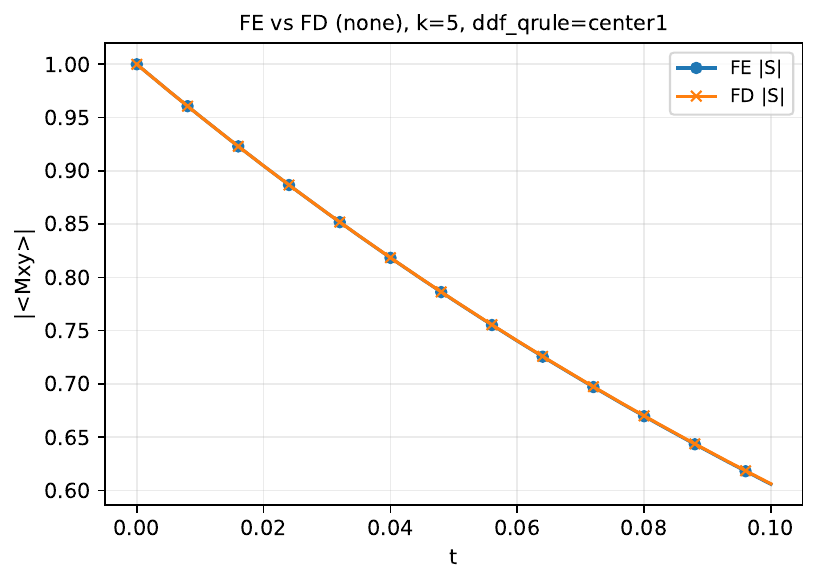}}
\caption{
Long-time lab-frame evolution of the global transverse signal
$S(t)=\langle M_x+iM_y\rangle$ with DDF enabled.
Panels show (a) $\mathrm{Re}\,S(t)$,
(b) $\mathrm{Im}\,S(t)$, and
(c) the envelope $|S(t)|$ for the same run.
}
\label{fig:osc_long}
\end{figure*}

\subsubsection{Periodic plane-wave eigenmode}\label{sec:analytic_planewave}

This benchmark targets the DDF symbol and the phase evolution of a single Fourier mode in a periodic setting.
Consider a periodic box $\Omega=[0,L_x)\times[0,L_y)\times[0,L_z)$, and let the transverse magnetization be a single mode
$M_x+iM_y=A(t)e^{i\mathbf q\cdot\mathbf r}$ with $\mathbf q=(2\pi m_x/L_x,2\pi m_y/L_y,2\pi m_z/L_z)$.
For periodic convolution, the DDF operator is diagonal in Fourier space, so
\begin{equation}
\mathcal T_a\!\left[e^{i\mathbf q\cdot\mathbf r}\right]
=
\lambda(\mathbf q)\,e^{i\mathbf q\cdot\mathbf r},
\label{eq:planewave_eig}
\end{equation}
where $\lambda(\mathbf q)$ is the Fourier symbol of the chosen regularized kernel (including the discretization weights in the discrete setting). This periodic benchmark uses FFT-based convolution on a uniform grid; it is not intended as a bounded-domain reference, since FFT-based DDF evaluation on nonperiodic domains requires padding or extensions that introduce boundary artifacts.
With $M_z$ held constant in the linear test, the complex amplitude satisfies
\begin{equation}
\frac{dA}{dt}
=
-\Bigl(D|\mathbf q|^2+\frac{1}{T_2}\Bigr)A
-i\,\gamma\Bigl(\omega+k\,\lambda(\mathbf q)\Bigr)A,
\label{eq:planewave_A_ode}
\end{equation}
and hence
\begin{equation}
A(t)=A(0)\exp\Bigl(-\bigl(D|\mathbf q|^2+T_2^{-1}\bigr)t\Bigr)\exp\Bigl(-i\,\gamma(\omega+k\lambda(\mathbf q))t\Bigr).
\label{eq:planewave_A_sol}
\end{equation}
For the case $(m_x,m_y,m_z)=(1,0,0)$ with $L_x=L_y=L_z=1$ and $a=0.04$ (a fixed coarse-graining/softening length used throughout unless otherwise stated), the computed symbol value was
$\lambda=-1.3889$.
With $\omega=300$, $k=5$, and $\gamma=1$, this gives $\omega_{\mathrm{eff}}=\omega+k\lambda=293.0553$.
With $D=0$ and $T_2=0.2$, the decay rate is $T_2^{-1}=5$.
\Cref{fig:analytic_planewave} shows the numerical amplitude (computed by FFT-based periodic convolution and Rodrigues precession) against \eqref{eq:planewave_A_sol}. Observed time-step dependence for this periodic plane-wave benchmark is summarized in Table~\ref{tab:planewave_dt}.

\begin{table}[t]
\caption{Plane-wave benchmark: observed time-step dependence for $\epsilon_{\mathrm{rel}}$ in $A(t)$. This benchmark implementation uses a first-order operator splitting (exact decay followed by a precession update), and the measured error decreases linearly with $\Delta t$. This benchmark is used to validate the kernel symbol and the associated phase and decay rates in a controlled periodic setting, rather than to demonstrate the temporal order of the full Bloch--DDF time integrator.
}
\begin{ruledtabular}
\begin{tabular}{ccc}
$\Delta t$ & $\epsilon_{\mathrm{rel}}$ & Observed order \\
\hline
$2.0\times 10^{-4}$ & $2.292\times 10^{-3}$ & --- \\
$1.0\times 10^{-4}$ & $1.146\times 10^{-3}$ & $\approx 1.00$ \\
$5.0\times 10^{-5}$ & $5.728\times 10^{-4}$ & $\approx 1.00$ \\
\end{tabular}
\end{ruledtabular}
\label{tab:planewave_dt}
\end{table}

\subsubsection{Longitudinal diffusion plus $T_1$ relaxation eigenmode}\label{sec:analytic_longitudinal}

This benchmark targets the diffusion operator and $T_1$ recovery under reflective (Neumann) boundaries in a rectangular box.
Let $\Omega=(0,L_x)\times(0,L_y)\times(0,L_z)$ and consider a Neumann Laplacian eigenfunction
\begin{equation}
\phi(\mathbf r)=\cos\Bigl(\frac{m_x\pi x}{L_x}\Bigr)\cos\Bigl(\frac{m_y\pi y}{L_y}\Bigr)\cos\Bigl(\frac{m_z\pi z}{L_z}\Bigr),
\label{eq:neumann_phi}
\end{equation}
which satisfies $-\Delta\phi=\lambda\phi$ with
\begin{equation}
\lambda
=
\pi^2\Bigl(\frac{m_x^2}{L_x^2}+\frac{m_y^2}{L_y^2}+\frac{m_z^2}{L_z^2}\Bigr).
\label{eq:neumann_lambda}
\end{equation}
Assume a purely longitudinal perturbation of the form
\begin{equation}
M_z(\mathbf r,t)=M_0 + c(t)\phi(\mathbf r),
\quad
M_x=M_y=0,
\label{eq:Mz_ansatz}
\end{equation}
with constant $D$ and $T_1$ and reflective diffusion boundaries.
Substituting \eqref{eq:Mz_ansatz} into the $z$ component of \eqref{eq:bloch_ddf} (with no precession terms active) yields
\begin{equation}
\frac{dc}{dt}
=
-\Bigl(D\lambda+\frac{1}{T_1}\Bigr)c,
\label{eq:c_ode}
\end{equation}
so
\begin{equation}
c(t)=c(0)\exp\Bigl(-\bigl(D\lambda+T_1^{-1}\bigr)t\Bigr).
\label{eq:c_sol}
\end{equation}
For the unit box with $(m_x,m_y,m_z)=(1,1,1)$, $\lambda=3\pi^2=29.6088$.
With $D=10^{-3}$ and $T_1=5$, the predicted decay rate is $D\lambda+T_1^{-1}=0.2296$.
The measured agreement is $\epsilon_{\mathrm{rel}}=1.40\times 10^{-5}$ for the coefficient time series.
\Cref{fig:analytic_longitudinal} shows the numerical and analytical $c(t)$.

\begin{figure}[t]
\centering
\includegraphics[width=\columnwidth]{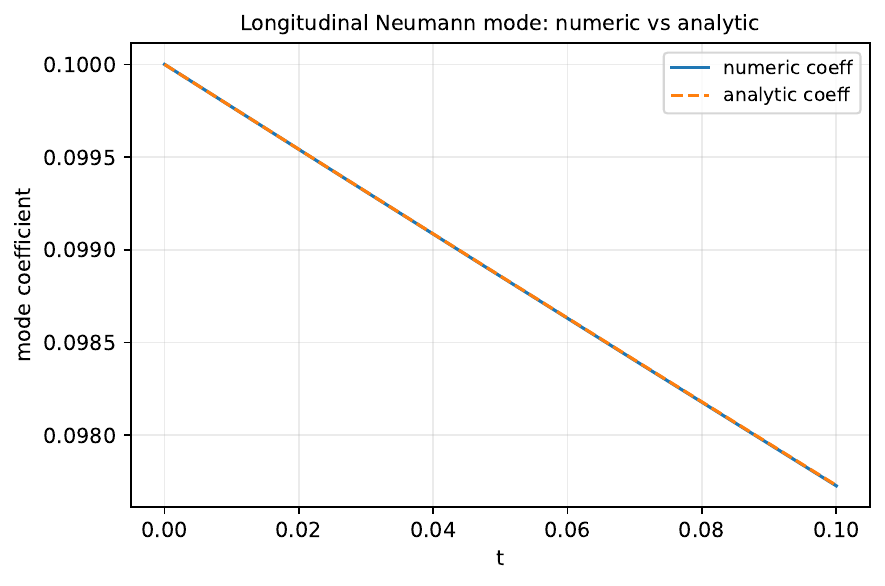}
\caption{
Longitudinal diffusion+$T_1$ analytical benchmark.
Time evolution of the modal coefficient $c(t)$ for the Neumann eigenmode \eqref{eq:neumann_phi} compared against \eqref{eq:c_sol}.
For $(1,1,1)$ in the unit box, $\lambda=3\pi^2$ and the predicted decay rate is $D\lambda+T_1^{-1}=0.2296$.
}
\label{fig:analytic_longitudinal}
\end{figure}

\subsection{Long-time lab-frame oscillations with DDF}\label{sec:validation_osc}

\Cref{fig:osc_long} shows the long-time lab-frame evolution of the global transverse signal
$S(t)=\langle M_x+iM_y\rangle$ with DDF enabled, where $\langle\cdot\rangle$ denotes the spatial average over $\Omega$ of the corresponding FE field. The real and imaginary parts show multiple zero crossings over the plotted interval. The envelope $|S(t)|$ decays smoothly. This run is generated using the structure-preserving explicit precession update in the implementation (Rodrigues rotation at DDF quadrature points followed by an $L^2$ projection). The plotted observable is time-step converged in the following quantitative sense: if $S_{\Delta t}(t_n)$ denotes the stored time series at step size $\Delta t$ and $S_{\Delta t/2}(t_n)$ denotes the stored time series at step size $\Delta t/2$ downsampled to the coarser grid, then
\[
\epsilon_{\mathrm{dt}}
:=\frac{\|S_{\Delta t}-S_{\Delta t/2}\|_2}{\|S_{\Delta t/2}\|_2}
\]
is small for the parameter choices shown, indicating that the displayed trace is insensitive to halving $\Delta t$ at fixed spatial discretization.

\subsection{DDF-on versus DDF-off envelope}\label{sec:validation_envelope}

\Cref{fig:ddf_envelope} isolates the effect of the DDF on the global envelope $|S(t)|$ by comparing DDF off to DDF on at two DDF scalings.
The $k=0$ curve provides a control in which only diffusion, relaxation, and the uniform offset contribute.
Increasing the DDF scaling changes the envelope measurably over the same time window (Fig.~\ref{fig:ddf_envelope}).

\begin{figure}[t]
\centering
\includegraphics[width=\columnwidth]{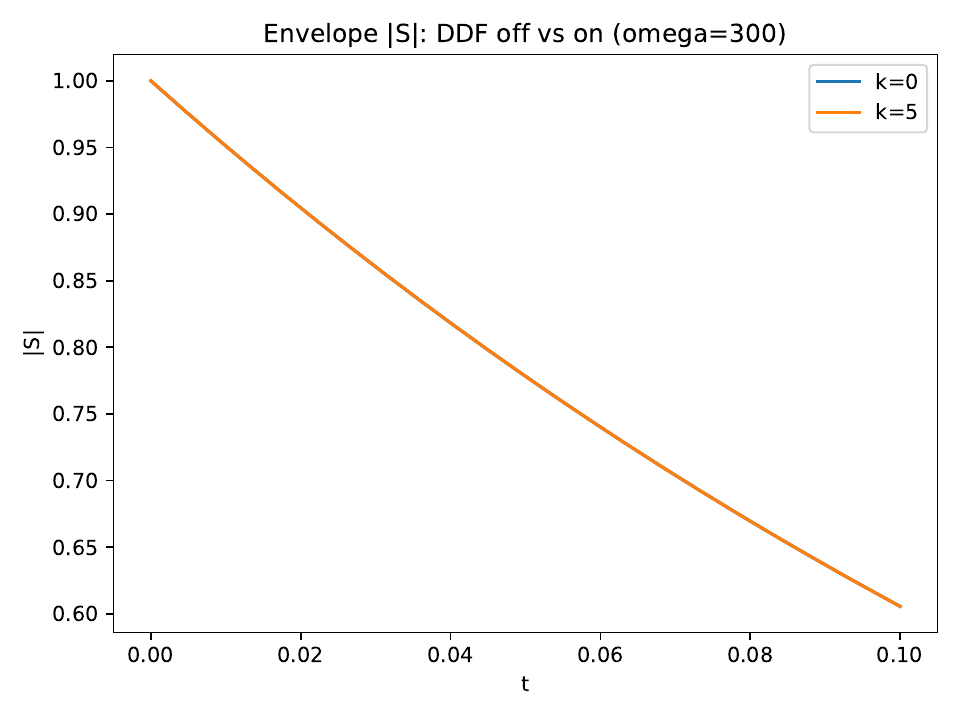}
\caption{
Envelope comparison with DDF off ($k=0$) versus DDF on ($k=5$) at otherwise fixed parameters.
}
\label{fig:ddf_envelope}
\end{figure}

\begin{figure}[t]
\centering
\includegraphics[width=\columnwidth]{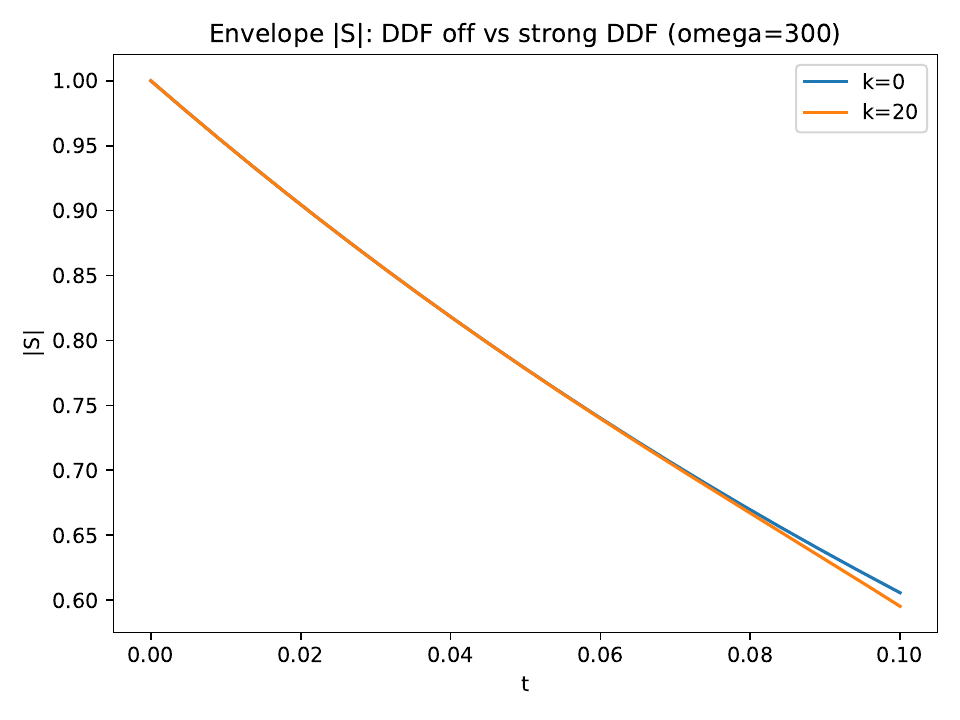}
\caption{
Envelope comparison with DDF off ($k=0$) versus stronger DDF ($k=20$) at otherwise fixed parameters.
}
\label{fig:ddf_envelope_strong}
\end{figure}

\subsection{Gradient dephasing with and without DDF}\label{sec:validation_gradient}

\Cref{fig:gz200} shows the envelope response under a constant $z$-gradient, comparing DDF off to DDF on.
The gradient drives rapid dephasing (and partial rephasing in the bounded domain), while the DDF modifies the envelope through nonlinear nonlocal feedback.
This example is included to illustrate a regime where spatial phase structure is present, which can amplify the macroscopic impact of long-range dipolar interactions.

\begin{figure}[t]
\centering
\includegraphics[width=\columnwidth]{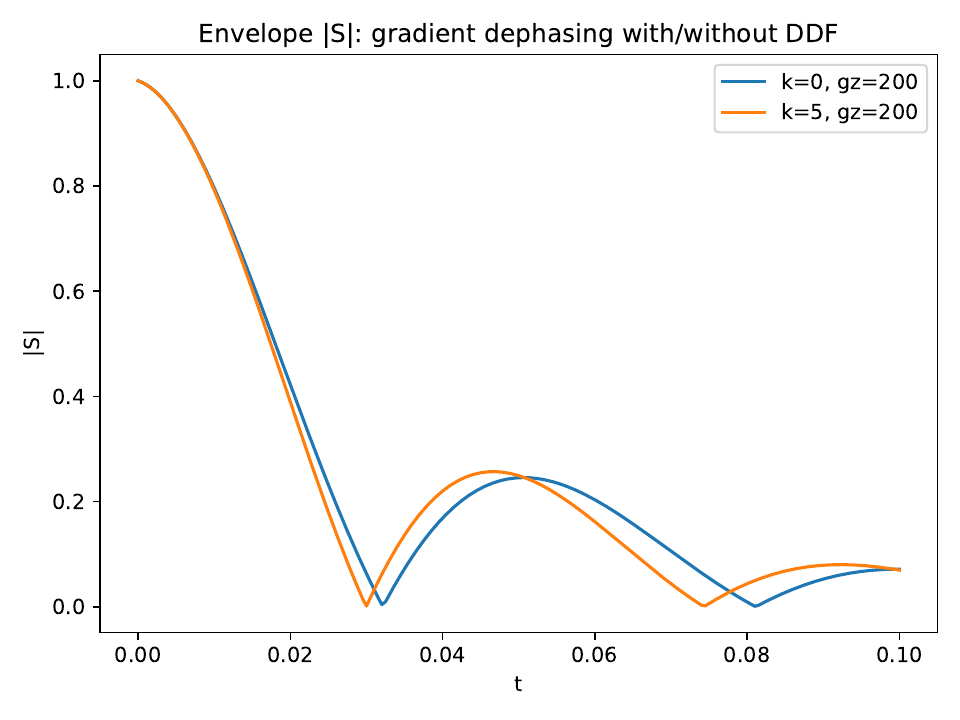}
\caption{
Envelope $|S(t)|$ under a constant $z$-gradient, comparing DDF off ($k=0$) to DDF on ($k=5$) at otherwise fixed parameters.
}
\label{fig:gz200}
\end{figure}

\section{Geometry-driven advantages of finite elements on curved domains}\label{sec:fe_adv_geometry}

This section continues the validation studies by focusing on curved Neumann boundaries, where body-fitted FE discretizations are expected to outperform voxelized FD baselines at comparable resolution. A central motivation for a FE formulation is robust treatment of curved and complex boundaries with reflective diffusion conditions. Finite-difference (FD) baselines on Cartesian grids can be accurate on simple boxes, but on curved domains they typically require either embedded-boundary technology (cut cells, ghost-fluid methods, or level-set reconstructions) or accept staircased boundaries whose discrete Neumann condition is only approximate. In this section we quantify this geometry effect by comparing FE and FD against an analytical reference on a sphere, using a benchmark that is sensitive to boundary accuracy.

\subsection{Analytical reference: Neumann diffusion plus $T_1$ relaxation on a sphere}\label{sec:analytic_sphere_neumann}

Let $\Omega=\{\mathbf r\in\mathbb{R}^3:\|\mathbf r\|<R\}$ be a ball of radius $R$ with reflective (Neumann) boundary condition $\partial_n u=0$ on $\partial\Omega$. Consider the longitudinal perturbation
\begin{equation}
u(\mathbf r,t) := M_z(\mathbf r,t)-M_0,
\end{equation}
with $M_x=M_y=0$ so that precession is inactive. For constant $D$ and $T_1$, the governing equation reduces to the linear problem
\begin{equation}
\partial_t u = D\Delta u - \frac{1}{T_1}u,
\quad
\partial_n u = 0 \ \text{on } \partial\Omega.
\label{eq:sphere_u_pde}
\end{equation}
Separation of variables in spherical coordinates yields eigenfunctions of the form
\begin{equation}
u_{\ell m n}(\mathbf r,t) = c_{\ell m n}(t)\, j_\ell\!\left(\alpha_{\ell n}\frac{r}{R}\right)Y_{\ell m}(\theta,\phi),
\end{equation}
where $j_\ell$ is a spherical Bessel function and $Y_{\ell m}$ is a spherical harmonic. The Neumann boundary condition implies $j_\ell'(\alpha_{\ell n})=0$. In this work we use the radially symmetric family $\ell=0$, for which
\begin{equation}
j_0(x)=\frac{\sin x}{x},
\quad
j_0'(x)=\frac{x\cos x-\sin x}{x^2}.
\end{equation}
Thus the Neumann condition $j_0'(\alpha_{0n})=0$ is equivalent to
\begin{equation}
\tan \alpha = \alpha,
\label{eq:tan_alpha_eq}
\end{equation}
with roots $0<\alpha_1<\alpha_2<\cdots$. The corresponding eigenvalue is $\lambda_n=(\alpha_n/R)^2$, and the coefficient satisfies
\begin{multline}
\frac{dc_n}{dt}
= -\left(D\lambda_n+\frac{1}{T_1}\right)c_n, \\
c_n(t)
= c_n(0)\exp\left[-\left(D\left(\frac{\alpha_n}{R}\right)^2+\frac{1}{T_1}\right)t\right].
\label{eq:c_sphere_decay}
\end{multline}
We emphasize that \eqref{eq:c_sphere_decay} is a closed-form, parameter-free reference once $(D,T_1,R)$ and the root index $n$ are specified. It is therefore well suited as a gold-standard check for boundary handling in numerical discretizations on curved domains.

\subsection{Numerical experiment: mapped-geometry FE versus voxelized FD}\label{sec:demo2_geometry}

We compare two discretizations of \eqref{eq:sphere_u_pde}. In the FE approach (mapped sphere), the sphere is represented by a geometry-mapped hexahedral mesh (isoparametric map from the reference cube to the sphere), and diffusion is assembled in the FE weak form with reflective boundary conditions imposed naturally through the variational formulation. The modal coefficient $c_n(t)$ is extracted by an $L^2$ projection onto the analytical eigenfunction evaluated on the numerical quadrature points. In the FD baseline (voxel mask), the sphere is represented as a voxel mask on a Cartesian grid, and diffusion is advanced by an explicit 7-point Laplacian restricted to the mask. At the mask boundary we use a simple no-flux treatment based on omitting fluxes to neighbors outside the mask. This voxel-mask FD is included as a straightforward Cartesian baseline; it is not an embedded-boundary or cut-cell method, and it inherits staircase boundary geometry at fixed grid resolution. The same analytical mode is used to extract the coefficient.

To stress the geometric boundary, we use the second nontrivial radial Neumann root ($n=2$ in \eqref{eq:tan_alpha_eq}), which yields a more oscillatory mode and larger boundary gradients than the fundamental mode. For $n=2$, the root is $\alpha_2\approx 7.7253$ and the analytic decay rate in \eqref{eq:c_sphere_decay} is
\begin{equation}
\rho = D\left(\frac{\alpha_2}{R}\right)^2 + \frac{1}{T_1}.
\end{equation}

\subsection{Results: FE achieves smaller error at fixed resolution}\label{sec:demo2_results}

\Cref{fig:sphere_diffusion_root2_R030,fig:sphere_diffusion_root2_R025} show the coefficient error $|c_{\mathrm{num}}(t)-c_{\mathrm{an}}(t)|$ for the mapped-geometry FE discretization and the voxel-mask FD baseline on two challenging sphere configurations. In both cases, FE yields a smaller error over most of the time window. The aggregate relative $L^2$ errors confirm the visual trend: for $R=0.30$ and $nx=12$, FE attains $2.66\times 10^{-3}$ while the voxel-mask FD baseline attains $8.26\times 10^{-3}$ (FD/FE $\approx 3.1$). For $R=0.25$ and $nx=10$, FE attains $5.48\times 10^{-3}$ while the voxel-mask FD baseline attains $2.59\times 10^{-2}$ (FD/FE $\approx 4.7$). These results indicate a geometry-driven benefit of the mapped-geometry FE formulation for boundary-sensitive Neumann diffusion modes when compared against the analytical reference \eqref{eq:c_sphere_decay}.

\begin{table}[t]
\centering
\caption{Curved-boundary Neumann diffusion benchmark on a sphere using the second radial Neumann root ($n=2$, $\alpha_2\approx 7.7253$). Errors are relative $L^2$ errors of the extracted modal coefficient $c(t)$ against the analytical decay \eqref{eq:c_sphere_decay}. The FD method is a voxel-mask baseline with a simple no-flux treatment at the mask boundary.}
\begin{tabular}{lccc}
\hline
Case & rel.\ $L^2$ (FE) & rel.\ $L^2$ (FD) & FD/FE \\
\hline
$R=0.30$, $nx=12$ & $2.66\times 10^{-3}$ & $8.26\times 10^{-3}$ & $3.1$ \\
$R=0.25$, $nx=10$ & $5.48\times 10^{-3}$ & $2.59\times 10^{-2}$ & $4.7$ \\
\hline
\end{tabular}
\label{tab:sphere_diffusion_root2}
\end{table}

\begin{figure}[t]
\centering
\includegraphics[width=\columnwidth]{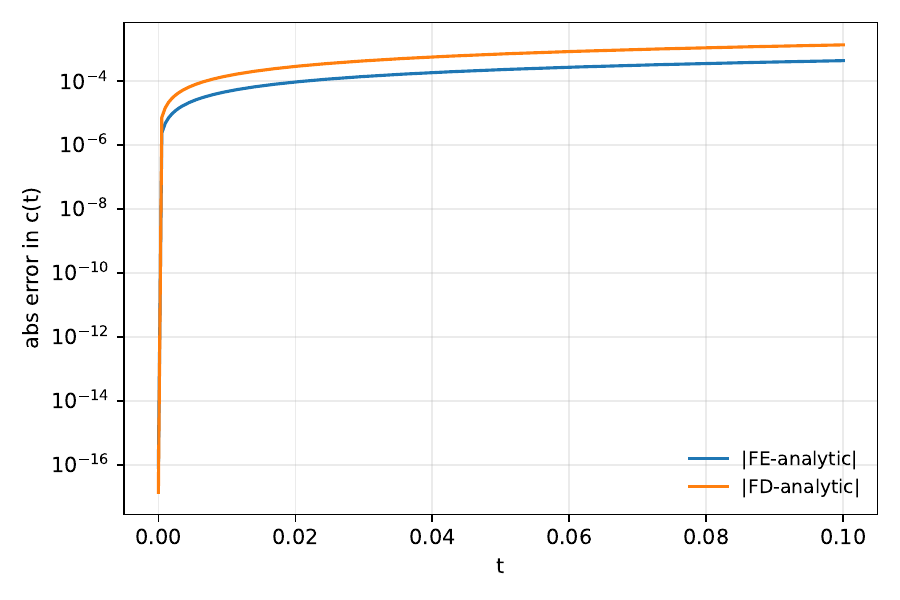}
\caption{Sphere diffusion benchmark ($R=0.30$, $nx=12$, root-index $n=2$). Plotted is the absolute error in the extracted coefficient $c(t)$ relative to the analytical decay \eqref{eq:c_sphere_decay} for FE (mapped geometry) and the voxel-mask FD baseline.}
\label{fig:sphere_diffusion_root2_R030}
\end{figure}

\begin{figure}[t]
\centering
\includegraphics[width=\columnwidth]{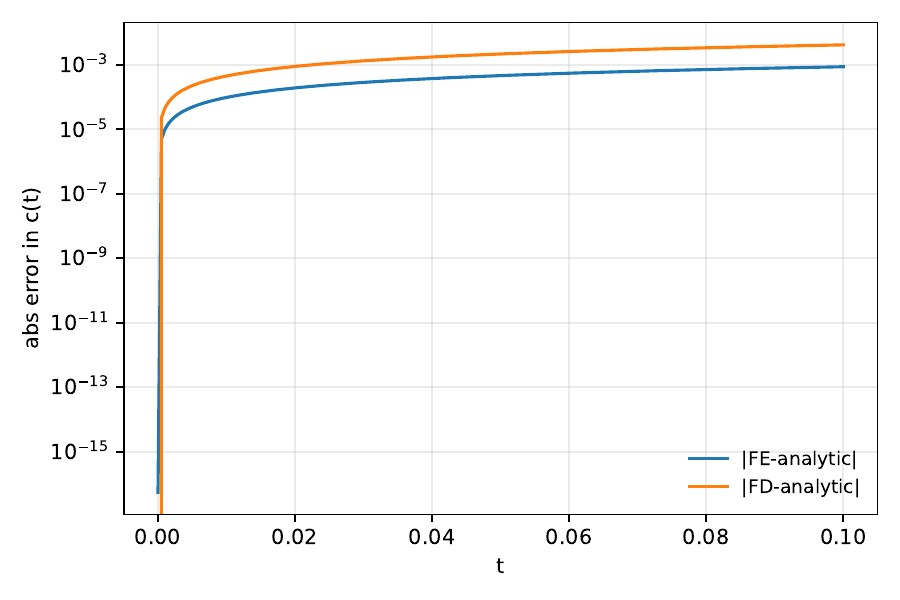}
\caption{Sphere diffusion benchmark ($R=0.25$, $nx=10$, root-index $n=2$). Same diagnostic as \Cref{fig:sphere_diffusion_root2_R030}, showing a larger separation between FE and FD as the curved boundary becomes more poorly resolved on the Cartesian grid.}
\label{fig:sphere_diffusion_root2_R025}
\end{figure}

On a voxelized curved boundary, the discrete no-flux condition is enforced only approximately and the effective boundary location is grid-dependent. For smooth modes this may be adequate, but for more oscillatory modes the staircased boundary geometry and boundary-condition approximation can introduce a systematic bias in the decay rate and mode shape. In contrast, the FE formulation imposes the reflective diffusion condition through the weak form on a geometry-mapped boundary, which yields smaller coefficient error for this benchmark at comparable grid resolution. These sphere results quantify the effect against the analytical reference \eqref{eq:c_sphere_decay} and motivate the use of FE discretizations when curved Neumann boundaries are important (Table~\ref{tab:sphere_diffusion_root2}).

\section{Conclusion}\label{sec:conclusion}

We presented a FE weak formulation of the Bloch equations with the distant dipolar field and derived the semi-discrete system \eqref{eq:semi_discrete} for bounded domains with spatially varying material parameters and optional flow. Representative long-time lab-frame oscillations and envelope-level DDF effects are shown in \Cref{sec:validation} (see \Cref{fig:osc_long,fig:ddf_envelope,fig:gz200}). The formulation relies on a regularized DDF kernel with a short-distance length scale $a>0$. The regularization yields a bounded DDF operator (\Cref{prop:Ta_bdd}) and enables standard variational analysis on bounded domains.

We established an $L^2$ energy balance and associated a priori estimate in which precession is neutral and diffusion and transverse relaxation are dissipative (\Cref{prop:energy}). We proved local well-posedness with continuous dependence on the data (\Cref{thm:local}, \Cref{prop:contdep}) and obtained global existence under additional conditions that ensure transport does not inject energy through volume compression or boundary fluxes (\Cref{cor:global}). At the discrete level we showed an energy identity for the FE semi-discretization under an energy-neutral advection discretization (\Cref{thm:semi_disc_energy}) and a corresponding stability result for a second-order IMEX scheme under a time-step restriction controlled by effective precession and transport bounds (\Cref{thm:imex_stability}).

A major emphasis of this work is validation and practical robustness on bounded domains. We validated the implementation against closed-form analytical benchmarks that isolate key components of the model: a uniform-mode DDF reduction with a deterministically computed kernel average, a periodic plane-wave eigenmode based on the kernel symbol, and a longitudinal Neumann diffusion+$T_1$ mode. We also quantified a geometry-driven effect of curved Neumann boundaries by comparing mapped-geometry FE against a voxel-mask finite-difference baseline on a spherical Neumann eigenmode decay for a boundary-sensitive mode, where FE yields smaller error at comparable grid resolution (see \Cref{sec:fe_adv_geometry}).

On the algorithmic side, geometry-dependent operators (mass, diffusion, relaxation, and auxiliary structures for near/far evaluation) are assembled once. For bounded-domain runs, the nonlocal DDF is applied at each step in a matrix-free manner using a real-space near/far evaluation, so no periodicity assumptions are required. The time integrator treats diffusion and relaxation implicitly and treats precession and advection explicitly. In the implementation, a structure-preserving explicit precession stage (Rodrigues rotation at DDF quadrature points followed by an $L^2$ projection) enables stable multi-cycle lab-frame simulations.

There are limitations. Results depend on the chosen regularization length $a$ and on the accuracy of the quadrature and far-field approximation used in the DDF evaluation. Energy-neutral transport requires $\nabla\cdot\vec v=0$ and $\vec v\cdot\hat{\mathbf n}=0$ on $\partial\Omega$ (or else an appropriate inflow treatment and additional bounds). Fully rigorous error estimates for compressed far-field operators (e.g., FMM or $\mathcal H$-matrices) are not included. These topics are natural targets for future work, together with adaptive $h/p$ refinement, anisotropic diffusion tensors, a principal-value analysis of the singular kernel limit $a\to 0$, and validation against  alternative nonlocal solvers on canonical geometries.

Overall, the combination of a bounded-operator setting, discrete energy structure, and validated matrix-free implementation supports Bloch--DDF simulation on bounded domains where geometry and boundaries influence the dynamics.

\appendix
\section{Collected theorem-like statements and proofs}\label{app:collected_results}
\noindent In this Appendix we have collected lemma, proposition, corollary, and theorem statements as well as their proofs in order to streamline the presentation in the main body.

\subsection{Continuum operator and weak-solution results}
\noindent This subsection collects the continuum lemmas, propositions, corollary, and theorem used in the operator, energy, and weak-solution analysis.

\begin{proposition}[Boundedness of $\mathcal T_a$]\label{prop:Ta_bdd}
Assume \Cref{sec:assumptions} and fix $a>0$. There exists a constant $C_a=C(a,\Omega)>0$ such that for all $1\le p\le\infty$ and all $\vec M\in (L^p(\Omega))^3$,
\begin{equation}
\|\mathcal T_a[\vec M]\|_{L^p(\Omega)} \le C_a \|\vec M\|_{L^p(\Omega)}.
\end{equation}
Moreover, for any $s\in[0,1]$, there exists $C_{a,s}=C(a,\Omega,s)>0$ such that for all $\vec M\in H^s(\Omega)^3$,
\begin{equation}
\|\mathcal T_a[\vec M]\|_{H^s(\Omega)} \le C_{a,s}\,\|\vec M\|_{H^s(\Omega)}.
\end{equation}
\end{proposition}

\begin{proof}
Fix $a>0$. For $\mathbf x\ne \mathbf 0$, define
\begin{equation}
K_a(\mathbf x)
:=
\frac{1-3(\hat{\mathbf z}\cdot \hat{\mathbf x})^2}{2(\|\mathbf x\|^2+a^2)^{3/2}},
\quad
\hat{\mathbf x}:=\mathbf x/\|\mathbf x\|,
\end{equation}
and set $K_a(\mathbf 0):=0$. The value at $\mathbf x=\mathbf 0$ is immaterial for the integral operator.

Since $\Omega\subset\mathbb R^3$ is bounded, there exists $R>0$ such that
\begin{equation}
\Omega-\Omega:=\{\mathbf r-\mathbf r' : \mathbf r,\mathbf r'\in\Omega\}\subset B_R(0).
\end{equation}
Choose $\chi\in C_c^\infty(\mathbb R^3)$ such that $\chi\equiv 1$ on $B_R(0)$, and define
\begin{equation}
G_a(\mathbf x):=\chi(\mathbf x)\,K_a(\mathbf x).
\end{equation}
Then $G_a$ is compactly supported, and for every $\mathbf r,\mathbf r'\in\Omega$ one has
\begin{equation}
G_a(\mathbf r-\mathbf r')=K_a(\mathbf r-\mathbf r').
\label{eq:Ta_cutoff_identity}
\end{equation}
We first show that
\begin{equation}
G_a\in W^{1,1}(\mathbb R^3).
\label{eq:Ga_W11}
\end{equation}
Indeed, for $\mathbf x\ne \mathbf 0$,
\begin{equation}
|K_a(\mathbf x)|
\le
\frac{1}{(\|\mathbf x\|^2+a^2)^{3/2}}
\le a^{-3},
\end{equation}
because $|1-3(\hat{\mathbf z}\cdot \hat{\mathbf x})^2|\le 2$. Hence $K_a\in L^\infty(B_R(0))$, so $G_a\in L^1(\mathbb R^3)$.

For the gradient, write
\begin{align*}
K_a(\mathbf x)=& q(\hat{\mathbf x})\,\rho_a(\|\mathbf x\|), \\
q(\omega):=&\frac{1-3(\hat{\mathbf z}\cdot \omega)^2}{2}, \\ 
\rho_a(r):= & (r^2+a^2)^{-3/2}.
\end{align*}
The function $q$ is smooth on the unit sphere, and therefore its homogeneous extension satisfies
\begin{equation}
|\nabla q(\hat{\mathbf x})|\le C\,\|\mathbf x\|^{-1}
\quad (\mathbf x\ne \mathbf 0).
\end{equation}
Also,
\begin{equation}
|\rho_a(r)|\le a^{-3},
\quad
|\rho_a'(r)|=\frac{3r}{(r^2+a^2)^{5/2}}\le C_a .
\end{equation}
By the product rule,
\begin{equation}
|\nabla K_a(\mathbf x)|
\le
C_a\bigl(1+\|\mathbf x\|^{-1}\bigr)
\quad
(\mathbf x\ne \mathbf 0).
\label{eq:grad_Ka_bound}
\end{equation}
Since $\|\mathbf x\|^{-1}\in L^1(B_R(0))$ in three dimensions, it follows that $\nabla K_a\in L^1(B_R(0))$. Because $\chi$ is smooth and compactly supported,
\begin{equation}
\nabla G_a = \chi\,\nabla K_a + K_a\,\nabla\chi \in L^1(\mathbb R^3),
\end{equation}
which proves \eqref{eq:Ga_W11}.

Now let $\vec M\in (L^p(\Omega))^3$, $1\le p\le \infty$, and let $\widetilde{\vec M}$ denote its zero extension to $\mathbb R^3$. Set
\begin{equation}
\widetilde{\vec F}:=\mathcal A\,\widetilde{\vec M}.
\end{equation}
For $\mathbf r\in\Omega$, using \eqref{eq:Ta_cutoff_identity},
\begin{align}
\mathcal T_a[\vec M](\mathbf r)
&=
\int_\Omega K_a(\mathbf r-\mathbf r')\,\mathcal A\,\vec M(\mathbf r')\,\mathrm d^3\mathbf r'
\notag\\
&=
\int_{\mathbb R^3} G_a(\mathbf r-\mathbf y)\,\widetilde{\vec F}(\mathbf y)\,\mathrm d^3\mathbf y
=
(G_a\ast \widetilde{\vec F})(\mathbf r).
\label{eq:Ta_as_convolution}
\end{align}
Therefore, by Young's inequality on $\mathbb R^3$,
\begin{align}
\|\mathcal T_a[\vec M]\|_{L^p(\Omega)}
&\le
\|G_a\ast \widetilde{\vec F}\|_{L^p(\mathbb R^3)}
\notag\\
&\le
\|G_a\|_{L^1(\mathbb R^3)}\,\|\widetilde{\vec F}\|_{L^p(\mathbb R^3)}
\notag\\
&\le
\|G_a\|_{L^1(\mathbb R^3)}\,\|\mathcal A\|\,\|\widetilde{\vec M}\|_{L^p(\mathbb R^3)}
\notag\\
&=
C_a\,\|\vec M\|_{L^p(\Omega)}.
\end{align}
This proves the $L^p$ bound.

We next prove the $H^1$ estimate, in the stronger form
\begin{equation}
\|\mathcal T_a[\vec M]\|_{H^1(\Omega)}
\le
C_a\,\|\vec M\|_{L^2(\Omega)}
\,
\text{for all } \vec M\in (L^2(\Omega))^3.
\label{eq:Ta_L2_to_H1}
\end{equation}
By \eqref{eq:Ta_as_convolution} and the fact that $G_a\in W^{1,1}(\mathbb R^3)$, differentiation in the sense of distributions gives
\begin{equation}
\nabla \mathcal T_a[\vec M]
=
\bigl((\nabla G_a)\ast \widetilde{\vec F}\bigr)\big|_\Omega .
\end{equation}
Hence, using Young's inequality again,
\begin{align}
\|\nabla \mathcal T_a[\vec M]\|_{L^2(\Omega)}
&\le
\|(\nabla G_a)\ast \widetilde{\vec F}\|_{L^2(\mathbb R^3)}
\notag\\
&\le
\|\nabla G_a\|_{L^1(\mathbb R^3)}\,\|\widetilde{\vec F}\|_{L^2(\mathbb R^3)}
\notag\\
&\le
\|\nabla G_a\|_{L^1(\mathbb R^3)}\,\|\mathcal A\|\,\|\vec M\|_{L^2(\Omega)}.
\end{align}
Combining this with the already proved $L^2$ bound yields \eqref{eq:Ta_L2_to_H1}.

Now fix $s\in[0,1]$. Since $\mathcal T_a$ is bounded
\begin{equation}
\mathcal T_a : (L^2(\Omega))^3 \to (L^2(\Omega))^3
\, \text{and} \, 
\mathcal T_a : (L^2(\Omega))^3 \to H^1(\Omega)^3,
\end{equation}
standard interpolation on bounded Lipschitz domains gives
\begin{equation}
\mathcal T_a : (L^2(\Omega))^3 \to H^s(\Omega)^3
\end{equation}
and
\begin{equation}
\|\mathcal T_a[\vec M]\|_{H^s(\Omega)}
\le
C_{a,s}\,\|\vec M\|_{L^2(\Omega)}
\quad
\text{for all } \vec M\in (L^2(\Omega))^3.
\label{eq:Ta_L2_to_Hs}
\end{equation}
Finally, if $\vec M\in H^s(\Omega)^3$, then $H^s(\Omega)\hookrightarrow L^2(\Omega)$ continuously because $\Omega$ is bounded. Therefore \eqref{eq:Ta_L2_to_Hs} implies
\begin{equation}
\|\mathcal T_a[\vec M]\|_{H^s(\Omega)}
\le
C_{a,s}\,\|\vec M\|_{L^2(\Omega)}
\le
C_{a,s}\,\|\vec M\|_{H^s(\Omega)}.
\end{equation}
This proves the stated $H^s$ bound.
\end{proof}

\begin{proposition}[Energy balance and $L^2$ estimate]\label{prop:energy}
Let $\vec M$ be a smooth solution of \eqref{eq:bloch_ddf} on $\Omega\times(0,T)$ satisfying reflective diffusion boundary conditions
\begin{equation}
\hat{\mathbf n}\cdot\bigl(D(\mathbf r)\nabla M_i\bigr)=0
\quad
\text{on } \partial\Omega\times(0,T),
\quad i\in\{x,y,z\}.
\label{eq:reflective_bc_energy}
\end{equation}
Then
\begin{multline}
\frac{\mathrm{d}}{\mathrm{d}t}\,
\frac12 \int_\Omega \lVert \vec M\rVert^2\,\mathrm{d}^3\mathbf r
+ \int_\Omega D(\mathbf r)\,\lVert\nabla \vec M\rVert^2\,\mathrm{d}^3\mathbf r \\
+ \int_\Omega \frac{M_x^2+M_y^2}{T_2(\mathbf r)}\,\mathrm{d}^3\mathbf r
+ \int_\Omega \frac{M_z^2}{T_1(\mathbf r)}\,\mathrm{d}^3\mathbf r \\
= \int_\Omega \frac{M_0(\mathbf r)\,M_z}{T_1(\mathbf r)}\,\mathrm{d}^3\mathbf r
+ \frac12\int_\Omega \bigl(\nabla\cdot\vec v\bigr)\,\lVert \vec M\rVert^2\,\mathrm{d}^3\mathbf r \\
- \frac12\int_{\partial\Omega}
\bigl(\vec v\cdot\hat{\mathbf n}\bigr)\,\lVert \vec M\rVert^2\,\mathrm{d}S .
\label{eq:energy_identity}
\end{multline}
In particular, if $\nabla\cdot\vec v=0$ in $\Omega$ and $\vec v\cdot\hat{\mathbf n}=0$ on $\partial\Omega$, then the transport terms vanish and
\begin{multline}
\frac{\mathrm{d}}{\mathrm{d}t}\,
\frac12 \int_\Omega \lVert \vec M\rVert^2\,\mathrm{d}^3\mathbf r
+ \int_\Omega D(\mathbf r)\,\lVert\nabla \vec M\rVert^2\,\mathrm{d}^3\mathbf r \\
+ \int_\Omega \frac{M_x^2+M_y^2}{T_2(\mathbf r)}\,\mathrm{d}^3\mathbf r
+ \int_\Omega \frac{M_z^2}{T_1(\mathbf r)}\,\mathrm{d}^3\mathbf r \\
= \int_\Omega \frac{M_0(\mathbf r)\,M_z}{T_1(\mathbf r)}\,\mathrm{d}^3\mathbf r .
\end{multline}
Moreover, the right-hand side admits the bound
\begin{equation}
\int_\Omega \frac{M_0\,M_z}{T_1}\,\mathrm{d}^3\mathbf r
\le
\frac12\int_\Omega \frac{M_0^2}{T_1}\,\mathrm{d}^3\mathbf r
+
\frac12\int_\Omega \frac{M_z^2}{T_1}\,\mathrm{d}^3\mathbf r ,
\label{eq:energy_rhs_young}
\end{equation}
which yields the a priori estimate
\begin{multline}
\frac{\mathrm{d}}{\mathrm{d}t}\,
\frac12 \int_\Omega \lVert \vec M\rVert^2\,\mathrm{d}^3\mathbf r
+ \int_\Omega D(\mathbf r)\,\lVert\nabla \vec M\rVert^2\,\mathrm{d}^3\mathbf r \\
+ \int_\Omega \frac{M_x^2+M_y^2}{T_2(\mathbf r)}\,\mathrm{d}^3\mathbf r
+ \frac12\int_\Omega \frac{M_z^2}{T_1(\mathbf r)}\,\mathrm{d}^3\mathbf r \\
\le
\frac12\int_\Omega \frac{M_0(\mathbf r)^2}{T_1(\mathbf r)}\,\mathrm{d}^3\mathbf r
+ \frac12\int_\Omega \bigl(\nabla\cdot\vec v\bigr)\,\lVert \vec M\rVert^2\,\mathrm{d}^3\mathbf r \\
- \frac12\int_{\partial\Omega}
\bigl(\vec v\cdot\hat{\mathbf n}\bigr)\,\lVert \vec M\rVert^2\,\mathrm{d}S .
\label{eq:energy_inequality}
\end{multline}
\end{proposition}

\begin{proof}
Since $\vec M$ is smooth, each term below is classically well defined and the integrations by parts are justified.

Take the $L^2(\Omega)^3$ inner product of \eqref{eq:bloch_ddf} with $\vec M$ and integrate over $\Omega$. This gives
\begin{align}
\int_\Omega \vec M\cdot \partial_t\vec M\,\mathrm d^3\mathbf r
&=
\gamma\int_\Omega \vec M\cdot\Bigl(\vec M\times(\vec B_d[\vec M]+\delta\vec B)\Bigr)\,\mathrm d^3\mathbf r \nonumber \\
& 
-
\int_\Omega \vec M\cdot\bigl((\vec v\cdot\nabla)\vec M\bigr)\,\mathrm d^3\mathbf r
\notag\\
&\quad
+
\int_\Omega \vec M\cdot \nabla\cdot\bigl(D(\mathbf r)\nabla \vec M\bigr)\,\mathrm d^3\mathbf r \nonumber \\
& -
\int_\Omega \vec M\cdot \frac{M_x\,\hat{\mathbf x}+M_y\,\hat{\mathbf y}}{T_2(\mathbf r)}\,\mathrm d^3\mathbf r
\notag\\
&\quad
+
\int_\Omega \vec M\cdot \frac{M_0(\mathbf r)-M_z}{T_1(\mathbf r)}\,\hat{\mathbf z}\,\mathrm d^3\mathbf r .
\label{eq:energy_start}
\end{align}
We evaluate the terms on the right-hand side one by one.

For the time derivative, using $\vec M\cdot\partial_t\vec M=\frac12\partial_t|\vec M|^2$, we obtain
\begin{equation}
\int_\Omega \vec M\cdot \partial_t\vec M\,\mathrm d^3\mathbf r
=
\frac{\mathrm d}{\mathrm dt}\,
\frac12\int_\Omega |\vec M|^2\,\mathrm d^3\mathbf r .
\label{eq:energy_time}
\end{equation}
For the precession term, the pointwise identity
\begin{equation}
\vec a\cdot(\vec a\times \vec b)=0
\quad
\text{for all } \vec a,\vec b\in\mathbb R^3
\end{equation}
implies
\begin{equation}
\gamma\int_\Omega \vec M\cdot\Bigl(\vec M\times(\vec B_d[\vec M]+\delta\vec B)\Bigr)\,\mathrm d^3\mathbf r
=0.
\label{eq:energy_precession}
\end{equation}
For the advection term, since
\begin{equation}
\vec M\cdot\bigl((\vec v\cdot\nabla)\vec M\bigr)
=
\frac12\,\vec v\cdot\nabla |\vec M|^2,
\end{equation}
the divergence theorem gives
\begin{align}
-\int_\Omega  \vec M\cdot\bigl((\vec v\cdot\nabla)\vec M\bigr)\,\mathrm d^3\mathbf r  
& = -\frac12\int_\Omega \vec v\cdot\nabla |\vec M|^2\,\mathrm d^3\mathbf r
\notag\\
&=
-\frac12\int_\Omega \nabla\cdot\bigl(\vec v\,|\vec M|^2\bigr)\,\mathrm d^3\mathbf r \nonumber \\
& 
+
\frac12\int_\Omega (\nabla\cdot \vec v)\,|\vec M|^2\,\mathrm d^3\mathbf r
\notag\\
&=
\frac12\int_\Omega (\nabla\cdot \vec v)\,|\vec M|^2\,\mathrm d^3\mathbf r  \nonumber \\
& 
-
\frac12\int_{\partial\Omega} (\vec v\cdot\hat{\mathbf n})\,|\vec M|^2\,\mathrm dS .
\label{eq:energy_transport}
\end{align}
For the diffusion term, writing componentwise and integrating by parts,
\begin{align}
\int_\Omega \vec M\cdot & \nabla\cdot\bigl(D(\mathbf r)\nabla \vec M\bigr)\,\mathrm d^3\mathbf r
\\
& = 
\sum_{i\in\{x,y,z\}}
\int_\Omega M_i\,\nabla\cdot\bigl(D(\mathbf r)\nabla M_i\bigr)\,\mathrm d^3\mathbf r
\notag\\
&=
-\sum_{i\in\{x,y,z\}}
\int_\Omega D(\mathbf r)\,|\nabla M_i|^2\,\mathrm d^3\mathbf r  \nonumber \\
&
+
\sum_{i\in\{x,y,z\}}
\int_{\partial\Omega} M_i\,\hat{\mathbf n}\cdot\bigl(D(\mathbf r)\nabla M_i\bigr)\,\mathrm dS .
\end{align}
The boundary term vanishes by \eqref{eq:reflective_bc_energy}, so
\begin{equation}
\int_\Omega \vec M\cdot \nabla\cdot\bigl(D(\mathbf r)\nabla \vec M\bigr)\,\mathrm d^3\mathbf r
=
-\int_\Omega D(\mathbf r)\,|\nabla \vec M|^2\,\mathrm d^3\mathbf r .
\label{eq:energy_diffusion}
\end{equation}
For the transverse relaxation term,
\begin{equation}
-\int_\Omega \vec M\cdot \frac{M_x\,\hat{\mathbf x}+M_y\,\hat{\mathbf y}}{T_2(\mathbf r)}\,\mathrm d^3\mathbf r
=
-\int_\Omega \frac{M_x^2+M_y^2}{T_2(\mathbf r)}\,\mathrm d^3\mathbf r .
\label{eq:energy_T2}
\end{equation}
For the longitudinal relaxation and recovery term,
\begin{align}
\int_\Omega \vec M\cdot \frac{M_0(\mathbf r)-M_z}{T_1(\mathbf r)}\,\hat{\mathbf z}\,\mathrm d^3\mathbf r
&=
\int_\Omega \frac{(M_0(\mathbf r)-M_z)M_z}{T_1(\mathbf r)}\,\mathrm d^3\mathbf r
\notag\\
&=
\int_\Omega \frac{M_0(\mathbf r)\,M_z}{T_1(\mathbf r)}\,\mathrm d^3\mathbf r \\
& \quad 
-
\int_\Omega \frac{M_z^2}{T_1(\mathbf r)}\,\mathrm d^3\mathbf r .
\label{eq:energy_T1}
\end{align}
Substituting \eqref{eq:energy_time}, \eqref{eq:energy_precession}, \eqref{eq:energy_transport}, \eqref{eq:energy_diffusion}, \eqref{eq:energy_T2}, and \eqref{eq:energy_T1} into \eqref{eq:energy_start} yields
\eqref{eq:energy_identity}. The stated special case $\nabla\cdot\vec v=0$ in $\Omega$ and $\vec v\cdot\hat{\mathbf n}=0$ on $\partial\Omega$ follows immediately.

To derive \eqref{eq:energy_rhs_young}, use the pointwise Young inequality
\begin{equation}
ab\le \frac12 a^2+\frac12 b^2
\quad
\text{for all } a,b\in\mathbb R,
\end{equation}
with
\begin{equation}
a=\frac{M_0(\mathbf r)}{\sqrt{T_1(\mathbf r)}},
\quad
b=\frac{M_z(\mathbf r)}{\sqrt{T_1(\mathbf r)}}.
\end{equation}
Because $T_1(\mathbf r)>0$, this gives
\begin{equation}
\frac{M_0(\mathbf r)\,M_z(\mathbf r)}{T_1(\mathbf r)}
\le
\frac12\,\frac{M_0(\mathbf r)^2}{T_1(\mathbf r)}
+
\frac12\,\frac{M_z(\mathbf r)^2}{T_1(\mathbf r)} .
\label{eq:energy_young_pointwise}
\end{equation}
Integrating \eqref{eq:energy_young_pointwise} over $\Omega$ yields \eqref{eq:energy_rhs_young}. Inserting \eqref{eq:energy_rhs_young} into \eqref{eq:energy_identity} gives \eqref{eq:energy_inequality}.
\end{proof}

\begin{proposition}[Continuous dependence]\label{prop:contdep}
Let $\vec M^{(1)}$ and $\vec M^{(2)}$ be weak solutions on $[0,T^\ast]$
with initial data $\vec M_{\mathrm{init}}^{(1)}$ and
$\vec M_{\mathrm{init}}^{(2)}$. Then, for all $t\in[0,T^\ast]$,
\begin{multline}
\|\vec M^{(1)}(t)-\vec M^{(2)}(t)\|_{H}^2
\\
\le
\exp\Biggl(
C t
+
C\sum_{\ell=1}^2
\int_0^t \|\vec M^{(\ell)}(s)\|_{V}^{4/3}\,\mathrm ds
\Biggr)
\|\vec M_{\mathrm{init}}^{(1)}-\vec M_{\mathrm{init}}^{(2)}\|_{H}^2,
\label{eq:contdep_explicit}
\end{multline}
where $C$ depends only on $D_{\min}^{-1}$, $\Omega$, the Sobolev and trace
constants of $\Omega$, and $\|\mathcal T_a\|_{L^2(\Omega)^3\to L^2(\Omega)^3}$,
and, in the non-skew advection formulation, also on
$\|\nabla\cdot\vec v\|_{L^\infty(\Omega)}$,
$\|\vec v\cdot\hat{\mathbf n}\|_{L^\infty(\partial\Omega)}$,
and the chosen advection boundary conditions.

Consequently, since
$\vec M^{(1)},\vec M^{(2)}\in L^2(0,T^\ast;V)$, there exists a constant
\begin{equation}
C_\ast
=
C_\ast\Bigl(
T^\ast,
\|\vec M^{(1)}\|_{L^2(0,T^\ast;V)},
\|\vec M^{(2)}\|_{L^2(0,T^\ast;V)}
\Bigr)
\end{equation}
such that
\begin{equation}
\|\vec M^{(1)}(t)-\vec M^{(2)}(t)\|_{H}^2
\le
C_\ast e^{C_\ast t}
\|\vec M_{\mathrm{init}}^{(1)}-\vec M_{\mathrm{init}}^{(2)}\|_{H}^2
\end{equation}
for all $t\in[0,T^\ast]$.
\end{proposition}

\begin{proof}
Set
\begin{equation}
\vec W := \vec M^{(1)}-\vec M^{(2)},
\quad
\vec B^{(\ell)}
:=
\mathcal T_a[\vec M^{(\ell)}]+\delta\vec B,
\quad
\ell\in\{1,2\}.
\end{equation}
Since $\vec M^{(1)},\vec M^{(2)}\in X\cap C([0,T^\ast];H)$, we have
\begin{equation}
\vec W\in X\cap C([0,T^\ast];H).
\end{equation}
Subtract the two weak formulations. Then, for almost every
$t\in(0,T^\ast)$ and every $\vec\Phi\in V$,
\begin{align}
\langle \partial_t \vec W,\vec\Phi\rangle_{V',V}
&+
\int_\Omega
D(\mathbf r)\,
\nabla \vec W:\nabla \vec\Phi\,
\mathrm d^3\mathbf r
\notag\\
&+
\int_\Omega
\left(
\frac{W_1\Phi_1+W_2\Phi_2}{T_2(\mathbf r)}
+
\frac{W_3\Phi_3}{T_1(\mathbf r)}
\right)
\mathrm d^3\mathbf r
\notag\\
&+
\sum_{i=1}^3 a_{\mathrm{adv}}(W_i,\Phi_i)
\notag\\
&=
\gamma
\int_\Omega
\Bigl(
\vec M^{(1)}\times \vec B^{(1)}
-
\vec M^{(2)}\times \vec B^{(2)}
\Bigr)\cdot \vec\Phi\,
\mathrm d^3\mathbf r .
\label{eq:contdep_weak_difference}
\end{align}
The inhomogeneous recovery term cancels because it is the same in both
equations.

Choose $\vec\Phi=\vec W(t)$. Since
$\vec W\in L^2(0,T^\ast;V)\cap H^1(0,T^\ast;V')$, the Lions--Magenes
identity yields
\begin{equation}
\langle \partial_t\vec W(t),\vec W(t)\rangle_{V',V}
=
\frac12\frac{\mathrm d}{\mathrm dt}\|\vec W(t)\|_H^2
\end{equation}
for almost every $t\in(0,T^\ast)$. Therefore
\begin{align}
\frac12\frac{\mathrm d}{\mathrm dt}\|\vec W(t)\|_H^2
&+
\int_\Omega
D(\mathbf r)\,|\nabla \vec W|^2\,
\mathrm d^3\mathbf r
\notag\\
&+
\int_\Omega
\frac{W_1^2+W_2^2}{T_2(\mathbf r)}\,
\mathrm d^3\mathbf r
+
\int_\Omega
\frac{W_3^2}{T_1(\mathbf r)}\,
\mathrm d^3\mathbf r
\notag\\
&=
\gamma
\int_\Omega
\Bigl(
\vec M^{(1)}\times \vec B^{(1)}
-
\vec M^{(2)}\times \vec B^{(2)}
\Bigr)\cdot \vec W\,
\mathrm d^3\mathbf r
\notag\\
&\quad
-
\sum_{i=1}^3 a_{\mathrm{adv}}(W_i,W_i).
\label{eq:contdep_energy_identity}
\end{align}
We first estimate the nonlinear precession term. Using
\begin{equation}
\vec M^{(1)}\times \vec B^{(1)}
-
\vec M^{(2)}\times \vec B^{(2)}
=
\vec W\times \vec B^{(1)}
+
\vec M^{(2)}\times \mathcal T_a[\vec W],
\label{eq:contdep_cross_decomp}
\end{equation}
and the pointwise identity
\begin{equation}
\vec W\cdot\bigl(\vec W\times \vec B^{(1)}\bigr)=0,
\end{equation}
we obtain
\begin{align}
&
\left|
\int_\Omega
\Bigl(
\vec M^{(1)}\times \vec B^{(1)}
-
\vec M^{(2)}\times \vec B^{(2)}
\Bigr)\cdot \vec W\,
\mathrm d^3\mathbf r
\right|
\notag\\
&\quad =
\left|
\int_\Omega
\bigl(\vec M^{(2)}\times \mathcal T_a[\vec W]\bigr)\cdot \vec W\,
\mathrm d^3\mathbf r
\right|
\notag\\
&\quad \le
\|\vec M^{(2)}\|_{L^6(\Omega)}
\|\vec W\|_{L^3(\Omega)}
\|\mathcal T_a[\vec W]\|_{L^2(\Omega)} .
\label{eq:contdep_precession_est1}
\end{align}
By Sobolev embedding, interpolation, and the $L^2$ boundedness of
$\mathcal T_a$,
\begin{align*}
\|\vec M^{(2)}\|_{L^6(\Omega)}
&\le
C\|\vec M^{(2)}\|_V,
\\
\|\vec W\|_{L^3(\Omega)}
&\le
C\|\vec W\|_H^{1/2}\|\vec W\|_V^{1/2},
\\
\|\mathcal T_a[\vec W]\|_{L^2(\Omega)}
&\le
C\|\vec W\|_H.
\end{align*}
Hence
\begin{multline}
\left|
\gamma
\int_\Omega
\Bigl(
\vec M^{(1)}\times \vec B^{(1)}
-
\vec M^{(2)}\times \vec B^{(2)}
\Bigr)\cdot \vec W\,
\mathrm d^3\mathbf r
\right| \\
\le
C\|\vec M^{(2)}\|_V
\|\vec W\|_H^{3/2}
\|\vec W\|_V^{1/2}.
\label{eq:contdep_precession_est2}
\end{multline}
Applying Young's inequality with exponents $4$ and $4/3$ gives, for every
$\varepsilon>0$,
\begin{equation}
C\|\vec M^{(2)}\|_V
\|\vec W\|_H^{3/2}
\|\vec W\|_V^{1/2}
\le
\varepsilon \|\vec W\|_V^2
+
C_\varepsilon
\|\vec M^{(2)}\|_V^{4/3}
\|\vec W\|_H^2.
\label{eq:contdep_precession_est3}
\end{equation}
Since
\begin{equation}
\|\vec W\|_V^2
=
\|\vec W\|_H^2
+
\|\nabla \vec W\|_{L^2(\Omega)}^2
\end{equation}
and $D(\mathbf r)\ge D_{\min}>0$, we may choose $\varepsilon$ small enough
to obtain
\begin{multline}
\left|
\gamma
\int_\Omega
\Bigl(
\vec M^{(1)}\times \vec B^{(1)}
-
\vec M^{(2)}\times \vec B^{(2)}
\Bigr)\cdot \vec W\,
\mathrm d^3\mathbf r
\right|
\\
\le
\frac{D_{\min}}{4}
\|\nabla \vec W\|_{L^2(\Omega)}^2
+
C\Bigl(
1+\|\vec M^{(2)}\|_V^{4/3}
\Bigr)
\|\vec W\|_H^2 .
\label{eq:contdep_precession_est4}
\end{multline}
Interchanging the roles of $\vec M^{(1)}$ and $\vec M^{(2)}$ gives the
symmetric bound
\begin{multline}
\left|
\gamma
\int_\Omega
\Bigl(
\vec M^{(1)}\times \vec B^{(1)}
-
\vec M^{(2)}\times \vec B^{(2)}
\Bigr)\cdot \vec W\,
\mathrm d^3\mathbf r
\right|
\\
\le
\frac{D_{\min}}{4}
\|\nabla \vec W\|_{L^2(\Omega)}^2 \\
+
C\Bigl(
1
+
\|\vec M^{(1)}\|_V^{4/3} 
+
\|\vec M^{(2)}\|_V^{4/3}
\Bigr)
\|\vec W\|_H^2 .
\label{eq:contdep_precession_final}
\end{multline}
We next estimate the advection contribution. If advection is written in the
skew-symmetric form under \eqref{eq:transport_assumptions}, then
\begin{equation}
a_{\mathrm{adv}}(W_i,W_i)=0,
\quad i=1,2,3.
\label{eq:contdep_adv_skew}
\end{equation}
If instead one uses the standard advective form
\begin{equation}
a_{\mathrm{adv}}(W_i,W_i)
=
\int_\Omega
(\vec v\cdot \nabla W_i)\,W_i\,
\mathrm d^3\mathbf r,
\end{equation}
then integration by parts yields
\begin{multline}
-\sum_{i=1}^3 a_{\mathrm{adv}}(W_i,W_i)
=
\frac12
\int_\Omega
(\nabla\cdot \vec v)\,
|\vec W|^2\,
\mathrm d^3\mathbf r
\\
-
\frac12
\int_{\partial\Omega}
(\vec v\cdot \hat{\mathbf n})\,
|\vec W|^2\,
\mathrm dS .
\label{eq:contdep_adv_ibp}
\end{multline}
The volume term satisfies
\begin{equation}
\left|
\frac12
\int_\Omega
(\nabla\cdot \vec v)\,
|\vec W|^2\,
\mathrm d^3\mathbf r
\right|
\le
\frac12
\|\nabla\cdot \vec v\|_{L^\infty(\Omega)}
\|\vec W\|_H^2 .
\label{eq:contdep_adv_volume}
\end{equation}
For the boundary term, the trace inequality gives
\begin{equation}
\|\vec W\|_{L^2(\partial\Omega)}^2
\le
C_{\mathrm{tr}}
\|\vec W\|_H
\|\vec W\|_V.
\end{equation}
Hence, for every $\varepsilon>0$,
\begin{equation}
\|\vec W\|_{L^2(\partial\Omega)}^2
\le
\varepsilon \|\vec W\|_V^2
+
C_{\varepsilon,\Omega}\|\vec W\|_H^2,
\label{eq:contdep_trace}
\end{equation}
and therefore
\begin{multline}
\left|
\frac12
\int_{\partial\Omega}
(\vec v\cdot \hat{\mathbf n})\,
|\vec W|^2\,
\mathrm dS
\right|
\\
\le
\varepsilon \|\vec W\|_V^2
+
C_{\varepsilon,\Omega}
\|\vec v\cdot \hat{\mathbf n}\|_{L^\infty(\partial\Omega)}
\|\vec W\|_H^2 .
\label{eq:contdep_adv_boundary}
\end{multline}
Substituting \eqref{eq:contdep_precession_final} into
\eqref{eq:contdep_energy_identity}, and, in the non-skew advection case,
also using \eqref{eq:contdep_adv_volume} and
\eqref{eq:contdep_adv_boundary}, we choose $\varepsilon>0$ sufficiently
small so that all gradient terms are absorbed by the diffusion term. Since
the relaxation terms are nonnegative, they may be discarded. We obtain, for
almost every $t\in(0,T^\ast)$,
\begin{equation}
\frac{\mathrm d}{\mathrm dt}\|\vec W(t)\|_H^2
\le
g(t)\,\|\vec W(t)\|_H^2,
\label{eq:contdep_diffineq}
\end{equation}
where one may take
\begin{multline}
g(t)
=
C\Bigl(
1
+
\|\vec M^{(1)}(t)\|_V^{4/3}
+
\|\vec M^{(2)}(t)\|_V^{4/3}
\Bigr)
\end{multline}
in the skew-advection case, and in the non-skew case one adds the constant
terms coming from
$\|\nabla\cdot \vec v\|_{L^\infty(\Omega)}$ and
$\|\vec v\cdot \hat{\mathbf n}\|_{L^\infty(\partial\Omega)}$.

Because $\vec M^{(1)},\vec M^{(2)}\in L^2(0,T^\ast;V)$ and $4/3<2$, we have
$g\in L^1(0,T^\ast)$. Gr\"onwall's inequality therefore yields
\begin{multline}
\|\vec W(t)\|_H^2
\le
\exp\left(
\int_0^t g(s)\,\mathrm ds
\right)
\|\vec W(0)\|_H^2
\\
=
\exp\left(
\int_0^t g(s)\,\mathrm ds
\right)
\|\vec M_{\mathrm{init}}^{(1)}
-
\vec M_{\mathrm{init}}^{(2)}\|_H^2,
\quad
0\le t\le T^\ast.
\label{eq:contdep_gronwall}
\end{multline}
This proves \eqref{eq:contdep_explicit}.

Finally, Hölder's inequality gives
\begin{multline*}
\int_0^t
\|\vec M^{(\ell)}(s)\|_V^{4/3}\,\mathrm ds
\le
t^{1/3}
\|\vec M^{(\ell)}\|_{L^2(0,t;V)}^{4/3} \\
\le
(T^\ast)^{1/3}
\|\vec M^{(\ell)}\|_{L^2(0,T^\ast;V)}^{4/3},
\end{multline*}
so the exponential factor in \eqref{eq:contdep_explicit} is bounded by
$C_\ast e^{C_\ast t}$ for a suitable constant $C_\ast$ depending on
$T^\ast$ and the two $L^2(0,T^\ast;V)$ norms. This proves the final stated
estimate.
\end{proof}

\begin{corollary}[Global existence under additional conditions]\label{cor:global}
Assume \Cref{sec:assumptions} and fix $a>0$. If either $\vec v\equiv 0$
or \eqref{eq:transport_assumptions} holds and advection is imposed in an
energy-neutral form, then the weak solution in \Cref{thm:local} extends to
$[0,T]$ for any $T>0$.
\end{corollary}

\begin{proof}
Let $[0,T_{\max})$ denote the maximal interval of existence of the unique
weak solution furnished by \Cref{thm:local}. We prove that
$T_{\max}=\infty$. Suppose, for contradiction, that $T_{\max}<\infty$.

Under the hypotheses of the corollary, the advection contribution is either
absent or energy-neutral. Therefore the weak solution satisfies the same
a priori estimate as in \eqref{eq:energy_inequality}, with the transport
terms omitted. More precisely, this estimate follows by the standard
regularization and density argument for weak solutions: one first derives
the energy identity for sufficiently smooth approximants and then passes to
the limit using weak lower semicontinuity. Hence, for almost every
$t\in(0,T_{\max})$,
\begin{multline}
\frac{\mathrm d}{\mathrm dt}\,
\frac12\|\vec M(t)\|_H^2
+
\int_\Omega
D(\mathbf r)\,
|\nabla \vec M(t)|^2\,
\mathrm d^3\mathbf r
\\
+
\int_\Omega
\frac{M_1(t)^2+M_2(t)^2}{T_2(\mathbf r)}\,
\mathrm d^3\mathbf r
+
\frac12
\int_\Omega
\frac{M_3(t)^2}{T_1(\mathbf r)}\,
\mathrm d^3\mathbf r \\
\le
\frac12
\int_\Omega
\frac{M_0(\mathbf r)^2}{T_1(\mathbf r)}\,
\mathrm d^3\mathbf r .
\label{eq:global_energy_start}
\end{multline}
Integrating \eqref{eq:global_energy_start} from $0$ to $\tau<T_{\max}$
yields
\begin{multline}
\frac12\|\vec M(\tau)\|_H^2
+
\int_0^\tau\!\!\int_\Omega
D(\mathbf r)\,
|\nabla \vec M(t)|^2\,
\mathrm d^3\mathbf r\,\mathrm dt
\\
+
\int_0^\tau\!\!\int_\Omega
\frac{M_1(t)^2+M_2(t)^2}{T_2(\mathbf r)}\,
\mathrm d^3\mathbf r\,\mathrm dt \\
+
\frac12
\int_0^\tau\!\!\int_\Omega
\frac{M_3(t)^2}{T_1(\mathbf r)}\,
\mathrm d^3\mathbf r\,\mathrm dt
\\
\le
\frac12\|\vec M_{\mathrm{init}}\|_H^2
+
\frac{\tau}{2}
\int_\Omega
\frac{M_0(\mathbf r)^2}{T_1(\mathbf r)}\,
\mathrm d^3\mathbf r .
\label{eq:global_energy_integrated}
\end{multline}
It follows that
\begin{equation}
\sup_{0\le t<T_{\max}}\|\vec M(t)\|_H^2
\le
C_{T_{\max}},
\label{eq:global_H_bound}
\end{equation}
where $C_{T_{\max}}$ depends only on $T_{\max}$,
$\|\vec M_{\mathrm{init}}\|_H$, and
$\|M_0/\sqrt{T_1}\|_{L^2(\Omega)}$. Since
$D(\mathbf r)\ge D_{\min}>0$, \eqref{eq:global_energy_integrated} also
implies
\begin{equation}
\int_0^{T_{\max}}
\|\nabla \vec M(t)\|_{L^2(\Omega)}^2\,
\mathrm dt
\le
C_{T_{\max}}.
\label{eq:global_grad_bound}
\end{equation}
Combining \eqref{eq:global_H_bound} and \eqref{eq:global_grad_bound}, we
obtain
\begin{equation}
\vec M\in
L^\infty(0,T_{\max};H)\cap L^2(0,T_{\max};V).
\label{eq:global_X_part1}
\end{equation}
We next estimate $\partial_t\vec M$ in $L^2(0,T_{\max};V')$. Let
$\vec\Phi\in V$. Using the weak formulation \eqref{eq:weak_bloch}, we have,
for almost every $t\in(0,T_{\max})$,
\begin{align}
\left|
\langle \partial_t \vec M(t),\vec\Phi\rangle_{V',V}
\right|
&\le
|\gamma|
\left|
\int_\Omega
\bigl(\vec M(t)\times \mathcal T_a[\vec M(t)]\bigr)\cdot \vec\Phi\,
\mathrm d^3\mathbf r
\right|
\notag\\
&
+
|\gamma|
\left|
\int_\Omega
\bigl(\vec M(t)\times \delta\vec B\bigr)\cdot \vec\Phi\,
\mathrm d^3\mathbf r
\right|
\notag\\
&
+
\left|
\int_\Omega
D(\mathbf r)\,
\nabla \vec M(t):\nabla \vec\Phi\,
\mathrm d^3\mathbf r
\right|
\notag\\
&
+
\Biggl|
\int_\Omega
\Bigl(
\frac{M_1(t)\Phi_1+M_2(t)\Phi_2}{T_2(\mathbf r)}  \nonumber \\
& \quad +
\frac{M_3(t)\Phi_3}{T_1(\mathbf r)}
\Bigr)
\mathrm d^3\mathbf r
\Biggr|
\notag \\
& 
+
\left|
\sum_{i=1}^3 a_{\mathrm{adv}}(M_i(t),\Phi_i)
\right|  \nonumber \\
& \quad 
+
\left|
\int_\Omega
\frac{M_0(\mathbf r)}{T_1(\mathbf r)}\,
\Phi_3\,
\mathrm d^3\mathbf r
\right|.
\label{eq:global_dt_start}
\end{align}
Each term on the right-hand side is bounded by a multiple of
$\|\vec\Phi\|_V$. For the DDF term, Sobolev embedding,
$V\hookrightarrow L^3(\Omega)^3$, and boundedness of $\mathcal T_a$ on
$L^2(\Omega)^3$ yield
\begin{align}
\left|
\int_\Omega
\bigl(\vec M\times \mathcal T_a[\vec M]\bigr)\cdot \vec\Phi\,
\mathrm d^3\mathbf r
\right|
&\le
\|\vec M\|_{L^6(\Omega)}
\|\mathcal T_a[\vec M]\|_{L^2(\Omega)}
\|\vec\Phi\|_{L^3(\Omega)}
\notag\\
&\le
C\,
\|\vec M\|_V\,
\|\vec M\|_H\,
\|\vec\Phi\|_V .
\label{eq:global_dt_ddf}
\end{align}
For the offset-field term,
\begin{multline}
\left|
\int_\Omega
\bigl(\vec M\times \delta\vec B\bigr)\cdot \vec\Phi\,
\mathrm d^3\mathbf r
\right|
\le
\|\delta\vec B\|_{L^\infty(\Omega)}
\|\vec M\|_H
\|\vec\Phi\|_H \\
\le
C\,
\|\vec M\|_H
\|\vec\Phi\|_V .
\label{eq:global_dt_offset}
\end{multline}
For diffusion,
\begin{equation}
\left|
\int_\Omega
D(\mathbf r)\,
\nabla \vec M:\nabla \vec\Phi\,
\mathrm d^3\mathbf r
\right|
\le
\|D\|_{L^\infty(\Omega)}
\|\vec M\|_V
\|\vec\Phi\|_V .
\label{eq:global_dt_diff}
\end{equation}
For relaxation,
\begin{align}
\left|
\int_\Omega
\left(
\frac{M_1\Phi_1+M_2\Phi_2}{T_2(\mathbf r)}
+
\frac{M_3\Phi_3}{T_1(\mathbf r)}
\right)
\mathrm d^3\mathbf r
\right|
&\le
C\,
\|\vec M\|_H
\|\vec\Phi\|_H
\notag\\
&\le
C\,
\|\vec M\|_H
\|\vec\Phi\|_V .
\label{eq:global_dt_relax}
\end{align}
If $\vec v\equiv 0$, the advection term is absent. If advection is imposed
in the skew form under \eqref{eq:transport_assumptions}, then
\begin{equation}
|a_{\mathrm{adv}}(M_i,\Phi_i)|
\le
C\,
\|M_i\|_{H^1(\Omega)}
\|\Phi_i\|_{H^1(\Omega)},
\quad i=1,2,3,
\label{eq:global_dt_adv}
\end{equation}
so the total advection contribution is bounded by
$C\|\vec M\|_V\|\vec\Phi\|_V$. Finally,
\begin{equation}
\left|
\int_\Omega
\frac{M_0(\mathbf r)}{T_1(\mathbf r)}\,
\Phi_3\,
\mathrm d^3\mathbf r
\right|
\le
\left\|
\frac{M_0}{T_1}
\right\|_{L^2(\Omega)}
\|\vec\Phi\|_H
\le
C\|\vec\Phi\|_V .
\label{eq:global_dt_source}
\end{equation}
Combining
\eqref{eq:global_dt_start}--\eqref{eq:global_dt_source}
and taking the supremum over $\vec\Phi\in V$ with $\|\vec\Phi\|_V=1$, we
obtain
\begin{equation}
\|\partial_t \vec M(t)\|_{V'}
\le
C\Bigl(
\|\vec M(t)\|_V\|\vec M(t)\|_H
+
\|\vec M(t)\|_V
+
\|\vec M(t)\|_H
+
1
\Bigr)
\label{eq:global_dt_bound_pointwise}
\end{equation}
for almost every $t\in(0,T_{\max})$. By \eqref{eq:global_X_part1}, the
right-hand side belongs to $L^2(0,T_{\max})$. Therefore
\begin{equation}
\partial_t \vec M\in L^2(0,T_{\max};V').
\label{eq:global_dt_bound}
\end{equation}
Together with \eqref{eq:global_X_part1}, this gives
\begin{equation}
\vec M\in
L^2(0,T_{\max};V)\cap H^1(0,T_{\max};V').
\label{eq:global_in_X}
\end{equation}
By \Cref{lem:LM}, $\vec M$ admits an $H$-continuous representative on
$[0,T_{\max}]$. In particular, the one-sided limit
\begin{equation}
\vec M_\ast
:=
\lim_{t\uparrow T_{\max}} \vec M(t)
\end{equation}
exists in $H$.

We now restart the local theory at time $T_{\max}$. Since the equation is
autonomous and the coefficients are time-independent, the same argument as
in \Cref{thm:local} applies with initial time $T_{\max}$ and initial datum
$\vec M_\ast$. Therefore there exist $\varepsilon>0$ and a weak solution
\begin{multline*}
\widetilde{\vec M}\in
L^2(T_{\max},T_{\max}+\varepsilon;V)
\cap
H^1(T_{\max},T_{\max}+\varepsilon;V') \\
\cap
C([T_{\max},T_{\max}+\varepsilon];H)
\end{multline*}
such that
\begin{equation}
\widetilde{\vec M}(T_{\max})=\vec M_\ast .
\end{equation}
Define
\begin{equation}
\widehat{\vec M}(t):=
\begin{cases}
\vec M(t), & 0\le t\le T_{\max},\\
\widetilde{\vec M}(t), & T_{\max}\le t\le T_{\max}+\varepsilon.
\end{cases}
\end{equation}
Because both pieces satisfy the weak formulation on their respective time
intervals and coincide at $t=T_{\max}$ in $H$, the function
$\widehat{\vec M}$ is a weak solution on $[0,T_{\max}+\varepsilon]$. This
contradicts the maximality of $T_{\max}$.

Hence $T_{\max}=\infty$. Therefore the weak solution extends to $[0,T]$ for
every finite $T>0$.
\end{proof}

\begin{lemma}[Precession is $L^2$-neutral]\label{lem:neutral}
For any $\vec M,\vec B\in (L^2(\Omega))^3$,
\begin{equation}
\int_\Omega
\vec M\cdot\bigl(\vec M\times \vec B\bigr)\,
\mathrm d^3\mathbf r
=0.
\end{equation}
\end{lemma}

\begin{proof}
Since $\vec M,\vec B\in (L^2(\Omega))^3$, all component functions are
measurable, and therefore
\begin{equation}
\mathbf r\mapsto
\vec M(\mathbf r)\cdot
\bigl(\vec M(\mathbf r)\times \vec B(\mathbf r)\bigr)
\end{equation}
is measurable on $\Omega$, being a polynomial expression in the components
of $\vec M$ and $\vec B$.

For almost every $\mathbf r\in\Omega$, the vectors
$\vec M(\mathbf r)$ and $\vec B(\mathbf r)$ are defined, and the algebraic
identity
\begin{equation}
\vec a\cdot(\vec a\times \vec b)=0
\quad
\text{for all } \vec a,\vec b\in\mathbb R^3
\end{equation}
gives
\begin{equation}
\vec M(\mathbf r)\cdot
\bigl(\vec M(\mathbf r)\times \vec B(\mathbf r)\bigr)
=0
\quad
\text{for a.e. } \mathbf r\in\Omega.
\end{equation}
Hence the integrand vanishes almost everywhere, so its integral is zero.
\end{proof}

\begin{lemma}[Lipschitz estimate for the precession nonlinearity]\label{lem:lipschitz}
Define
\begin{equation}
\mathcal N(\vec M)
=
\vec M\times
\Bigl(
\mathcal T_a[\vec M]+\delta\vec B
\Bigr).
\end{equation}
There exists a constant $C>0$, depending only on $a$, $\Omega$, and
$\|\delta\vec B\|_{L^\infty(\Omega)}$, such that for all
$\vec M,\vec N\in H^1(\Omega)^3$,
\begin{multline}
\bigl\|
\mathcal N(\vec M)-\mathcal N(\vec N)
\bigr\|_{H^{-1}(\Omega)^3}
\\
\le
C\Bigl(
\|\vec M\|_{H^1(\Omega)}
+
\|\vec N\|_{H^1(\Omega)}
+
1
\Bigr)
\|\vec M-\vec N\|_{H^1(\Omega)} .
\label{eq:lipschitz_estimate}
\end{multline}
\end{lemma}

\begin{proof}
Let
\begin{equation}
\vec W:=\vec M-\vec N.
\end{equation}
Then
\begin{align}
\mathcal N(\vec M)-\mathcal N(\vec N)
&=
\vec W\times
\Bigl(
\mathcal T_a[\vec M]+\delta\vec B
\Bigr)
+
\vec N\times \mathcal T_a[\vec W].
\label{eq:N_diff_decomp}
\end{align}
We estimate the $H^{-1}(\Omega)^3$ norm through the duality with
$H^1(\Omega)^3$:
\begin{equation}
\|\vec F\|_{H^{-1}(\Omega)^3}
=
\sup_{
\substack{
\vec v\in H^1(\Omega)^3\\
\|\vec v\|_{H^1(\Omega)}=1
}}
\left|
\int_\Omega
\vec F\cdot \vec v\,
\mathrm d^3\mathbf r
\right|.
\label{eq:Hminus1_dual_norm}
\end{equation}
Fix $\vec v\in H^1(\Omega)^3$ with
$\|\vec v\|_{H^1(\Omega)}=1$.

We begin with the first term in \eqref{eq:N_diff_decomp}. By the pointwise
bound
\begin{equation}
|(\vec a\times \vec b)\cdot \vec c|
\le
|\vec a|\,|\vec b|\,|\vec c|,
\end{equation}
followed by H\"older's inequality with exponents $(6,3,2)$, we obtain
\begin{align}
&
\left|
\int_\Omega
\Bigl(
\vec W\times
\bigl(
\mathcal T_a[\vec M]+\delta\vec B
\bigr)
\Bigr)\cdot \vec v\,
\mathrm d^3\mathbf r
\right|
\notag\\
&\quad
\le
\|\vec W\|_{L^6(\Omega)}
\Bigl(
\|\mathcal T_a[\vec M]\|_{L^3(\Omega)}
+
\|\delta\vec B\|_{L^3(\Omega)}
\Bigr)
\|\vec v\|_{L^2(\Omega)}.
\label{eq:term1_start}
\end{align}
Since $\Omega$ is bounded and Lipschitz,
\begin{align}
\|\vec W\|_{L^6(\Omega)}
&\le
C\|\vec W\|_{H^1(\Omega)},
\\
\|\vec v\|_{L^2(\Omega)}
&\le
C\|\vec v\|_{H^1(\Omega)}
=
C.
\end{align}
Also, by \Cref{prop:Ta_bdd},
\begin{equation}
\|\mathcal T_a[\vec M]\|_{L^3(\Omega)}
\le
C\|\vec M\|_{L^3(\Omega)}
\le
C\|\vec M\|_{H^1(\Omega)},
\end{equation}
and since $\delta\vec B\in L^\infty(\Omega)^3$ and $\Omega$ is bounded,
\begin{equation}
\|\delta\vec B\|_{L^3(\Omega)}
\le
|\Omega|^{1/3}
\|\delta\vec B\|_{L^\infty(\Omega)}.
\end{equation}
Substituting these bounds into \eqref{eq:term1_start} gives
\begin{multline}
\left|
\int_\Omega
\Bigl(
\vec W\times
\bigl(
\mathcal T_a[\vec M]+\delta\vec B
\bigr)
\Bigr)\cdot \vec v\,
\mathrm d^3\mathbf r
\right|
\\
\le
C\Bigl(
\|\vec M\|_{H^1(\Omega)}+1
\Bigr)
\|\vec W\|_{H^1(\Omega)}.
\label{eq:term1_bound}
\end{multline}
For the second term in \eqref{eq:N_diff_decomp}, the same pointwise bound
and the same H\"older exponents give
\begin{align}
&
\left|
\int_\Omega
\bigl(
\vec N\times \mathcal T_a[\vec W]
\bigr)\cdot \vec v\,
\mathrm d^3\mathbf r
\right|
\notag\\
&\quad
\le
\|\vec N\|_{L^6(\Omega)}
\|\mathcal T_a[\vec W]\|_{L^3(\Omega)}
\|\vec v\|_{L^2(\Omega)}.
\label{eq:term2_start}
\end{align}
Using again the Sobolev embedding
$H^1(\Omega)\hookrightarrow L^6(\Omega)$,
the boundedness of $\mathcal T_a$ on $L^3(\Omega)^3$, and the embedding
$H^1(\Omega)\hookrightarrow L^3(\Omega)$, we obtain
\begin{align}
\|\vec N\|_{L^6(\Omega)}
&\le
C\|\vec N\|_{H^1(\Omega)},
\\
\|\mathcal T_a[\vec W]\|_{L^3(\Omega)}
&\le
C\|\vec W\|_{L^3(\Omega)}
\le
C\|\vec W\|_{H^1(\Omega)},
\\
\|\vec v\|_{L^2(\Omega)}
&\le
C.
\end{align}
Therefore
\begin{multline}
\left|
\int_\Omega
\bigl(
\vec N\times \mathcal T_a[\vec W]
\bigr)\cdot \vec v\,
\mathrm d^3\mathbf r
\right|
\\
\le
C
\|\vec N\|_{H^1(\Omega)}
\|\vec W\|_{H^1(\Omega)}.
\label{eq:term2_bound}
\end{multline}
Combining \eqref{eq:term1_bound} and \eqref{eq:term2_bound}, we find
\begin{multline}
\left|
\int_\Omega
\bigl(
\mathcal N(\vec M)-\mathcal N(\vec N)
\bigr)\cdot \vec v\,
\mathrm d^3\mathbf r
\right|
\\
\le
C\Bigl(
\|\vec M\|_{H^1(\Omega)}
+
\|\vec N\|_{H^1(\Omega)}
+
1
\Bigr)
\|\vec M-\vec N\|_{H^1(\Omega)}.
\end{multline}
Taking the supremum over all
$\vec v\in H^1(\Omega)^3$ with $\|\vec v\|_{H^1(\Omega)}=1$ and using
\eqref{eq:Hminus1_dual_norm} proves \eqref{eq:lipschitz_estimate}.
\end{proof}

\begin{lemma}[Time continuity]\label{lem:LM}
If $\vec M\in X$, then $\vec M\in C([0,T];H)$.
\end{lemma}

\begin{proof}
Recall that
\begin{equation}
V=H^1(\Omega)^3,
\quad
H=L^2(\Omega)^3,
\quad
V'=H^{-1}(\Omega)^3,
\end{equation}
and that $V\hookrightarrow H$ is continuous and dense, while
$H\hookrightarrow V'$ continuously via the canonical identification of
$H$ with a subspace of $V'$. Hence
\begin{equation}
V\hookrightarrow H\hookrightarrow V'
\end{equation}
is a Gelfand triple.

By definition,
\begin{equation}
X
=
L^2(0,T;V)\cap H^1(0,T;V').
\end{equation}
The classical Lions--Magenes continuity theorem for Hilbert triples states
that
\begin{equation}
L^2(0,T;V)\cap H^1(0,T;V')
\hookrightarrow
C([0,T];H)
\end{equation}
continuously. Therefore every $\vec M\in X$ admits an
$H$-continuous representative on $[0,T]$. In particular,
\begin{equation}
\vec M\in C([0,T];H).
\end{equation}
\end{proof}

\begin{theorem}[Local existence and uniqueness]\label{thm:local}
Assume \Cref{sec:assumptions} and fix $a>0$. For any
$\vec M_{\mathrm{init}}\in H$ there exists $T^\ast>0$ and a unique weak
solution
\begin{equation}
\vec M\in X\cap C([0,T^\ast];H)
\end{equation}
on $[0,T^\ast]$.
\end{theorem}

\begin{proof}
For $T>0$, define
\begin{equation}
X_T
:=
L^2(0,T;V)\cap H^1(0,T;V'),
\end{equation}
and
\begin{multline*}
Y_T
:=
X_T\cap C([0,T];H),
\quad
\|\vec U\|_{Y_T}
:=
\|\vec U\|_{L^2(0,T;V)} \\
+
\|\partial_t\vec U\|_{L^2(0,T;V')}
+
\|\vec U\|_{C([0,T];H)}.
\end{multline*}
Since $X_T\hookrightarrow C([0,T];H)$ by \Cref{lem:LM}, the space $Y_T$ is a Banach space with the above norm.

We write the weak problem in the semilinear form
\begin{equation}
\partial_t \vec M + A\vec M
=
\gamma\,\mathcal N(\vec M)+\vec F
\quad
\text{in } V',
\label{eq:local_semilinear}
\end{equation}
where
\begin{equation}
\mathcal N(\vec M)
=
\vec M\times
\bigl(
\mathcal T_a[\vec M]+\delta\vec B
\bigr),
\end{equation}
and $\vec F\in V'$ is defined by
\begin{equation}
\langle \vec F,\vec \Phi\rangle_{V',V}
=
\int_\Omega
\frac{M_0(\mathbf r)}{T_1(\mathbf r)}\,
\Phi_3\,
\mathrm d^3\mathbf r .
\label{eq:local_F_def}
\end{equation}
The linear operator $A:V\to V'$ is induced by the bilinear form
\begin{align}
a(\vec U,\vec \Phi)
&:=
\int_\Omega
D(\mathbf r)\,
\nabla \vec U:\nabla \vec \Phi\,
\mathrm d^3\mathbf r
\notag\\
&\quad
+
\int_\Omega
\left(
\frac{U_1\Phi_1+U_2\Phi_2}{T_2(\mathbf r)}
+
\frac{U_3\Phi_3}{T_1(\mathbf r)}
\right)
\mathrm d^3\mathbf r
\notag\\
&\quad
+
\sum_{i=1}^3 a_{\mathrm{adv}}(U_i,\Phi_i),
\label{eq:local_bilinear_form}
\end{align}
through
\begin{equation}
\langle A\vec U,\vec \Phi\rangle_{V',V}
=
a(\vec U,\vec \Phi).
\end{equation}
By the assumptions on $D$, $T_1$, $T_2$, and $\vec v$, the form $a$ is
bounded on $V\times V$:
\begin{equation}
|a(\vec U,\vec \Phi)|
\le
C_A\|\vec U\|_V\|\vec \Phi\|_V
\quad
\text{for all } \vec U,\vec \Phi\in V.
\label{eq:local_a_bounded}
\end{equation}
Moreover, there exist constants $\alpha_A>0$ and $\lambda_A\ge 0$ such that
\begin{equation}
a(\vec U,\vec U)
+
\lambda_A\|\vec U\|_H^2
\ge
\alpha_A\|\vec U\|_V^2
\quad
\text{for all } \vec U\in V.
\label{eq:local_garding}
\end{equation}
Indeed, the diffusion and relaxation terms are coercive, while the
advection term is either skew and contributes nothing to
$a(\vec U,\vec U)$, or else is controlled by
$\|\vec v\|_{W^{1,\infty}(\Omega)}$, the trace theorem, and the prescribed
advection boundary conditions.

Fix $T_0>0$. By the Lions--Magenes theory for linear parabolic equations
associated with bounded bilinear forms satisfying
\eqref{eq:local_garding}, for every
$\vec G\in L^2(0,T_0;V')$ and every
$\vec M_{\mathrm{init}}\in H$ there exists a unique solution
$\vec U\in X_{T_0}$ of
\begin{equation}
\partial_t \vec U + A\vec U
=
\vec G+\vec F,
\quad
\vec U(0)=\vec M_{\mathrm{init}}.
\label{eq:local_linear_problem}
\end{equation}
By \Cref{lem:LM}, this solution belongs to $Y_{T_0}$, and there exists a
constant $C_{T_0}>0$ such that
\begin{multline}
\|\vec U\|_{Y_{T_0}}
\\
\le
C_{T_0}
\Bigl(
\|\vec M_{\mathrm{init}}\|_H
+
\|\vec G\|_{L^2(0,T_0;V')}
+
\|\vec F\|_{L^2(0,T_0;V')}
\Bigr).
\label{eq:local_linear_estimate}
\end{multline}
We next estimate the nonlinear term. Let $\vec Z\in Y_T$ with
$0<T\le T_0$. For any $\vec \Phi\in V$ with $\|\vec \Phi\|_V=1$,
\begin{align}
\left|
\langle \mathcal N(\vec Z),\vec \Phi\rangle_{V',V}
\right|
&=
\left|
\int_\Omega
\bigl(
\vec Z\times (\mathcal T_a[\vec Z]+\delta\vec B)
\bigr)\cdot \vec \Phi\,
\mathrm d^3\mathbf r
\right|
\notag\\
&\le
\|\vec Z\|_{L^2(\Omega)}
\|\mathcal T_a[\vec Z]\|_{L^6(\Omega)}
\|\vec \Phi\|_{L^3(\Omega)}
\notag\\
&\quad
+
\|\delta\vec B\|_{L^\infty(\Omega)}
\|\vec Z\|_{L^2(\Omega)}
\|\vec \Phi\|_{L^2(\Omega)}.
\label{eq:local_N_est_start}
\end{align}
By \Cref{prop:Ta_bdd}, $\mathcal T_a:H^1(\Omega)^3\to H^1(\Omega)^3$
continuously. Hence, by Sobolev embedding,
\begin{equation}
\|\mathcal T_a[\vec Z]\|_{L^6(\Omega)}
\le
C\|\mathcal T_a[\vec Z]\|_{H^1(\Omega)}
\le
C\|\vec Z\|_{H^1(\Omega)}.
\end{equation}
Also,
\begin{equation}
\|\vec \Phi\|_{L^3(\Omega)}
+
\|\vec \Phi\|_{L^2(\Omega)}
\le
C\|\vec \Phi\|_V
=
C.
\end{equation}
Therefore
\begin{equation}
\|\mathcal N(\vec Z)\|_{V'}
\le
C\Bigl(
\|\vec Z\|_H\|\vec Z\|_V
+
\|\vec Z\|_H
\Bigr)
\quad
\text{for a.e. } t\in(0,T),
\label{eq:local_N_est_pointwise}
\end{equation}
and thus
\begin{multline}
\|\mathcal N(\vec Z)\|_{L^2(0,T;V')}
\\
\le
C\,
\|\vec Z\|_{C([0,T];H)}
\Bigl(
\|\vec Z\|_{L^2(0,T;V)}
+
T^{1/2}
\Bigr).
\label{eq:local_N_est_L2}
\end{multline}
Now let $\vec Z^{(1)},\vec Z^{(2)}\in Y_T$, and set
\begin{equation}
\vec W:=\vec Z^{(1)}-\vec Z^{(2)}.
\end{equation}
Then
\begin{equation}
\mathcal N(\vec Z^{(1)})-\mathcal N(\vec Z^{(2)})
=
\vec W\times
\bigl(
\mathcal T_a[\vec Z^{(1)}]+\delta\vec B
\bigr)
+
\vec Z^{(2)}\times \mathcal T_a[\vec W].
\label{eq:local_N_diff}
\end{equation}
Arguing as above, for any $\vec \Phi\in V$ with $\|\vec \Phi\|_V=1$,
\begin{align}
&
\left|
\langle
\mathcal N(\vec Z^{(1)})-\mathcal N(\vec Z^{(2)}),
\vec \Phi
\rangle_{V',V}
\right|
\notag\\
&\quad
\le
\|\vec W\|_{L^2(\Omega)}
\|\mathcal T_a[\vec Z^{(1)}]\|_{L^6(\Omega)}
\|\vec \Phi\|_{L^3(\Omega)}
\notag\\
&\quad\quad
+
\|\delta\vec B\|_{L^\infty(\Omega)}
\|\vec W\|_{L^2(\Omega)}
\|\vec \Phi\|_{L^2(\Omega)}
\notag\\
&\quad\quad
+
\|\vec Z^{(2)}\|_{L^6(\Omega)}
\|\mathcal T_a[\vec W]\|_{L^2(\Omega)}
\|\vec \Phi\|_{L^3(\Omega)}
\notag\\
&\quad
\le
C\Bigl(
1+\|\vec Z^{(1)}\|_V+\|\vec Z^{(2)}\|_V
\Bigr)
\|\vec W\|_H.
\end{align}
Hence, for almost every $t\in(0,T)$,
\begin{equation}
\|\mathcal N(\vec Z^{(1)})-\mathcal N(\vec Z^{(2)})\|_{V'}
\le
C\Bigl(
1+\|\vec Z^{(1)}\|_V+\|\vec Z^{(2)}\|_V
\Bigr)
\|\vec W\|_H,
\label{eq:local_N_lipschitz_pointwise}
\end{equation}
and therefore
\begin{multline}
\|\mathcal N(\vec Z^{(1)})-\mathcal N(\vec Z^{(2)})\|_{L^2(0,T;V')}
\\
\le
C\Bigl(
T^{1/2}
+
\|\vec Z^{(1)}\|_{L^2(0,T;V)}
+
\|\vec Z^{(2)}\|_{L^2(0,T;V)}
\Bigr) \\
\times
\|\vec Z^{(1)}-\vec Z^{(2)}\|_{C([0,T];H)}.
\label{eq:local_N_lipschitz_L2}
\end{multline}
Let $\vec U_0\in Y_{T_0}$ denote the unique solution of the linear problem
\begin{equation}
\partial_t \vec U_0 + A\vec U_0
=
\vec F,
\quad
\vec U_0(0)=\vec M_{\mathrm{init}}
\label{eq:local_u0}
\end{equation}
on $[0,T_0]$. For $0<T\le T_0$, we still denote by $\vec U_0$ its
restriction to $[0,T]$.

Fix $\rho>0$ and define the closed ball
\begin{equation}
\mathbb B_\rho(T)
:=
\left\{
\vec Z\in Y_T:
\|\vec Z-\vec U_0\|_{Y_T}\le \rho
\right\}.
\label{eq:local_ball}
\end{equation}
Because $Y_T$ is Banach, $\mathbb B_\rho(T)$ is complete.

For $\vec Z\in \mathbb B_\rho(T)$, define $\Psi(\vec Z)=\vec U$, where
$\vec U\in Y_T$ is the unique solution of
\begin{equation}
\partial_t \vec U + A\vec U
=
\gamma\,\mathcal N(\vec Z)+\vec F,
\quad
\vec U(0)=\vec M_{\mathrm{init}}
\label{eq:local_fixed_point_eq}
\end{equation}
on $[0,T]$. This is well defined by \eqref{eq:local_linear_problem},
\eqref{eq:local_linear_estimate}, and \eqref{eq:local_N_est_L2}.

Set
\begin{align}
K_H(T,\rho)
:=&
\|\vec U_0\|_{C([0,T];H)}+\rho,
\nonumber \\
K_V(T,\rho)
:=& 
\|\vec U_0\|_{L^2(0,T;V)}+\rho.
\label{eq:local_K_def}
\end{align}
If $\vec Z\in \mathbb B_\rho(T)$, then
\begin{equation}
\|\vec Z\|_{C([0,T];H)}
\le
K_H(T,\rho),
\quad
\|\vec Z\|_{L^2(0,T;V)}
\le
K_V(T,\rho).
\label{eq:local_ball_bounds}
\end{equation}
Subtracting \eqref{eq:local_u0} from \eqref{eq:local_fixed_point_eq} and
using \eqref{eq:local_linear_estimate} gives
\begin{align}
\|\Psi(\vec Z)-\vec U_0\|_{Y_T}
&\le
C_{T_0}|\gamma|\,
\|\mathcal N(\vec Z)\|_{L^2(0,T;V')}
\notag\\
&\le
C_{T_0}|\gamma|\,
K_H(T,\rho)
\Bigl(
K_V(T,\rho)+T^{1/2}
\Bigr).
\label{eq:local_map_into_ball}
\end{align}
Likewise, if $\vec Z^{(1)},\vec Z^{(2)}\in \mathbb B_\rho(T)$, then
\begin{align}
\|\Psi(\vec Z^{(1)}) & -\Psi(\vec Z^{(2)})\|_{Y_T} \\
& \le C_{T_0}|\gamma|\,
\|\mathcal N(\vec Z^{(1)})-\mathcal N(\vec Z^{(2)})\|_{L^2(0,T;V')}
\notag\\
&\le
C_{T_0}|\gamma|
\Bigl(
T^{1/2}+2K_V(T,\rho)
\Bigr)
\notag\\
&\quad\times
\|\vec Z^{(1)}-\vec Z^{(2)}\|_{C([0,T];H)}
\notag\\
&\le
C_{T_0}|\gamma|
\Bigl(
T^{1/2}+2K_V(T,\rho)
\Bigr)
\|\vec Z^{(1)}-\vec Z^{(2)}\|_{Y_T}.
\label{eq:local_contraction_pre}
\end{align}
Since $\vec U_0\in C([0,T_0];H)\cap L^2(0,T_0;V)$, we have
\begin{equation}
\|\vec U_0\|_{C([0,T];H)}
\le
\|\vec U_0\|_{C([0,T_0];H)},
\quad
\|\vec U_0\|_{L^2(0,T;V)}
\to 0
\end{equation}
as $T\downarrow 0$. Therefore, for fixed $\rho>0$, we may choose $T^\ast\in(0,T_0]$ so small
that
\begin{equation}
C_{T_0}|\gamma|\,
K_H(T^\ast,\rho)
\Bigl(
K_V(T^\ast,\rho)+(T^\ast)^{1/2}
\Bigr)
\le
\rho
\label{eq:local_smallness_1}
\end{equation}
and
\begin{equation}
C_{T_0}|\gamma|
\Bigl(
(T^\ast)^{1/2}+2K_V(T^\ast,\rho)
\Bigr)
<
1.
\label{eq:local_smallness_2}
\end{equation}
Then \eqref{eq:local_map_into_ball} shows that
$\Psi:\mathbb B_\rho(T^\ast)\to \mathbb B_\rho(T^\ast)$, and
\eqref{eq:local_contraction_pre} together with
\eqref{eq:local_smallness_2} shows that $\Psi$ is a contraction on
$\mathbb B_\rho(T^\ast)$ with respect to the $Y_{T^\ast}$ norm.

By Banach's fixed-point theorem, there exists a unique
$\vec M\in \mathbb B_\rho(T^\ast)$ such that
\begin{equation}
\Psi(\vec M)=\vec M.
\end{equation}
By construction,
\begin{equation}
\vec M\in Y_{T^\ast}
=
X_{T^\ast}\cap C([0,T^\ast];H),
\end{equation}
and \eqref{eq:local_fixed_point_eq} shows that $\vec M$ satisfies the weak
formulation \eqref{eq:weak_bloch} on $[0,T^\ast]$ with initial condition
\begin{equation}
\vec M(0)=\vec M_{\mathrm{init}}.
\end{equation}
Thus a weak solution exists on $[0,T^\ast]$.

Finally, uniqueness among weak solutions on $[0,T^\ast]$ follows from
\Cref{prop:contdep}: if two weak solutions have the same initial datum,
then their difference vanishes identically on $[0,T^\ast]$.
\end{proof}

\subsection{Semi-discrete FE results}
\noindent This subsection collects the discrete lemmas and theorem used in the FE energy analysis.

\begin{lemma}[Discrete skew-symmetry of precession]\label{lem:disc_skew}
For any coefficient vector
\begin{equation}
w=(w_1,w_2,w_3)\in (\mathbb R^{N_h})^3,
\end{equation}
the discrete precession load vectors satisfy
\begin{equation}
\sum_{i=1}^3 w_i^\top \mathcal P_i(w) = 0 .
\label{eq:disc_skew_id}
\end{equation}
\end{lemma}

\begin{proof}
Let $\{\varphi_n\}_{n=1}^{N_h}$ be the FE basis on $\Omega$.
For $i\in\{1,2,3\}$, write
\begin{equation}
w_i=((w_i)_1,\dots,(w_i)_{N_h})^\top \in \mathbb R^{N_h},
\end{equation}
and define the discrete magnetization field
\begin{equation}
\vec M_h(\mathbf r)
=
\sum_{i=1}^3 \sum_{n=1}^{N_h}
(w_i)_n\,\varphi_n(\mathbf r)\,\hat{\mathbf e}_i.
\end{equation}
Equivalently, if we set
\begin{equation}
M_{h,i}(\mathbf r)
:=
\sum_{n=1}^{N_h} (w_i)_n\,\varphi_n(\mathbf r),
\quad i=1,2,3,
\end{equation}
then
\begin{equation}
\vec M_h(\mathbf r)
=
\sum_{i=1}^3 M_{h,i}(\mathbf r)\,\hat{\mathbf e}_i.
\end{equation}
Let $\{\mathbf r_q\}_{q=1}^{N_q}\subset \Omega$ be the quadrature points and
$\{\omega_q\}_{q=1}^{N_q}$ the associated quadrature weights used in the
matrix-free definition \eqref{eq:P_matrix_free}. For each quadrature point,
let $\vec B_h(\mathbf r_q)$ denote the discrete DDF field obtained from the
same coefficient vector $w$ by the same discrete kernel evaluation used in
\eqref{eq:P_matrix_free}.

By the definition of $\mathcal P_i(w)$, for each basis index
$m\in\{1,\dots,N_h\}$,
\begin{equation}
(\mathcal P_i(w))_m
=
\sum_{q=1}^{N_q}
\omega_q\,
\bigl(
\vec M_h(\mathbf r_q)\times \vec B_h(\mathbf r_q)
\bigr)\cdot \hat{\mathbf e}_i\,
\varphi_m(\mathbf r_q).
\label{eq:disc_skew_component}
\end{equation}
Therefore,
\begin{align}
w_i^\top \mathcal P_i(w)
&=
\sum_{m=1}^{N_h} (w_i)_m\,(\mathcal P_i(w))_m
\notag\\
&=
\sum_{m=1}^{N_h} (w_i)_m
\sum_{q=1}^{N_q}
\omega_q\,
\bigl(
\vec M_h(\mathbf r_q)\times \vec B_h(\mathbf r_q)
\bigr)\cdot \hat{\mathbf e}_i\,
\varphi_m(\mathbf r_q)
\notag\\
&=
\sum_{q=1}^{N_q}
\omega_q\,
\bigl(
\vec M_h(\mathbf r_q)\times \vec B_h(\mathbf r_q)
\bigr)\cdot \hat{\mathbf e}_i\,
\sum_{m=1}^{N_h} (w_i)_m\,\varphi_m(\mathbf r_q)
\notag\\
&=
\sum_{q=1}^{N_q}
\omega_q\,
\bigl(
\vec M_h(\mathbf r_q)\times \vec B_h(\mathbf r_q)
\bigr)\cdot \hat{\mathbf e}_i\,
M_{h,i}(\mathbf r_q).
\label{eq:disc_skew_one_component}
\end{align}
Summing over $i=1,2,3$ gives
\begin{align}
\sum_{i=1}^3 w_i^\top \mathcal P_i(w)
&=
\sum_{q=1}^{N_q}
\omega_q
\sum_{i=1}^3
\bigl(
\vec M_h(\mathbf r_q)\times \vec B_h(\mathbf r_q)
\bigr)\cdot \hat{\mathbf e}_i\,
M_{h,i}(\mathbf r_q)
\notag\\
&=
\sum_{q=1}^{N_q}
\omega_q\,
\vec M_h(\mathbf r_q)\cdot
\bigl(
\vec M_h(\mathbf r_q)\times \vec B_h(\mathbf r_q)
\bigr).
\label{eq:disc_skew_sum}
\end{align}
For each quadrature point $\mathbf r_q$, the scalar triple product vanishes:
\begin{equation}
\vec M_h(\mathbf r_q)\cdot
\bigl(
\vec M_h(\mathbf r_q)\times \vec B_h(\mathbf r_q)
\bigr)
=0,
\end{equation}
because
\begin{equation}
\vec a\cdot(\vec a\times \vec b)=0
\quad
\text{for all } \vec a,\vec b\in\mathbb R^3.
\end{equation}
Hence every summand in \eqref{eq:disc_skew_sum} is zero, and therefore
\begin{equation}
\sum_{i=1}^3 w_i^\top \mathcal P_i(w)=0.
\end{equation}
This proves \eqref{eq:disc_skew_id}.
\end{proof}

\begin{lemma}[SPD and positivity]\label{lem:spd}
The mass matrix $M\in\mathbb R^{N_h\times N_h}$ is symmetric positive
definite. If $D(\mathbf r)\ge D_{\min}>0$ almost everywhere, then
$K\in\mathbb R^{N_h\times N_h}$ is symmetric positive semidefinite and, for
every $w\in\mathbb R^{N_h}$,
\begin{multline}
w^\top K w
=
\int_\Omega
D(\mathbf r)\,
|\nabla u_h|^2\,
\mathrm d^3\mathbf r
\\
\ge
D_{\min}
\int_\Omega
|\nabla u_h|^2\,
\mathrm d^3\mathbf r,
\quad
u_h=\sum_{n=1}^{N_h} w_n\varphi_n .
\end{multline}
The relaxation matrices $S_1,S_2\in\mathbb R^{N_h\times N_h}$ are symmetric
positive semidefinite and satisfy
\begin{equation}
w_1^\top S_2 w_1
+
w_2^\top S_2 w_2
+
w_3^\top S_1 w_3
\ge 0
\end{equation}
for all $w_1,w_2,w_3\in\mathbb R^{N_h}$. If
\eqref{eq:transport_assumptions} holds and $N_v=N_v^{\mathrm{skew}}$, then
\begin{equation}
w^\top N_v w=0
\quad
\text{for all } w\in\mathbb R^{N_h}.
\end{equation}
\end{lemma}

\begin{proof}
We begin with the mass matrix. By definition,
\begin{equation}
M_{mn}
=
\int_\Omega
\varphi_n(\mathbf r)\,\varphi_m(\mathbf r)\,
\mathrm d^3\mathbf r,
\quad
1\le m,n\le N_h,
\end{equation}
so $M$ is symmetric. Let $w\in\mathbb R^{N_h}$ and define
\begin{equation}
u_h
=
\sum_{n=1}^{N_h} w_n\varphi_n.
\end{equation}
Then
\begin{align}
w^\top M w
&=
\sum_{m,n=1}^{N_h}
w_m w_n
\int_\Omega
\varphi_n\varphi_m\,
\mathrm d^3\mathbf r
\notag\\
&=
\int_\Omega
\left(
\sum_{n=1}^{N_h} w_n\varphi_n
\right)
\left(
\sum_{m=1}^{N_h} w_m\varphi_m
\right)
\mathrm d^3\mathbf r
\notag\\
&=
\int_\Omega
u_h(\mathbf r)^2\,
\mathrm d^3\mathbf r .
\label{eq:M_spd}
\end{align}
Hence $w^\top M w\ge 0$. If $w^\top M w=0$, then $u_h=0$ almost everywhere in $\Omega$. Since $u_h$ is continuous and belongs to the FE space spanned by the linearly independent basis $\{\varphi_n\}_{n=1}^{N_h}$, this implies $u_h\equiv 0$ and therefore $w=0$. Thus $M$ is positive definite.

Next consider the diffusion matrix
\begin{equation}
K_{mn}
=
\int_\Omega
D(\mathbf r)\,
\nabla\varphi_n(\mathbf r)\cdot \nabla\varphi_m(\mathbf r)\,
\mathrm d^3\mathbf r .
\end{equation}
Symmetry is immediate. For $w\in\mathbb R^{N_h}$ and the associated field
$u_h=\sum_n w_n\varphi_n$,
\begin{align}
w^\top K w
&=
\sum_{m,n=1}^{N_h}
w_m w_n
\int_\Omega
D(\mathbf r)\,
\nabla\varphi_n\cdot \nabla\varphi_m\,
\mathrm d^3\mathbf r
\notag\\
&=
\int_\Omega
D(\mathbf r)\,
\left|
\sum_{n=1}^{N_h} w_n\nabla\varphi_n
\right|^2
\mathrm d^3\mathbf r
\notag\\
&=
\int_\Omega
D(\mathbf r)\,
|\nabla u_h|^2\,
\mathrm d^3\mathbf r .
\label{eq:K_psd}
\end{align}
Since $D(\mathbf r)\ge 0$ almost everywhere, this shows that $K$ is
positive semidefinite. If in addition $D(\mathbf r)\ge D_{\min}>0$ almost
everywhere, then
\begin{equation}
w^\top K w
\ge
D_{\min}
\int_\Omega
|\nabla u_h|^2\,
\mathrm d^3\mathbf r .
\label{eq:K_lower}
\end{equation}
For the relaxation matrices,
\begin{align*}
(S_1)_{mn}
&=
\int_\Omega
T_1(\mathbf r)^{-1}\,
\varphi_n(\mathbf r)\varphi_m(\mathbf r)\,
\mathrm d^3\mathbf r,
\\
(S_2)_{mn}
&=
\int_\Omega
T_2(\mathbf r)^{-1}\,
\varphi_n(\mathbf r)\varphi_m(\mathbf r)\,
\mathrm d^3\mathbf r.
\end{align*}
These matrices are symmetric. Let $z\in\mathbb R^{N_h}$ and define
$z_h=\sum_{n=1}^{N_h} z_n\varphi_n$. Then
\begin{align*}
z^\top S_1 z
&=
\int_\Omega
T_1(\mathbf r)^{-1}\,
z_h(\mathbf r)^2\,
\mathrm d^3\mathbf r,
\\
z^\top S_2 z
&=
\int_\Omega
T_2(\mathbf r)^{-1}\,
z_h(\mathbf r)^2\,
\mathrm d^3\mathbf r.
\end{align*}
Under the standing assumptions, $T_1^{-1},T_2^{-1}\ge 0$ almost
everywhere, and therefore $S_1$ and $S_2$ are positive semidefinite.
Applying these identities to $w_1,w_2,w_3\in\mathbb R^{N_h}$ gives
\begin{multline}
w_1^\top S_2 w_1
+
w_2^\top S_2 w_2
+
w_3^\top S_1 w_3
\\
=
\int_\Omega
T_2(\mathbf r)^{-1}
\bigl(
(w_1)_h^2+(w_2)_h^2
\bigr)\,
\mathrm d^3\mathbf r \\
+
\int_\Omega
T_1(\mathbf r)^{-1}
(w_3)_h^2\,
\mathrm d^3\mathbf r
\ge 0.
\end{multline}
Finally, assume \eqref{eq:transport_assumptions} holds and
$N_v=N_v^{\mathrm{skew}}$. By definition of the skew discretization,
\begin{equation}
(N_v^{\mathrm{skew}})_{mn}
=
-\frac12
\int_\Omega
(\vec v\cdot \nabla\varphi_n)\,\varphi_m\,
\mathrm d^3\mathbf r
+
\frac12
\int_\Omega
(\vec v\cdot \nabla\varphi_m)\,\varphi_n\,
\mathrm d^3\mathbf r .
\label{eq:Nv_skew_entries}
\end{equation}
Interchanging $m$ and $n$ in \eqref{eq:Nv_skew_entries} yields
\begin{equation}
(N_v^{\mathrm{skew}})_{nm}
=
-(N_v^{\mathrm{skew}})_{mn},
\end{equation}
so $N_v^{\mathrm{skew}}$ is skew-symmetric. Therefore, for every
$w\in\mathbb R^{N_h}$,
\begin{multline}
w^\top N_v w
=
w^\top N_v^{\mathrm{skew}} w
=
\bigl(w^\top N_v^{\mathrm{skew}} w\bigr)^\top \\
=
w^\top (N_v^{\mathrm{skew}})^\top w
=
-\,w^\top N_v^{\mathrm{skew}} w,
\end{multline}
and hence
\begin{equation}
w^\top N_v w=0.
\end{equation}
This completes the proof.
\end{proof}

\begin{theorem}[Semi-discrete energy inequality]\label{thm:semi_disc_energy}
Assume either $\vec v\equiv 0$ or \eqref{eq:transport_assumptions} holds
and advection is discretized by the skew operator
$N_v^{\mathrm{skew}}$ so that
\begin{equation}
z^\top N_v^{\mathrm{skew}} z = 0
\quad
\text{for all } z\in\mathbb R^{N_h}.
\end{equation}
Then any solution of \eqref{eq:semi_discrete} satisfies
\begin{multline}
\frac{\mathrm d}{\mathrm dt}\,
\frac12
\sum_{i=1}^3 w_i^\top M\,w_i
+
\sum_{i=1}^3 w_i^\top K\,w_i
\\
+
w_1^\top S_2 w_1
+
w_2^\top S_2 w_2
+
w_3^\top S_1 w_3
=
w_3^\top b_{T_1} .
\label{eq:semi_disc_energy_id}
\end{multline}
Consequently, for all $t\ge 0$,
\begin{multline}
\frac12
\sum_{i=1}^3 w_i(t)^\top M\,w_i(t)
+
\int_0^t
\sum_{i=1}^3 w_i(s)^\top K\,w_i(s)\,
\mathrm ds
\\
+
\int_0^t
\Bigl(
w_1(s)^\top S_2 w_1(s)
+
w_2(s)^\top S_2 w_2(s)
+
w_3(s)^\top S_1 w_3(s)
\Bigr)\,
\mathrm ds
\\
\le
\frac12
\sum_{i=1}^3 w_i(0)^\top M\,w_i(0)
+
\int_0^t
w_3(s)^\top b_{T_1}\,
\mathrm ds .
\label{eq:semi_disc_energy_ineq}
\end{multline}
\end{theorem}

\begin{proof}
Let
\begin{equation}
E_h(t)
:=
\frac12
\sum_{i=1}^3 w_i(t)^\top M\,w_i(t).
\end{equation}
Since the semi-discrete system \eqref{eq:semi_discrete} is a finite-
dimensional ODE system, any solution is differentiable in time. Because the
mass matrix $M$ is constant and symmetric, we have
\begin{equation}
\frac{\mathrm d}{\mathrm dt} E_h(t)
=
\sum_{i=1}^3 w_i^\top M\,\dot w_i .
\label{eq:Eh_prime}
\end{equation}
Write the three component equations in \eqref{eq:semi_discrete} as
\begin{align}
M\,\dot w_1
&=
\gamma\,\mathcal P_1(w)
-
K\,w_1
-
S_2\,w_1
-
N_v^{\mathrm{skew}}\,w_1,
\\
M\,\dot w_2
&=
\gamma\,\mathcal P_2(w)
-
K\,w_2
-
S_2\,w_2
-
N_v^{\mathrm{skew}}\,w_2,
\\
M\,\dot w_3
&=
\gamma\,\mathcal P_3(w)
-
K\,w_3
-
S_1\,w_3
-
N_v^{\mathrm{skew}}\,w_3
+
b_{T_1}.
\end{align}
Multiply the $i$th equation on the left by $w_i^\top$ and sum over
$i=1,2,3$. Using \eqref{eq:Eh_prime}, we obtain
\begin{align}
\frac{\mathrm d}{\mathrm dt} E_h(t)
&=
\gamma
\sum_{i=1}^3
w_i^\top \mathcal P_i(w)
-
\sum_{i=1}^3
w_i^\top K\,w_i
\notag\\
&\quad
-
w_1^\top S_2 w_1
-
w_2^\top S_2 w_2
-
w_3^\top S_1 w_3
\notag\\
&\quad
-
\sum_{i=1}^3
w_i^\top N_v^{\mathrm{skew}}\,w_i
+
w_3^\top b_{T_1}.
\label{eq:energy_expansion}
\end{align}
We now evaluate the terms on the right-hand side.

First, by \Cref{lem:disc_skew},
\begin{equation}
\sum_{i=1}^3
w_i^\top \mathcal P_i(w)
=
0.
\label{eq:energy_prec_zero}
\end{equation}
Second, by \Cref{lem:spd}, the stiffness matrix $K$ is symmetric positive
semidefinite, so
\begin{equation}
w_i^\top K\,w_i \ge 0
\quad
\text{for } i=1,2,3.
\label{eq:energy_K_nonneg}
\end{equation}
Likewise, \Cref{lem:spd} gives
\begin{equation}
w_1^\top S_2 w_1
+
w_2^\top S_2 w_2
+
w_3^\top S_1 w_3
\ge 0.
\label{eq:energy_S_nonneg}
\end{equation}
Third, by the skew-advection assumption,
\begin{equation}
w_i^\top N_v^{\mathrm{skew}}\,w_i = 0
\quad
\text{for } i=1,2,3,
\label{eq:energy_adv_zero}
\end{equation}
and hence
\begin{equation}
\sum_{i=1}^3
w_i^\top N_v^{\mathrm{skew}}\,w_i
=
0.
\end{equation}
Substituting \eqref{eq:energy_prec_zero} and
\eqref{eq:energy_adv_zero} into \eqref{eq:energy_expansion} yields
\begin{multline}
\frac{\mathrm d}{\mathrm dt} E_h(t)
+
\sum_{i=1}^3
w_i^\top K\,w_i
\\
+
w_1^\top S_2 w_1
+
w_2^\top S_2 w_2
+
w_3^\top S_1 w_3
=
w_3^\top b_{T_1},
\end{multline}
which is exactly \eqref{eq:semi_disc_energy_id}.

Integrating \eqref{eq:semi_disc_energy_id} from $0$ to $t$ gives the
stronger identity
\begin{multline}
E_h(t)-E_h(0)
+
\int_0^t
\sum_{i=1}^3
w_i(s)^\top K\,w_i(s)\,
\mathrm ds
\\
+
\int_0^t
\Bigl(
w_1(s)^\top S_2 w_1(s)
+
w_2(s)^\top S_2 w_2(s)
+
w_3(s)^\top S_1 w_3(s)
\Bigr)\,
\mathrm ds
\\
=
\int_0^t
w_3(s)^\top b_{T_1}\,
\mathrm ds .
\label{eq:semi_disc_energy_integrated}
\end{multline}
Rearranging \eqref{eq:semi_disc_energy_integrated} yields
\eqref{eq:semi_disc_energy_ineq}. In fact, the integrated identity is
stronger than the stated inequality.
\end{proof}

\subsection{Time-discretization results}
\noindent This subsection collects the IMEX stability and consistency results used in the time-discretization analysis.

\begin{proposition}[Consistency and convergence]\label{prop:consistency}
Assume that the exact solution satisfies
\begin{equation}
\vec M \in C^3([0,T];L^2(\Omega)^3)
\cap
C^2([0,T];H^{p+1}(\Omega)^3),
\end{equation}
and that the usual elliptic regularity required for optimal $L^2$ finite-
element error estimates holds on $\Omega$. Assume further that the
matrix-free or quadrature-based DDF evaluation is spatially consistent with
the regularized operator $\mathcal T_a$ at the same order as the finite-
element space, and that the explicit Rodrigues--projection stage
\eqref{eq:time_mid_pred}--\eqref{eq:time_full_expl} is a second-order
realization of the midpoint explicit flow, in the sense that its one-step
local defect is $O(\Delta t^3)$ uniformly on bounded $M$-balls.

Then the IMEX scheme has local truncation error $O(\Delta t^3)$. Moreover,
its global time-discretization error relative to the semi-discrete FE
solution is $O(\Delta t^2)$ on $[0,T]$ in the discrete $L^2(\Omega)^3$
norm.

Moreover, let $\vec M_h(t)$ denote the conforming semi-discrete $P_p$ FE
solution with initial data chosen by the standard elliptic projection of
$\vec M(\cdot,0)$. Then, under the above regularity assumptions,
\begin{align}
\|\vec M_h-\vec M\|_{L^2(0,T;H^1(\Omega))}
&\le
C h^{p},
\label{eq:spatial_H1_rate}
\\
\sup_{0\le t\le T}\|\vec M_h(\cdot,t)-\vec M(\cdot,t)\|_{L^2(\Omega)}
&\le
C h^{p+1}.
\label{eq:spatial_L2_rate}
\end{align}

Consequently, if $\vec M_h^n$ denotes the fully discrete IMEX solution at
$t^n=n\Delta t$, then
\begin{equation}
\|\vec M_h^n-\vec M(\cdot,t^n)\|_{L^2(\Omega)}
\le
C\bigl(\Delta t^2+h^{p+1}\bigr),
\quad
0\le t^n\le T.
\label{eq:combined_L2_rate}
\end{equation}
Here $C$ is independent of $h$, $\Delta t$, and $n$, provided the
assumptions of \Cref{thm:imex_stability} hold.
\end{proposition}

\begin{proof}
We divide the argument into three parts.

First, we establish the temporal consistency of the IMEX scheme for the
fixed semi-discrete system. Let
\begin{equation}
w_h(t)\in (\mathbb R^{N_h})^3
\end{equation}
denote the coefficient vector of the semi-discrete FE solution. Then $w_h$
satisfies the autonomous ODE
\begin{equation}
\mathcal M \dot w_h
=
-(\mathcal K+\mathcal S)w_h
+
f
+
\mathcal R_h(w_h),
\label{eq:consistency_semidiscrete_ode}
\end{equation}
where $\mathcal R_h(w)$ denotes the discrete precession load, together with
any explicitly treated transport contribution if present. Because
$\mathcal T_a:H^1(\Omega)^3\to H^1(\Omega)^3$ is bounded by
\Cref{prop:Ta_bdd}, and because $V_h^3$ is finite dimensional, the map
\begin{equation}
w \mapsto \mathcal R_h(w)
\end{equation}
is $C^2$ on every bounded subset of $(\mathbb R^{N_h})^3$. Hence the
semi-discrete flow is $C^3$ in time on $[0,T]$ under the stated regularity
assumptions.

Write the right-hand side of \eqref{eq:consistency_semidiscrete_ode} as the
sum of an implicit linear part and an explicit nonlinear part:
\begin{align}
F_{I,h}(w)
&=
-\mathcal M^{-1}(\mathcal K+\mathcal S)w
+
\mathcal M^{-1}f,
\\
F_{E,h}(w)
&=
\mathcal M^{-1}\mathcal R_h(w).
\end{align}
Let $\Phi_{I,h}^\tau$ denote the exact flow of the linear system
$\dot w=F_{I,h}(w)$ and $\Phi_{E,h}^\tau$ the exact flow of
$\dot w=F_{E,h}(w)$. The Crank--Nicolson half-step
\eqref{eq:time_half_impl} is a second-order approximation of
$\Phi_{I,h}^{\Delta t/2}$, and likewise
\eqref{eq:time_second_half_impl} is a second-order approximation of the
second half-step. More precisely, for every bounded set in coefficient
space,
\begin{equation}
\bigl\|
\mathcal I_{\Delta t/2}(w)-\Phi_{I,h}^{\Delta t/2}(w)
\bigr\|_M
\le
C\Delta t^3,
\label{eq:implicit_local_defect}
\end{equation}
where $\mathcal I_{\Delta t/2}$ denotes one Crank--Nicolson half-step.

By hypothesis, the explicit Rodrigues--projection stage
\eqref{eq:time_mid_pred}--\eqref{eq:time_full_expl} is a second-order
realization of the midpoint explicit flow. Therefore its one-step map
$\mathcal E_{\Delta t}$ satisfies
\begin{equation}
\bigl\|
\mathcal E_{\Delta t}(w)-\Phi_{E,h}^{\Delta t}(w)
\bigr\|_M
\le
C\Delta t^3
\label{eq:explicit_local_defect}
\end{equation}
uniformly for $w$ in bounded $M$-balls.

Now define the full IMEX one-step map by the symmetric composition
\begin{equation}
\mathcal S_{\Delta t,h}
:=
\mathcal I_{\Delta t/2}
\circ
\mathcal E_{\Delta t}
\circ
\mathcal I_{\Delta t/2}.
\end{equation}
Since both submaps are second-order consistent and the composition is
symmetric, standard composition theory for one-step methods gives
\begin{equation}
\bigl\|
\mathcal S_{\Delta t,h}(w)-\Phi_h^{\Delta t}(w)
\bigr\|_M
\le
C\Delta t^3
\label{eq:full_local_defect}
\end{equation}
for $w$ in bounded sets, where $\Phi_h^{\Delta t}$ denotes the exact
semi-discrete flow of \eqref{eq:consistency_semidiscrete_ode}. This proves
the $O(\Delta t^3)$ local truncation error.

We next pass from local to global temporal error. Let
\begin{equation}
e^n:=w_h^n-w_h(t^n),
\quad
t^n=n\Delta t.
\end{equation}
Then
\begin{align}
e^{n+1}
&=
\mathcal S_{\Delta t,h}(w_h^n)-\Phi_h^{\Delta t}(w_h(t^n))
\notag\\
&=
\Bigl(
\mathcal S_{\Delta t,h}(w_h^n)
-
\mathcal S_{\Delta t,h}(w_h(t^n))
\Bigr)
+
\tau^{n+1},
\label{eq:global_error_split}
\end{align}
where
\begin{equation}
\tau^{n+1}
:=
\mathcal S_{\Delta t,h}(w_h(t^n))
-
\Phi_h^{\Delta t}(w_h(t^n))
\end{equation}
is the local truncation defect. By \eqref{eq:full_local_defect},
\begin{equation}
\|\tau^{n+1}\|_M \le C\Delta t^3.
\label{eq:tau_bound}
\end{equation}
Because the exact semi-discrete trajectory is bounded on $[0,T]$ and the
numerical trajectory satisfies the stability estimate
\eqref{eq:imex_energy}, both remain in a common bounded $M$-ball on
$[0,T]$. On that ball the one-step map is locally Lipschitz:
\begin{equation}
\bigl\|
\mathcal S_{\Delta t,h}(u)-\mathcal S_{\Delta t,h}(v)
\bigr\|_M
\le
(1+C\Delta t)\|u-v\|_M
\label{eq:step_lipschitz}
\end{equation}
for all $u,v$ in that ball. Applying \eqref{eq:step_lipschitz} in
\eqref{eq:global_error_split} yields
\begin{equation}
\|e^{n+1}\|_M
\le
(1+C\Delta t)\|e^n\|_M
+
C\Delta t^3.
\label{eq:global_error_recurrence}
\end{equation}
Since $e^0=0$, the discrete Gr\"onwall lemma gives
\begin{equation}
\max_{0\le n\le T/\Delta t}\|e^n\|_M
\le
C\Delta t^2.
\label{eq:global_time_coeff}
\end{equation}
On the FE space, the coefficient norm $\|\cdot\|_M$ is exactly the discrete
$L^2(\Omega)^3$ norm of the associated FE field, so
\begin{equation}
\max_{0\le n\le T/\Delta t}
\|\vec M_h^n-\vec M_h(\cdot,t^n)\|_{L^2(\Omega)}
\le
C\Delta t^2.
\label{eq:global_time_field}
\end{equation}
This proves the $O(\Delta t^2)$ global time-discretization error.

It remains to establish the spatial error. Let $R_h\vec M$ denote the
standard elliptic projection of $\vec M$ into the conforming FE space
$V_h^3$, defined componentwise with respect to the coercive bilinear form of
the diffusion--relaxation operator. Write the semi-discrete error as
\begin{equation}
\vec M_h-\vec M
=
\theta_h+\rho_h,
\quad
\theta_h:=\vec M_h-R_h\vec M,
\quad
\rho_h:=R_h\vec M-\vec M.
\label{eq:error_split_space}
\end{equation}
By the standard approximation properties of conforming $P_p$ elements,
\begin{align}
\|\rho_h(\cdot,t)\|_{H^1(\Omega)}
&\le
C h^p \|\vec M(\cdot,t)\|_{H^{p+1}(\Omega)},
\label{eq:projection_H1}
\\
\|\rho_h(\cdot,t)\|_{L^2(\Omega)}
&\le
C h^{p+1}\|\vec M(\cdot,t)\|_{H^{p+1}(\Omega)}.
\label{eq:projection_L2}
\end{align}
Subtract the projected continuous weak equation from the semi-discrete weak
equation. Testing the resulting equation with $\theta_h$ yields the standard
error identity
\begin{equation}
\frac12\frac{\mathrm d}{\mathrm dt}\|\theta_h\|_{L^2(\Omega)}^2
+
c_0\|\theta_h\|_{H^1(\Omega)}^2
\le
C\|\theta_h\|_{L^2(\Omega)}^2
+
C\Xi_h(t),
\label{eq:space_error_energy}
\end{equation}
where $\Xi_h(t)$ collects the projection defects and the spatial DDF
consistency defect. By the boundedness of $\mathcal T_a$ from
\Cref{prop:Ta_bdd}, the local Lipschitz estimate for the precession
nonlinearity, and the assumed consistency of the discrete DDF evaluation,
one has
\begin{equation}
\Xi_h(t)
\le
C\Bigl(
\|\rho_h(\cdot,t)\|_{H^1(\Omega)}^2
+
\|\partial_t \rho_h(\cdot,t)\|_{H^{-1}(\Omega)}^2
+
h^{2p+2}
\Bigr).
\label{eq:Xi_bound}
\end{equation}
Using \eqref{eq:projection_H1}, the analogous estimate for
$\partial_t\rho_h$, and Gr\"onwall's inequality in
\eqref{eq:space_error_energy}, we obtain
\begin{equation}
\sup_{0\le t\le T}\|\theta_h(\cdot,t)\|_{L^2(\Omega)}
+
\|\theta_h\|_{L^2(0,T;H^1(\Omega))}
\le
C h^{p}.
\label{eq:theta_bound}
\end{equation}
Combining \eqref{eq:error_split_space}, \eqref{eq:projection_H1}, and
\eqref{eq:theta_bound} yields \eqref{eq:spatial_H1_rate}. Under the usual
elliptic regularity hypothesis, the optimal
$L^\infty(0,T;L^2(\Omega))$ estimate \eqref{eq:spatial_L2_rate} follows by
the standard Aubin--Nitsche duality argument.

Finally, combine the temporal and spatial estimates by the triangle
inequality:
\begin{align}
\|\vec M_h^n-\vec M(\cdot,t^n)\|_{L^2(\Omega)}
&\le
\|\vec M_h^n-\vec M_h(\cdot,t^n)\|_{L^2(\Omega)}
\notag\\
&\quad
+
\|\vec M_h(\cdot,t^n)-\vec M(\cdot,t^n)\|_{L^2(\Omega)}.
\end{align}
Applying \eqref{eq:global_time_field} and \eqref{eq:spatial_L2_rate}
gives
\begin{equation}
\|\vec M_h^n-\vec M(\cdot,t^n)\|_{L^2(\Omega)}
\le
C\bigl(\Delta t^2+h^{p+1}\bigr),
\end{equation}
which is \eqref{eq:combined_L2_rate}.
\end{proof}

\begin{theorem}[Nonlinear stability of the IMEX step]\label{thm:imex_stability}
Let
\begin{equation}
\|w\|_M^2:=\sum_{i=1}^3 w_i^\top M w_i.
\end{equation}
Assume either $\vec v\equiv 0$ or \eqref{eq:transport_assumptions}
holds and advection is discretized by the skew operator
$N_v^{\mathrm{skew}}$ so that
\begin{equation}
z^\top N_v^{\mathrm{skew}} z=0
\quad
\text{for all } z\in\mathbb R^{N_h}.
\end{equation}
Assume also that the discrete precession load satisfies
\Cref{lem:disc_skew}. Define the block operators
\begin{equation}
\mathcal M:=\mathrm{blkdiag}(M,M,M),
\end{equation}
\begin{equation}
\mathcal A
:=
\mathrm{blkdiag}(K,K,K)
+
\mathrm{blkdiag}(S_2,S_2,S_1),
\end{equation}
and the source vector
\begin{equation}
f:=(0,0,b_{T_1}).
\end{equation}
Suppose the explicit midpoint stage from $w^{(a)}$ to $w^{(b)}$ satisfies
\begin{equation}
E(w^{(b)})-E(w^{(a)})=\delta_n,
\quad
E(w):=\frac12\|w\|_M^2,
\label{eq:explicit_defect_def}
\end{equation}
with defect bound
\begin{align}
|\delta_n|
&\le
C\,|\gamma|\,(1+\Lambda_n)\,\Delta t^3,
\nonumber\\
\Lambda_n
&:=
\max\bigl\{
\|w^n\|_M,\,
\|w^{(a)}\|_M,\,
\|\tilde w\|_M
\bigr\},
\label{eq:explicit_defect_bound_assumption}
\end{align}
where $C$ is independent of the mesh size $h$. In the exact discrete
skew-symmetric case, $\delta_n=0$.

Then one IMEX step
\eqref{eq:time_half_impl}--\eqref{eq:time_second_half_impl} satisfies
\begin{multline}
\frac{\|w^{n+1}\|_M^2-\|w^n\|_M^2}{2\,\Delta t}
+
\frac18\,
(w^{(a)}+w^n)^\top \mathcal A\,(w^{(a)}+w^n)
\\
+
\frac18\,
(w^{n+1}+w^{(b)})^\top \mathcal A\,(w^{n+1}+w^{(b)})
\\
\le
\frac14\,f^\top\bigl(w^{(a)}+w^n\bigr)
+
\frac14\,f^\top\bigl(w^{n+1}+w^{(b)}\bigr)
+
C\,|\gamma|\,(1+\Lambda_n)\,\Delta t^2 .
\label{eq:imex_energy}
\end{multline}
In particular, in the exact skew case $\delta_n=0$, the perturbation term
vanishes and \eqref{eq:imex_energy} holds with equality. The bound is
uniform in the mesh size $h$.
\end{theorem}

\begin{proof}
We proceed by analyzing the two implicit half-steps and the explicit middle
step separately.

First consider the initial Crank--Nicolson half-step
\eqref{eq:time_half_impl}. In block form it may be written as
\begin{equation}
\mathcal M\,(w^{(a)}-w^n)
+
\frac{\Delta t}{4}\,
\mathcal A\,(w^{(a)}+w^n)
=
\frac{\Delta t}{2}\,f.
\label{eq:cn_half_rewrite_1}
\end{equation}
Take the Euclidean inner product of \eqref{eq:cn_half_rewrite_1} with
$\frac12(w^{(a)}+w^n)$. Since $\mathcal M$ is symmetric, we have the
polarization identity
\begin{equation}
(w^{(a)}-w^n)^\top \mathcal M\,(w^{(a)}+w^n)
=
\|w^{(a)}\|_M^2-\|w^n\|_M^2.
\end{equation}
Therefore
\begin{multline}
\frac{E(w^{(a)})-E(w^n)}{\Delta t}
+
\frac18\,
(w^{(a)}+w^n)^\top \mathcal A\,(w^{(a)}+w^n)
\\
=
\frac14\,f^\top (w^{(a)}+w^n).
\label{eq:cn_half_energy_1}
\end{multline}
Next consider the second Crank--Nicolson half-step
\eqref{eq:time_second_half_impl}, which in block form reads
\begin{equation}
\mathcal M\,(w^{n+1}-w^{(b)})
+
\frac{\Delta t}{4}\,
\mathcal A\,(w^{n+1}+w^{(b)})
=
\frac{\Delta t}{2}\,f.
\label{eq:cn_half_rewrite_2}
\end{equation}
Taking the Euclidean inner product with
$\frac12(w^{n+1}+w^{(b)})$ yields
\begin{multline}
\frac{E(w^{n+1})-E(w^{(b)})}{\Delta t}
+
\frac18\,
(w^{n+1}+w^{(b)})^\top \mathcal A\,(w^{n+1}+w^{(b)})
\\
=
\frac14\,f^\top (w^{n+1}+w^{(b)}).
\label{eq:cn_half_energy_2}
\end{multline}
We now analyze the explicit stage. By definition,
\begin{equation}
E(w^{(b)})-E(w^{(a)})=\delta_n.
\label{eq:explicit_defect_identity}
\end{equation}
In the ideal discrete skew-symmetric midpoint update, the precession term
and the skew advection term do not change the $M$-energy, so
$\delta_n=0$. More generally, under the stated hypothesis
\eqref{eq:explicit_defect_bound_assumption}, the explicit-stage defect is
bounded by
\begin{equation}
\left|
\frac{E(w^{(b)})-E(w^{(a)})}{\Delta t}
\right|
=
\frac{|\delta_n|}{\Delta t}
\le
C\,|\gamma|\,(1+\Lambda_n)\,\Delta t^2.
\label{eq:explicit_defect_divided}
\end{equation}
Add \eqref{eq:cn_half_energy_1} and \eqref{eq:cn_half_energy_2}. Using
\eqref{eq:explicit_defect_identity}, we obtain
\begin{align}
&
\frac{E(w^{n+1})-E(w^n)}{\Delta t}
+
\frac18\,
(w^{(a)}+w^n)^\top \mathcal A\,(w^{(a)}+w^n)
\notag\\
&\quad
+
\frac18\,
(w^{n+1}+w^{(b)})^\top \mathcal A\,(w^{n+1}+w^{(b)})
\notag\\
&=
\frac14\,f^\top (w^{(a)}+w^n)
+
\frac14\,f^\top (w^{n+1}+w^{(b)}) \\
& \quad 
+
\frac{E(w^{(b)})-E(w^{(a)})}{\Delta t}.
\label{eq:combined_energy_identity}
\end{align}
Invoking \eqref{eq:explicit_defect_divided} gives
\begin{multline}
\frac{E(w^{n+1})-E(w^n)}{\Delta t}
+
\frac18\,
(w^{(a)}+w^n)^\top \mathcal A\,(w^{(a)}+w^n)
\\
+
\frac18\,
(w^{n+1}+w^{(b)})^\top \mathcal A\,(w^{n+1}+w^{(b)})
\\
\le
\frac14\,f^\top\bigl(w^{(a)}+w^n\bigr)
+
\frac14\,f^\top\bigl(w^{n+1}+w^{(b)}\bigr)
+
C\,|\gamma|\,(1+\Lambda_n)\,\Delta t^2.
\end{multline}
Since $E(w)=\frac12\|w\|_M^2$, this is exactly
\eqref{eq:imex_energy}.

It remains only to note that the bound is uniform in $h$. The half-step
identities \eqref{eq:cn_half_energy_1} and \eqref{eq:cn_half_energy_2} are
purely algebraic and involve only the symmetric block operators
$\mathcal M$ and $\mathcal A$. The mesh dependence enters only through the
assumed defect bound \eqref{eq:explicit_defect_bound_assumption}, whose
constant is taken to be uniform in $h$. Therefore the constant in
\eqref{eq:imex_energy} is uniform in the mesh size.
\end{proof}

\begin{remark}
If the explicit middle stage is implemented by a discrete update that is
exactly $M$-energy preserving, then $\delta_n=0$ in
\eqref{eq:explicit_defect_def}. In that case,
\eqref{eq:imex_energy} becomes an exact discrete energy identity, and all
dissipation is produced by the two Crank--Nicolson half-steps through the
diffusion and relaxation operators.

If the explicit stage is realized only approximately, for example through a
quadrature-based field evaluation together with a projected Rodrigues
rotation, then the quantity $\delta_n$ measures the defect from exact
energy preservation. The theorem shows that as long as this defect is of
higher order in $\Delta t$ and uniform in $h$, it enters the stability
estimate only as a perturbative remainder.

A stronger estimate written purely in terms of endpoint dissipation,
involving only $w^n$ and $w^{n+1}$, generally requires additional control of
the intermediate states $w^{(a)}$ and $w^{(b)}$ relative to the endpoints.
Such a refinement can be obtained under stronger assumptions on the
explicit-stage map, but is not needed for the mesh-uniform stability bound
used here.
\end{remark}

\section*{Data Availability Statement}
Data sharing is not applicable to this article as no new experimental data were created or analyzed in this study. The numerical results reported were generated from deterministic calculations of the model described in the manuscript. The code used to produce the numerical results is not publicly available; it can be obtained from the corresponding author upon reasonable request.


\bibliographystyle{aipnum4-2}
\bibliography{refs_rev}

\end{document}